\documentclass[11pt]{article}
\usepackage{amsfonts}
\usepackage{amssymb}
\usepackage{color, graphics, graphicx}
\usepackage{amsmath}
\usepackage{mathrsfs}
\usepackage{caption,subcaption}
\usepackage{geometry}
\usepackage{fontenc}
\usepackage{lscape}
\usepackage{sectsty}
\usepackage{bm}
\usepackage{amsfonts}
\usepackage{amssymb}
\usepackage{amsmath}
\usepackage{natbib}
\usepackage{fontenc}
\usepackage{epsfig}
\usepackage{latexsym,graphics,graphicx,amsbsy,amsmath,amsthm}
\usepackage[table]{xcolor}
\usepackage{setspace}
\usepackage{soul}
\usepackage{appendix}
\usepackage{algorithm}
\usepackage[noend]{algpseudocode}
\usepackage{graphicx}
\usepackage{pgfplots}
\usepackage{flafter}
\usepackage{tikz}
\newlength\figureheight
\newlength\figurewidth
\pgfplotsset{compat=newest}
\pgfplotsset{plot coordinates/math parser=false}
\usetikzlibrary{shapes, backgrounds, arrows, positioning}
\usetikzlibrary{plotmarks}
\newtheorem{theorem}{Theorem}
\newtheorem{Assumption}{Assumption}

\newtheorem{corollary}{Corollary}

\newtheorem{definition}{Definition}

\newtheorem{lemma}{Lemma}

\newtheorem{remark}{Remark}

\RequirePackage[OT1]{fontenc}
\DeclareMathOperator*{\argmin}{arg\,min}
\DeclareMathOperator*{\argmax}{arg\,max}

\RequirePackage[OT1]{fontenc}
\RequirePackage[colorlinks]{hyperref} 
\hypersetup{
	colorlinks=red,%
	citecolor=blue,%
	filecolor=blue,%
	linkcolor=blue,%
	urlcolor=blue
}

\newcommand{\bk}[1]{\textcolor{red}{#1}}

\begin{document}
	\title{\textsc{\Large Indirect Inference for Locally Stationary Models }}
	\author{David T. Frazier\thanks{
			Department of Econometrics and Business Statistics, Monash University, PO
			Box 11E, Clayton Campus, VIC 3800, Australia; e-mail: david.frazier@monash.edu} 
		\qquad Bonsoo Koo\thanks{\noindent Corresponding author.
			Department of Econometrics and Business Statistics, Monash University, PO
			Box 11E, Clayton Campus, VIC 3800, Australia; e-mail: bonsoo.koo@monash.edu} \\
		Monash University}
	\title{\textsc{\Large Indirect Inference for Locally Stationary Models }}
	\date{}
	\maketitle
	
	\begin{abstract}
		We propose the use of indirect inference estimation to conduct inference in complex locally stationary models. We develop a local indirect inference algorithm and establish the asymptotic properties of the proposed estimator. Due to the nonparametric nature of locally stationary models, the resulting indirect inference estimator exhibits nonparametric rates of convergence. We validate our methodology with simulation studies in the confines of a locally stationary moving average model and a new locally stationary multiplicative stochastic volatility model. Using this indirect inference methodology and the new locally stationary volatility model, we obtain evidence of non-linear, time-varying volatility trends for monthly returns on several Fama-French portfolios.     
	\end{abstract}
	
	
	\textit{Journal of Economic Literature Classification}:\ C13, C14, C22
	\maketitle
	\newpage
	\section{Introduction}
	
	Time-varying economic and financial variables, and relationships thereof, are stable features in applied econometrics. Notable examples include asset pricing models with time-varying features \citep{ghysels1998stable, wang2003asset} and trending macroeconomic models \citep{stock1998median, phillips2001trending}. While classical analyses of time series are built on the assumption of stationarity, data studied in finance and economics often exhibit nonstationary features. 
	
	Many different schools of modeling and estimation methods are used to accommodate the nonstationary behavior of observed time series data. In particular, statistical tools developed for locally stationary processes provide a convenient means of conducting analyses of trending economic and financial models. Heuristically, local stationarity implies that a process behaves in a stationary manner (at least) in the vicinity of a given time point but could be nonstationary over the entire time horizon. For certain widely-studied time series models, slowly time-varying parameters ensure local stationarity under some regularity conditions; for instance, see \citet{dahlhaus1996kullback} and \citet{D97} (AR(1)), \citet{DSR06} (ARCH($\infty$)), \citet{dahlhaus2009empirical} (MA($\infty$)), \citet{KL12} (Diffusion processes) and \citet{koo2015let} (GARCH(1,1) with a time-varying unconditional variance) among many other classes of locally stationary processes. 
	
	While many classes of well-known time series models can be generalized to locally stationary processes, it is worth noting that estimation and inference procedures developed in one class of locally stationary processes often cannot be applied to a different class of locally stationary processes. In particular, many estimation methods for locally stationary processes are composed of estimation approaches that primarily focus on local regression with closed-form estimators, local maximum likelihood estimation (MLE) with a closed-form likelihood function (in the time domain) and spectral density approach (in the frequency domain), all of which could be intractable or simply difficult to implement for various locally stationary extensions of commonly used structural econometric models; we refer to \citet{vogt2012nonparametric}, \citet{DSR06} and \citet{dahlhaus2009empirical}, for examples. As such, model specifications compatible with the above statistical methods are rather limited and cannot be used for estimation and inference in more complicated locally stationary models, such as, for instance, models with latent variables or unobservable factors. 
	
	More importantly, structural models of economic and financial relationships commonly rely on the use of latent variables to represent information that is unavailable to the econometrician. This modeling approach implies, almost by definition, that simple (closed-form) representations for the conditional distributions of the endogenous variables are unavailable, with simple straightforward estimation methods often infeasible as a consequence. In such cases, if we were to extend common locally-stationary models to include the latent variables that are necessary to structurally model phenomena found in economics and finance, this would render the existing estimation methods used for such models infeasible. For instance, this situation arises in state-space models if either the measurement or state transition densities do not have closed forms, as in the case of stochastic volatility models. A secondary example is the fact that estimation of univariate locally stationary diffusion models cannot be straightforwardly extended to versions of these models with stochastic volatility.
	
	To circumvent the above issue, and to help proliferate the use of locally stationary models and methods in econometrics and finance, we propose a novel nonparametric indirect inference (hereafter, II) method to estimate locally stationary processes. Instead of estimating complex \textit{structural} locally stationary models directly, we indirectly obtain our estimator by targeting consistent estimators of simpler \textit{auxiliary} models, and use these consistent estimates to conduct inference on the structural parameters. See, \citet{smith1993estimating}, \citet{GMR93} and \citet{gourieroux1996simulation} for discussion of indirect inference in parametric models. 
	
	To illustrate the main idea behind our nonparametric II approach for locally stationary processes, we consider the following motivating example. Suppose that the true data generating process evolves according to
	\begin{equation}
	Y_{t,T}=\sqrt{\xi(t/T)}\exp\left(h_t/2\right)\varepsilon_{t},\;\text{where }h_{t}=\omega+\delta h_{t-1}+\sigma v_t,\;(\varepsilon_t,v_t)'\sim \mathcal{N}\left(0,\begin{bmatrix}1&0\\0&1 \end{bmatrix}\right),
	\label{eq: ex1_str}
	\end{equation}
	where $\xi(t/T)>0,$ for all $t\leq T$. This locally stationary multiplicative stochastic volatility (LS-SV) model decomposes volatility into a short-term, latent volatility process, $h_t$, and a slowly time-varying component, captured by $\xi(\cdot)$, and can capture a wide range of volatility behaviors. The above model allows for non-stationary, but slowly changing, volatility dynamics, which may result from the transitory nature of the business cycle.

	Suppose that we wish to estimate and conduct inference on the unknown volatility function $\xi(\cdot)$ in \eqref{eq: ex1_str}. While (G)ARCH-based versions of the locally stationary volatility model have been analyzed by several researchers (see, e.g., \citealp{DSR06}, \citealp{ER08}, \citealp{fryzlewicz2008normalized}, and \citealp{koo2015let}), since the latent volatility process, $h_t$, pollutes the observed data, $Y_{t,T}$, it is not entirely clear how to estimate parameters in \eqref{eq: ex1_str}. Indeed, largely due to this fact, locally stationary volatility models have not been previously explored in the literature, even though their stationary counterparts form the backbone of many empirical studies in finance and financial econometrics. 
	
	In this paper, we generalize the II approach of \citet{GMR93} to present a convenient estimator for unknown functions in locally stationary models, such as the LS-SV model. This approach to II estimation relies on a locally stationary auxiliary model that can be easily estimated using the observed data and that captures the underlying features of interest in the structural model. For example, in the context of the LS-SV model, a reasonable auxiliary model would be the locally stationary GARCH model:  
	\begin{flalign}
	Y_{t,T} =\sqrt{\rho(t/T)} \sigma_{t} z_{t}, \text{ where } \sigma_{t+1}^{2} &=\alpha_0+\alpha_1 z_{t}^{2}+\beta\sigma_{t}^{2} ,	\label{eq: ex1_aux}
	\end{flalign} where $\rho(t/T)>0$ for all $t\leq T$, and where $z_t$ is an error process.
	
	The remainder of this paper further develops the ideas behind this estimation method in the context of a general locally stationary model and establishes the asymptotic properties of the proposed estimation procedure under regularity conditions. To establish the asymptotic properties of these II estimators, we must first develop conditions that guarantee locally stationary models admit consistent estimators of their corresponding limit values. This is itself a novel result since the vast majority of research into locally stationary models has focused on estimators defined by relatively simple criterion functions, and all under the auspices of correct model specification. Indeed, \cite{kristensen2019local} is the only other study of which the authors are aware that treats genuinely misspecified locally stationary models. These new results for locally stationary estimators of the auxiliary model enable us to deduce the asymptotic properties of our proposed II estimator for the structural model parameters. 	
	
	The estimation procedure proposed herein is demonstrated through two Monte Carlo examples, and an empirical application. The empirical application applies the LS-SV model to examine the volatility structure of several commonly analyzed Fama-French portfolios. We find that most of these portfolios display time-varying volatility patterns that broadly track the underlying (low-frequency) expansion and contractions of the United States economy. 
	
	The remainder of the paper is organized as follows. 
	Section \ref{BBBmodel} introduces the general model and the related framework. In Section \ref{BBBestimators}, we present our general approach and define the corresponding local II (L-II) estimators for a general locally stationary model. Section \ref{BBBasymptotics} develops asymptotic results that demonstrate the properties of this estimation procedure. Simulation results for a simple example of a locally stationary moving average model of order one are discussed in Section \ref{BBBsim}. In Section \ref{sec:lssv} we analyze the locally stationary stochastic volatility model. We consider a small Monte Carlo to demonstrate our estimation method, then apply this method to analyze the volatility behavior of Fama-French portfolio returns, where we find ample evidence for smoothly time-varying nonlinear volatility dynamics over the sample period. All proofs are relegated to Appendix \ref{app: proofs}.  The tables and figures associated with the application in Section \ref{sec:lssv} are given in Appendix \ref{app: tf}. The proof of Corollary 2 and additional details for the LS-SV model are provided in the supplementary appendix.

	Throughout this paper, the following notations are used. The symbol $\mathbb{R}$ denotes the real numbers, while $\mathbb{N}$ denotes the natural numbers. For $x\in\mathbb{R}^d$, we let $\|x\|$ denote the Euclidean norm, while $|\cdot|$ denotes the absolute value function, and for $\Omega$ a $d\times d$ positive-definite matrix, we let $\Vert x\Vert^2_{\Omega}:=x'\Omega x$ denote the weighted norm of $x$. For $g:\mathbb{R}^{d}\rightarrow \mathbb{R}$ denoting a given function, we let $\|g\|_\infty:=\sup_{x\in\mathbb{R}^d}|g(x)|$ denote the sup-norm. For an unknown parameter $\theta$, the subscript $0$ denotes the true value of $\theta$. The quantities $O_p(\cdot)$ and $o_p(\cdot)$ denote the usual big $O$ and little $o$ in probability.
	$C$ denotes a generic constant that can take different values in different places. 
	
	\section{The model}
	\label{BBBmodel}
	\subsection{Structural models}
	
	We assume the researcher is interested in conducting inference on a model in the class of locally stationary processes.
	\begin{definition}
		\label{def: ls}
		Let $\{Y_{t,T}\}_{t=1,...,T;T=1,2,...}$ denote a triangular array of observations. The process $\{Y_{t,T}\}$ is locally stationary if there exists a stationary process $\{y_{t/T,t}\}$ for each re-scaled time point $t/T\in[0,1]$, such that for all $T$, 
		\begin{equation*}
		\text{{P}}\left(\max_{1\leq t\leq T}\left|Y_{t,T}-y_{t/T,t}\right|\leq C_{T}T^{-1}\right)=1,
		\end{equation*}where $\{C_{T}\}$ is a measurable process satisfying, for some $\eta>0$, $\sup_T {E}\left(\left|C_T\right|^{\eta}\right)<\infty$. 
	\end{definition}
	The magnitude of $\eta$ captures the degree of approximation of $y_{t/T,t}$ to $Y_{t,T}$, which reflects the characteristics of the underlying processes of interest. The larger $\eta$, the better the approximation. We do not specify the magnitude of $\eta$ to maintain generality, which allows us to represent various types of processes, and instead allow $\eta$ to vary from model to model. See, for instance, \citet{DSR06} for ARCH($\infty$), \citet{KL12} for diffusion processes, \citet{vogt2012nonparametric} for AR processes and \citet{dahlhaus2009empirical} for MA processes among many other processes.
	
	We consider that the process $\{Y_{t,T}\}$ is generated from the following {locally stationary} structural model:
	\begin{equation}
	\begin{aligned}
	Y_{t,T}&=r(\epsilon_{t,T};\theta_{0}(t/T)),\\
	\epsilon_{t,T}&=\varphi(\nu_{t};\theta_{0}(t/T)),\label{struct1}
	\end{aligned}
	\end{equation}
	where both $r(\cdot)$ and $\varphi(\cdot)$ are real-valued functions that are known up to the unknown function $\theta_0$. The function of interest is $\theta_{0}\in\mathcal{H}_\theta$, where $(\mathcal{H}_\theta,\|\cdot\|)$ denotes a  normed vector space of function. The structural model, and $\theta_0$ satisfy the following regularity conditions.
	\begin{Assumption}\label{ass:1}(i) For a positive $\delta=o(1)$, and $u=t/T\in[\delta,1-\delta]$, the function $\theta_0(u)$ has uniformly bounded second-derivatives with respect to $u$. (ii) The functions $r(\cdot)$, $\varphi(\cdot)$, known up to $\theta_0(\cdot)$, are twice continuously differentiable with respect to $\theta$, with uniformly bounded second derivatives. (iii) The error term $\{\nu_{t}\}_{t\ge1}$ is a white noise process with known distribution.
	\end{Assumption}	

	The structural model in \eqref{struct1} is quite general and can accommodate many interesting processes, including models with complex time-varying features, such as time-varying autoregressive conditional heteroskedasticity (ARCH). In addition, the structural model in \eqref{struct1} can always be augmented with additional exogenous regressors at the cost of additional notation. Such regressors may be used, for instance, to capture some conditionally heteroskedastic features of the data. Critically for our purposes, under Assumption \ref{ass:1}, if $\theta_0(\cdot)$ were known, simulated realizations of $\{Y_{t,T}\}$ could easily be generated from the model in equation \eqref{struct1}.\footnote{We note here that Assumption \ref{ass:1}(iii) is standard in the II literature. Indeed, \cite{GMR93} argue that this is not a real assumption since the error term ``can always be considered as a function of a white noise with a known distribution and of a parameter which can be incorporated'' into the unknown parameters.}
	If the process in \eqref{struct1} is locally stationary, inference on $\theta_{0}(\cdot)$ can be carried out through an approximate structural model defining a stationary process indexed by $u\in\mathcal{U}$, where $\mathcal{U}$ denotes the domain of re-scaled time point $u=t/T$, i.e. $\mathcal{U}=[\delta,1-\delta]$ with a positive $\delta=o(1)$:
	\begin{equation}
	\begin{aligned}
	\label{struct2}
	y_{u,t}&=r(\epsilon_{u,t};\theta_{0}(u)),\\
	\epsilon_{u,t}&=\varphi(\nu_{t};\theta_{0}(u)).
	\end{aligned}
	\end{equation}
	\begin{lemma}\label{Lem:struct_ls}
		Suppose that $\{Y_{t,T}\}$ in \eqref{struct1} is locally stationary as in Definition \ref{def: ls}. Under Assumption \ref{ass:1}, as $T\rightarrow\infty$, the process $\{y_{u,t}\}$ in \eqref{struct2} is such that
		\begin{equation}
		| Y_{t,T}-y_{u,t} | = O_p\left(\left| {t}/{T}-u \right| + T^{-1}\right).
		\label{eq: approx}
		\end{equation}  		
	\end{lemma}
Lemma \ref{Lem:struct_ls} is consistent with Proposition 3.1 of \citet{dahlhaus2019towards}, and implies that in the neighborhood of a re-scaled time point $u = t/T$, the local behavior of $\{Y_{t,T}\}$ can be approximated by the behavior of $\{{y}_{u,t}\}$.  Consequently, statistical analysis on $\{Y_{t,T}\}$ can be based on a collection of locally stationary processes $\{y_{u,t}:u\in\mathcal{U}\}$. 
	
	Under local stationarity, we will demonstrate that estimation of the unknown (vector) function $\theta_{0}(\cdot)$ in \eqref{struct1} can proceed through a local version of II (L-II) conducted at the time points $u=t/T$. This approach relies on the fact that, for any $u\in\mathcal{U}$, $\theta_0(u)$ in \eqref{struct2} satisfies $\theta_{0}(u)\in\Theta\subset\mathbb{R}^{d_{\theta}}$; i.e., in the locally stationary structural model we view the function of interest as a map $\theta_{0}(\cdot):\mathcal{U}\mapsto\Theta$. The assumption that $\theta_{0}(\cdot)$ is our only parameter of interest is without loss of generality as we may always redefine $\theta_{0}(\cdot)$ to include those elements (time-varying or otherwise) of the distribution for the errors that are unknown. 
	This paper is particularly concerned with estimation and inference when the structural model, \eqref{struct1}, rules out direct estimation approaches developed in the existing literature, for instance, due to the presence of latent variables that make computation of the likelihood function intractable. 
	
    Consider that our goal is to estimate the unknown map $\theta_0:\mathcal{U}\mapsto\Theta$ at a given point $u\in\mathcal{U}$. Since $\theta_0(u)\in\Theta\subset\mathbb{R}^{d_\theta}$, we associate to this unknown function (evaluated at the point $u$) a vector $\theta\in\Theta$. Even if the vector $\theta$ can not be estimated by direct means, since $\{y_{u,t}\}_{}$ is stationary (at the fixed value $u$) we can easily simulate a realization of this series by replacing $\theta_0(u)$ in equation \eqref{struct1} by $\theta$. For fixed $u \in \mathcal{U}$ and some $\theta\in\Theta$, a simulated series $\{\tilde{y}_{u,t}(\theta)\}_{t\le T}$ can be generated according to
	\begin{equation}
	\begin{aligned}
	\label{struct3}
	\tilde{y}_{u,t}(\theta)&=r(\tilde{\epsilon}_{u,t}(\theta);{\theta}),\\
	\tilde{\epsilon}_{u,t}(\theta)&=\varphi(\tilde{\nu}_{t};{\theta}),
	\end{aligned}
	\end{equation}where $\tilde{\nu}_t$ denotes a simulated realization of the random variable $\nu_t$.\footnote{The use of slightly misspecified simulators in II is not uncommon, see, e.g., \citet{DGR2007}, \citet{AEA2013}, \citet{bruins2015}, and \citet{frazier2018indirect} for examples of misspecified simulators in the context of II estimation. In this sense, we follow the above papers in that the version of the structural model used to simulate data is a (locally) misspecified version of the true DGP.} Throughout the remainder, a tilde, $\tilde{}$, over a variable will denote that this variable is simulated and when no confusion will result we drop simulated series dependence on $\theta$, e.g., we take $\tilde{y}_{u,t}$ to mean $\tilde{y}_{u,t}(\theta)$. 
	
	Given the simulated series $\{\tilde{y}_{u,t}\}_{t\le T}$, II estimation of $\theta_0(u)$ can then proceed by minimizing the difference between statistics calculated from the observed data, $\{Y_{t,T}\}_{t\le T}$, and the simulated data, $\{\tilde{y}_{u,t}\}_{t\le T}$. Repeating this procedure at a collection of points $u_1,\dots,u_m$ would then yield an estimate of the unknown function $\theta_0(\cdot)$.	
	
	\subsection{Auxiliary models and direct estimation}
	To employ our L-II estimation method, we specify an auxiliary model defined by the unknown (vector) function $\rho(\cdot)\in\mathcal{H}_\rho$, with $(\mathcal{H}_\rho,\|\cdot\|)$ a vector space of functions, and where $\rho(\cdot):\mathcal{U}\mapsto\Gamma\subset\mathbb{R}^{d_\rho}$ with $d_{\rho}\geq d_{\theta}$. Similar to the structural function of interest, for any given $u\in\mathcal{U}$ we associate to the unknown function $\rho(u)$ a vector $\rho\in\Gamma\subset\mathbb{R}^{d_{\rho}}$. In general, we will only emphasize the parameters' dependence on the point $u$ when necessary. 
		
	Reflecting the features of the true structural model, the auxiliary model is chosen such that it allows for direct estimation of $\rho(\cdot)$.	We estimate $\rho(\cdot)$ at the point $u$, i.e., $\rho=\rho(u)$, by minimizing a local criterion function: for kernel function $K(\cdot)$ and bandwidth parameter $h$, define
	\begin{equation}
	M_{T}[\rho;u]:=\frac{1}{Th}\sum_{t=1}^{T}g[Y_{t,T};\rho]K\left(\frac{u-t/T}{h}\right),
	\label{eq:aux_obj}
	\end{equation}
	where $g(\cdot)$ is a known function whose properties we later specify. Note that, technically $M_{T}[\rho;u]$ depends on the array $\{Y_{t,T}\}_{t\le T}$, however, we obviate this dependence to keep notation as simple as possible. Given $M_{T}[\rho;u]$, an estimator for $\rho(u)$ can be defined as  
	\begin{flalign}
	\hat{\rho}(u;\theta_0(u))&:=\argmin_{\rho\in\Gamma}M_{T}[\rho;u].\label{eq: rho_hat}
	\end{flalign} 
	The explicit dependence of $\hat{\rho}(u;\theta_0(u))$ on $\theta_0(u)$ clarifies that the auxiliary estimator depends on the unknown $\theta_0(\cdot)$ at the point $u$. However, throughout the remainder, to simplify notation, we obviate this explicit dependence and simply define $\hat\rho(u):=\hat{\rho}(u;\theta_0(u))$.

	It is natural to consider an auxiliary model which allows for simple estimation of the auxiliary parameters. 	One such useful class of auxiliary models will be nonlinear regression models of the type considered in \citet{robinson1991time} and \citet{zhang2015time}: for $Z_{t,T}$  a triangular array of variables that are measurable at time $t$, and exogenous with respect to the error term $\eta_t$, the auxiliary model is given as $$Y_{t,T}=f\left(Z_{t,T};\rho(t/T)\right)+\eta_t,$$ where $f(\cdot)\in\boldsymbol{F}$ is known, up to the unknown $\rho(\cdot)$, and where $$\boldsymbol{F}:=\{f:|f(x,\rho_1)-f(x,\rho_2)|\le b(x)\|\rho_1-\rho_2\|_{\infty},\;\rho_1,\rho_2\in \mathcal{H}_\rho\}.$$ The set $\boldsymbol{F}$ restricts the form of $f(\cdot)$ to be locally (in $x$) Lipschitz (in $\rho$), with this restriction being satisfied by many regression functions. Under this specific nonlinear regression model, $M_T[\cdot;u]$ could be the local least squares criterion
	\begin{flalign*}
	M_{T}[\rho;u] &= \frac{1}{T h} \sum_{t=1}^{T}[Y_{t,T}-f(Z_{t,T},\rho)]^2K\left(\frac{u-t / T}{h}\right).
	\end{flalign*}
    
    While nonlinear regression models are a useful class of auxiliary models, we do not wish to restrict our analysis solely to this class, and we therefore allow the criterion function $M_T[\rho;u]$ to be general. However, to ensure our theory can easily accommodate this case, we further specialize the structure of the auxiliary criterion function $M_{T}$: For some kernel function, $K(\cdot)$ and a bandwidth parameter, $h$, some known function $f(\cdot)\in\mathbf{F}$ and observable exogenous variables $Z_{t,T}$, we assume that
	\begin{flalign}
	M_{T}[\rho;u]&:=\frac{1}{Th}\sum_{t=1}^{T}g[Y_{t,T};f(Z_{t,T},\rho)]K\left(\frac{u-t/T}{h}\right).
	\end{flalign}
	
	
	\subsection{Estimation of structural parameters}
	\label{BBBestimators}
	
	For $\{Y_{t,T}\}_{t\le T}$ denoting a set of observations from the locally stationary structural model \eqref{struct1}, satisfying Definition \ref{def: ls}, the auxiliary estimator $\hat\rho(u)$ in \eqref{eq: rho_hat} approximates the behavior of $\rho(\cdot)$ at the point $u$. Given $\hat\rho(u)$, an estimator of $\theta_0(u)$ can then be obtained by matching $\hat\rho(u)$ against a version that is calculated based on data simulated from the model under a given $\theta\in \Theta$, and a given $u\in\mathcal{U}$. However, we note that it is unclear in general how to simulate from the non-stationary structural model defined by \eqref{struct1}.
	
	Therefore, instead of attempting to simulate from the model \eqref{struct1}, we invoke the local stationarity of $\{Y_{t,T}\}$ and generate (simulated) realization from the stationary process $\{y_{u,t}:u\in\mathcal{U}\}$, defined by \eqref{struct2}, which approximates $\{Y_{t,T}\}$ in the sense of Definition \ref{def: ls}. Such an II estimation approach is by construction ``local'' in that all we can recover is $\theta_{0}(u)$. An estimate of $\theta_0(\cdot)$ can be obtained by repeatedly applying this local II (L-II) approach at a given set of time points $\{u_{i}\}_{i=1}^m$, where $\max_{i}\Delta u_i=O(T^{-1})$ and $\Delta u_i:=u_{i}-u_{i-1}$.

	More specifically, for some fixed $u_i\in\mathcal{U}$ and a corresponding candidate for  $\theta_0(u_i)$, say,  ${\theta}^{}={\theta}^{}(u_i)\in\Theta$, L-II then simulates data $\{\tilde{y}_{u_{i},t}^{}\}_{t\le T}$ from \eqref{struct2} using simulated errors $\{\tilde{\nu}_t\}_{t\le T}$. Given $\{\tilde{y}_{u_i,t}^{}\}_{t\le T}$, we estimate the auxiliary parameters using 
	\begin{flalign}
	\hat{\rho}(u_i;{\theta}^{})&:=\argmin_{\rho\in\Gamma}\frac{1}{T}\sum_{t=1}^{T}g[\tilde{y}_{u_i,t}^{};f(\tilde{z}_{u_i,t},{\rho})]\label{eq: rho_sim},
	\end{flalign}
	which corresponds to a simulated version of the local criterion function $M_T[\rho;u] $ in the vicinity of time point $u_i$. {Note that, similar to the notation we employ for $\hat{\rho}(u_i)$, the notation $\hat{\rho}(u_i;\theta^{})$ is an abbreviation for $\hat{\rho}(u_i;\theta^{}(u_i))$.} 
	
	Using $\hat{\rho}(u_i)$ and $\hat{\rho}(u_i;{\theta}^{})$, the L-II estimator of $\theta_{0}(u_i)$ can then be calculated, for positive-definite weighting matrix $\Omega$, as 
	\begin{flalign*}
	\hat{\theta}(u_i):=\arg\max_{{\theta}\in\Theta}{-}\|\hat{\rho}(u_i)-\hat{\rho}(u_i;{\theta})\|^2_{\Omega}.
	\end{flalign*}
	Using the same simulated errors $\{\tilde{\nu}_{t}\}_{t=1}^{T}$, we may repeat the above procedure for $\{u_{i}\}_{i=1}^{m}$, with $0<u_{1}<u_{2}<\cdots<u_{m}<1$, {and $\max_{i}\Delta u_i=O(T^{-1})$}, to obtain an estimator of $\theta_{0}(\cdot)$.

	The key feature of the above L-II procedure is that, due to the locally-stationary nature of \eqref{struct2}, the simulated series $\{\tilde{y}_{u_{i},t}\}_{t\le T}$ is stationary for each $u_{i}$, $i=1,...,m$. In this way, at each time point $u_{i}$, L-II matches a nonparametric estimator against a parametric estimator. As the following section illustrates, a consequence of this estimation approach is that the estimator $\hat{\theta}(\cdot)$ will inherit the asymptotic properties of the nonparametric estimator $\hat{\rho}(\cdot)$.

	\section{Asymptotic behavior of L-II}\label{BBBasymptotics}

	This section establishes the asymptotic properties of the L-II estimator. We establish the convergence (in probability) of $\hat{\theta}(\cdot)$ to $\theta_0(\cdot)$ and provide the asymptotic distribution of $\hat{\theta}(\cdot)$ under a fairly general setup.

	Before presenting the details, we introduce the limit quantities that will be needed for our results. 
	Consider the limit objective function and its minimizer corresponding to sample quantities, i.e. \eqref{eq:aux_obj} and \eqref{eq: rho_hat}, such that, for $u\in\mathcal{U} =[\delta,1-\delta]$ and a small, positive $\delta=o(1)$,
	\[
	\rho_0(u;\theta_0(u)):=\argmin_{\rho\in \Gamma}\mathbb{M}_0[\rho;u],\text{ where } \mathbb{M}_0[\rho; u]:=\lim_{T\rightarrow \infty}E M_T[\rho;u].
	\]	
	When no confusion will result, we denote $\rho_0(u;\theta_0(u))$ by $\rho_0(u)$. The value $\rho_{0}(u)$ is the minimizer of the limit map $\rho \mapsto\mathbb{M}_0[\rho;u]$ and depends on the features of the true distribution and the true value of the unknown function, $\theta_0(\cdot)$, in the structural model.
	
Likewise, we require that the simulated auxiliary estimator has a well-defined probability limit. Recalling the stationary nature of the simulated data, $\tilde{y}_{u,t}$, such a requirement boils down to standard results for the consistency of quasi-maximum likelihood estimators for the pseudo-true value; see, e.g., \citet{white1982maximum} and \citet{white1996estimation}. The simulated counterpart to the pseudo-true parameter $\rho_0(u)$ is the map $\theta\mapsto\rho_0(u;\theta)$, which we define as
	\begin{equation*}
	{\rho}^{}_0(u;\theta):=\argmin_{\rho\in \Gamma}\tilde{\mathbb{M}}_{0}[\rho;u],\text{ where }\tilde{\mathbb{M}}_{0}[\rho;u]:= \lim_{T\rightarrow\infty}\frac{1}{T}\sum_{t=1}^{T}E g_{}[\tilde{y}_{u,t};f(\tilde{z}_{u,t};\rho)], 
	\end{equation*}and where we remind the reader that we have suppressed the dependence of the simulated series $\tilde{y}_{u,t}$ on $\theta$ for notational simplicity.

\subsection{Consistency}

To demonstrate the asymptotic properties of our proposed L-II approach, we employ the following regularity conditions. 
\begin{Assumption}
	\label{ass:uniform} (i) $\left\{ (Y_{t,T},Z_{t,T});t=1,...,T;T=1,2,...\right\}$ are triangular arrays of locally stationary processes satisfying Definition \ref{def: ls} and are $\phi$-mixing with its mixing coefficients $\phi (k)$ such that for all integers $%
	0<t<\infty$ and $k> 0$,%
	\begin{equation*}
	\phi (k):=\sup_{-T\leq t\leq T}\sup_{A\in \mathcal{F}_{-\infty }^{T,t},B\in
		\mathcal{F}_{T,t+k}^{\infty },P(A)>0}\left\vert P\left(B|A\right)
	-P\left( B\right) \right\vert,
	\end{equation*}%
	where $1\ge\phi(0)\ge \phi(1) \ge ...$ and $\mathcal{F}_{-\infty }^{T,t}$ and $	\mathcal{F}_{T,t+k}^{\infty }$ are $\sigma$-fields generated by $\{(Y_{i,T},Z_{i,T});i\le t\}$ and $\{(Y_{i,T},Z_{i,T});i\ge t+k\}$ respectively. The mixing coefficients $\phi (k)$ converge to zero as $k\rightarrow \infty$, and are such that, for some sequence $m_T$, with $1\le m_T\le T$,
	$$\exists C<\infty: T\phi(m_T)/m_T\le C,\;\; \forall T\in \mathbb{N}.$$
	{(ii) For all $u\in \mathcal{U}$ and $\theta\in\Theta$, the approximate structural process, $\{y_{u,t}\}$, defined by \eqref{struct2}, satisfies (a) $E|y_{u,t}(\theta)|^2 < \infty$, (b) $Ey_{u,t}(\theta) = C$, and 
	(c) $Cov(y_{u,t}(\theta),y_{u,s}(\theta)) =  Cov(y_{u,t+m}(\theta),y_{u,s+m}(\theta))$ for all integers $t,s,m$.}
\end{Assumption}

For an arbitrary point $u\in\mathcal{U}$, define a local neighborhood of $\rho_{0}(u)$ as $\mathcal{E}:=\{\rho\in\Gamma:\|\rho-\rho_{0}(u)\|\le\varepsilon\}$
and $\mathcal{E}^{c}:=\{\rho\in\Gamma:\|\rho-\rho_{0}(u)\|>\varepsilon\}$. 

\begin{Assumption}  
	\label{ass:auxiliary}
	(i) For $f\in\mathbf{F}$, $g[Y_{t,T};f(Z_{t,T},\rho)]$ is twice continuously and boundedly differentiable in all arguments. For all $u\in\mathcal{U}$, $\sup_{\rho\in\mathcal{E}}E|g(Y_{t,T};f(Z_{t,T},\rho))|<\infty$.\newline 
	(ii)  For all $u\in\mathcal{U}$ and any $\rho\in \mathcal{E}$, $\rho\mapsto f(Z_{t,T};\rho)$ is measurable and twice differentiable at $\rho$ and satisfies 
	$\sup_{\rho\in\mathcal{E}} E|f(Z_{t,T};\rho)|<\infty$. In addition, there exists a function $\bar{f}(\cdot)$ such that
	$\sup_{\rho\in\mathcal{E}}f(z;\rho)\le \bar{f}(z)$ with $E|\bar{f}(Z_{t,T})|_{\infty} <\infty$.\newline
	(iii) Let $q[Y_{t,T};f(Z_{t,T},\rho)]=(\partial/\partial\rho)g[Y_{t,T};f(Z_{t,T},\rho)]$. The function $q(\cdot)$ is differentiable in $\rho$, for all $\rho \in \mathcal{E}$, and is strict monotonic, in $\rho$, in a neighborhood of $\rho_0$. Moreover, there exists a constant $c_q$ such that, up to an $O_p(T^{-1})$ term,
	$$\sup_{u\in\mathcal{U}}\sup_{\rho\in\mathcal{E}}\|q[Y_{t,T};f(Z_{t,T},\rho)]\|\le c_{q}.$$ \newline
    (iv) For any given $u\in\mathcal{U}$, the map 
	\begin{equation}
	\rho \mapsto \Psi_0(\rho;u):= (\partial/\partial\rho)\mathbb{M}_0[\rho;u]\label{eq:m_lim_obj},
	\end{equation}exists and $\rho_0(u)$ is the unique zero of 	$\Psi_0(\rho;u)$; i.e. for an arbitrarily small  $\varepsilon>0$, there exists $\eta >0$ such that $%
	\inf_{\rho\in\mathcal{E}^c}\mathbb{M}_0%
	\left[ \rho;u \right] -\mathbb{M}_0\left[ \rho_{0};u\right] \geq \eta .$	\newline
	(v)  $\forall\varepsilon>0,\exists a_1,a_2>0$ 
	\[
	\sup_{\substack{u_{1}\in\mathcal{U}\\ \rho_1\in\mathcal{E}}}\sup_{\substack{u_{2}:|u_{2}-u_{1}|\le a_1\\\rho_2:\|\rho_2-\rho_1\|\le a_2}}|Eg(y_{u_{1},t};f(z_{u_{1},t},\rho_1))-Eg(y_{u_{2},t};f(z_{u_{2},t},\rho_2))|\le\varepsilon
	\]
	is satisfied. \newline	
	(vi) For any  given $u\in\mathcal{U}$, $\theta\in\Theta$, $\theta\mapsto \rho_0(u;\theta)$ is continuous and injective for all $\theta\in\Theta$. \newline 
	(vii) The parameter spaces $\Gamma\subset\mathbb{R}^{d_{\rho}}$ and $\Theta\subset\mathbb{R}^{d_{\theta}}$, with $d_{\rho}\geq d_{\theta}$, are compact. 
\end{Assumption}
\begin{Assumption} \label{ass:AAker} 

	\medskip
	\noindent (i)The kernel function $x\mapsto K(\cdot )$ is positive, symmetric around zero, and bounded. In addition: 
	(i.a) $K(\cdot )$ is $r$-times continuously differentiable for
	$x\in\mathbb{R}$, with $r\geq2$;
	(i.b) $K(\cdot )$ satisfies $\displaystyle\int K(x)dx=1$, $\kappa_2=\int K^2(x)dx<\infty$, $\int |K(x)|dx<\infty$ and either $\sup_{x}K(x)<\infty$, $K(x)=0$ for $|x|>L$ with $L<\infty$ or $|\partial K(x)/\partial x|\le C$ and for some $v>1$, $|\partial K(x)/\partial x|\le C|x|^{-v}$ for $|x|>L$;
	(i.c) $\displaystyle\mu _{i}(K)=\int x^{i}K(x)dx=0$, $
	i=1,\ldots ,r-1$, and: $\displaystyle\int x^{r}K(x)dx\neq 0$, $\displaystyle
	\int |x|^{r}|K(x)|dx<\infty $, $\displaystyle\lim_{|x|\rightarrow \infty
	}|x|K\left( x\right) =0$;
	(i.d) $K(\cdot )$ is Lipschitz continuous, i.e. $%
	|K(x)-K(x^{\prime })|\leq C|x-x^{\prime }|$ for all $x,x^{\prime }\in
	\mathbb{R}$.
	
	\medskip	
	\noindent (ii) The bandwidth $h$ is such that, as $T\rightarrow
	\infty ,$ $h\rightarrow 0$, $Th\rightarrow \infty $ and $Th\big/(m_T \log T) \rightarrow \infty$.
	
\end{Assumption}
\begin{remark}\normalfont
Assumption \ref{ass:uniform}.(i) states that we restrict our attention to locally stationary processes and allows us to utilize the asymptotic independence property for heterogeneous data. The decay rate of the $\phi$-mixing coefficient is quite weak. For instance, any exponential decay rate satisfies the condition. In addition, the $\phi$-mixing can be relaxed to strong-mixing if we restrict the form of $g(\cdot)$. For instance, for the regression objective function - whether it is linear or nonlinear - strong-mixing assumption suffices. Assumption \ref{ass:uniform}.(ii) ensures that for all $\theta\in\Theta$, there exists an approximating stationary process. It is worth noting that this is a condition for the parameter space $\Theta$, and implicitly confines the size of $\Theta$, so that we exclude the possibility of generating non-stationary simulated series. Assumption \ref{ass:auxiliary} is concerned with the auxiliary model and its objective function. Assumptions \ref{ass:auxiliary}.(i) and \ref{ass:auxiliary}.(ii) ensure uniform continuity of the objective function in a neighborhood of the pseudo-true value, $\rho_0(u)$. They also ensure the existence of a well-behaved limit of the objective function due to the dominated convergence theorem. Assumption \ref{ass:auxiliary}.(iii) is concerned with the behavior of the first-order condition and the monotonicity warranted by the minimizer of the criterion is needed to ensure that the optimizer of the auxiliary criterion is unique. In general, one can replace this condition with the high-level condition that the auxiliary estimator ``nearly minimizes'' the criterion function, however, we believe this primitive condition is more informative than invoking this alternative high-level condition. Assumption \ref{ass:auxiliary}.(iv) is an asymptotic identification condition such that the unique minimizer of $\mathbb{M}_0$ is well separated and therefore unique. Assumption \ref{ass:auxiliary}.(v) states uniform equicontinuity for the uniform LLN. Assumption \ref{ass:auxiliary}.(vi) is an identification condition and is akin to a local version of the standard II identification condition.  Assumption \ref{ass:auxiliary}.(vii) requires that the parameter spaces for $\rho(u)$ and $\theta(u)$ are compact. Finally, 
Assumption \ref{ass:AAker} describes features of the kernel function and the bandwidth, which is standard in nonparametric kernel estimation.\hfill\(\Box\)

\end{remark}

Uniform (in $u$) consistency of the L-II estimator $\hat\theta(u)$ requires the uniform convergence of the auxiliary estimators $\hat{\rho}(u)$ and $\hat{\rho}(u;\theta)$ to their limit counterparts.
	
\color{black}
\begin{theorem}\label{thm0} 
	Under Assumptions \ref{ass:1}-\ref{ass:AAker}, $\hat{\rho}(u)$, and $\hat\rho(u;\theta)$ exist and are unique w.p.1. In addition, the following are satisfied.
	\begin{enumerate}
		\item The auxiliary estimator $\hat\rho(u)$, calculated using the observed sample $\{Y_{t,T}\}_{t\leq T}$, satisfies
		\begin{equation}
		\sup_{u\in\mathcal{U}}\|\hat{\rho}(u)-\rho_0(u)\|=o_{p}(1);\label{eq:uniformity_m_est}
		\end{equation}
		\item The auxiliary estimator $\hat\rho(u;\theta)$, calculated using the simulated sample $\{\tilde{y}_{u,t}\}_{t\leq T}$, satisfies
		\begin{equation}
		\sup_{u\in\mathcal{U}}\sup_{\theta\in\Theta}\|\hat{\rho}(u;\theta)-\rho_0(u;\theta)\|=o_{p}(1).\label{eq:uniformity_m_est2}
		\end{equation}
	\end{enumerate}
\end{theorem}
\begin{remark}\normalfont
We note here that Theorem \ref{thm0} is of independent interest. The result in equation \eqref{eq:uniformity_m_est} is one of the first results, to our knowledge, on the uniform consistency of estimators in general locally stationary models. The only other results in this direction that the authors are aware of are those for (quasi) maximum likelihood estimators in \cite{kristensen2019local} and \cite{dahlhaus2019towards}.\hfill\(\Box\)
\end{remark}
\begin{remark}\normalfont
As stated earlier, a useful class of auxiliary models for  L-II is the class of nonlinear regression models. 
Suppose that the auxiliary model is given by
\[
Y_{t,T}=f(Z_{t,T};\rho  (t/T))+\eta_{t},
\]	where $\eta_t$ is strictly stationary and $\phi$-mixing with $E|\eta_t|<\infty$ and independent of the explanatory variables $Z_{t,T}$. The estimator of the auxiliary parameter is given as
\begin{flalign}
\hat{\rho}(u)&=\argmin_{\rho\in\Gamma}M_T[\rho;u]\label{eq: rho_nonlinear},\text{ where }M_T[\rho;u]=\frac{1}{Th}\sum_{t=1}^{T}\left(Y_{t,T}-f(Z_{t,T};\rho)\right)^2 K\left(\frac{u-t/T}{h}\right).
\end{flalign} 
For this specific choice of auxiliary model and criterion function, we have the following immediate corollary to Theorem \ref{thm0}.  
\begin{corollary}
	\label{Lem:uniform_theta}
	Under Assumptions \ref{ass:1}-\ref{ass:AAker}, for $\hat{\rho}(u)$ defined as in \eqref{eq: rho_nonlinear} and $\hat{\rho}(u;\theta)$ its simulated counterpart, we have
	\begin{equation*}
	\sup_{u\in\mathcal{U}}\|\hat{\rho}(u)-\rho_0(u)\|=o_{p}(1),\text{ and }\sup_{u\in\mathcal{U}}\sup_{\theta\in\Theta}\|\hat{\rho}(u;{\theta})-\rho_0(u;{\theta})\|=o_{p}(1).
	\end{equation*}	
\end{corollary}	\hfill\(\Box\)
\end{remark}

The (uniform) consistency of $\hat{\rho}(u)$ and $\hat{\rho}(u;\theta)$ allows us to deduce the uniform consistency of the L-II estimator. 
	
	\begin{theorem}\label{thm1}
		Let Assumptions \ref{ass:1}-\ref{ass:AAker} be satisfied. For $\Omega$ a symmetric, positive-definite weighting matrix, the estimator $$\hat{\theta}(u):=\arg\max_{\theta\in\Theta}-\|\hat{\rho}(u)-\hat{\rho}(u;\theta)\|^2_{\Omega},$$ satisfies $$\sup_{u\in\mathcal{U}}\|\hat{\theta}(u)-\theta_{0}(u)\|_{}=o_{p}(1).$$
	\end{theorem}
\begin{remark}\normalfont
	Theorem \ref{thm1} requires, among other things, a condition guaranteeing identification of $\theta_{0}(u)$ for any $u\in\mathcal{U}$. This requires that, for any $u\in\mathcal{U}$ and for some $\theta\in\Theta$, $\rho_{0}(u;{\theta})$ is able to match $\rho_0(u)$, and that this matching be unique. Recalling that $\rho_{0}(u_{})=\rho_{0}(u_{};\theta_0({u}))$, this identification requires that $\theta_0(u)$ be the unique solution, in $\theta$, to$$\rho_0(u;\theta_0(u))=\rho_0(u;\theta)$$ {for $u\in\mathcal{U}$.} For  $\rho_0(\cdot;\theta)$ continuous and strictly monotonic, in $\theta$, for any $u$, in the case of $d_{\rho}=d_{\theta}=1$, this defines $\theta(\cdot)$ as $$\theta(\cdot)=\rho_0^{-1}(\cdot;\rho_{0}(\cdot;\theta_{0}(\cdot)).$$ Therefore, under continuity and monotonicity of $\rho^{}(\cdot;\theta)$, in $\theta$, for any $u\in{\mathcal{U}}$, $\theta_{0}(\cdot)$ is identified. Such a condition is equivalent to the injectivity conditions required by Theorem \ref{thm1}, which is a necessary condition required of parametric II (\citealt{GMR93}). Therefore, we see that if $\theta(t/T)=\theta$ for all $t/T$, i.e., the unknown function is constant, this identification condition is equivalent to the identification condition generally employed in parametric II estimation and Theorem \ref{thm1} reduces to the standard consistency result for II estimation.\hfill\(\Box\)
\end{remark}

\subsection{Asymptotic distribution}
In what follows, let
$\Psi_{T}(\rho;u):=\sum_{t=1}^{T}q[Y_{t,T};f(Z_{t,T},\rho)]K\left(\frac{u-t/T}{h} \right)/Th$ and recall the definitions $\Psi_0(\rho;u):=(\partial/\partial\rho)\mathbb{M}_0[\rho;u]$ and $\mathcal{E}:=\{\rho\in\Gamma:\|\rho-\rho_{0}(u)\|<\varepsilon\}$. We deduce the asymptotic distribution of the L-II estimator under the following high-level regularity conditions.
	 
\begin{Assumption}\label{ass:anormal}For fixed $u\in\mathcal{U}$, the following are satisfied. 
\begin{enumerate}
	\item There exists a matrix $V(u)$, satisfying $0<\inf_{u\in\mathcal{U}}\|V(u)\|\le\sup_{u\in\mathcal{U}}\|V(u)\|<\infty$, such that
	{$$\sqrt{Th}\Psi_{T}(\rho_0(u);u)\rightarrow_d \mathcal{N}\left(0,\kappa_2 V(u)\right).$$}
\item For $V(u)$ as in the above assumption, and for $\tilde{y}^0_{u,t}=\tilde{y}_{u,t}(\theta_0(u))$ denoting a realization simulated under $\theta_0(u)$, 
$
	\frac{1}{\sqrt{T}} \sum_{t=1}^{T}\left\{q(\tilde{y}^0_{u,t},\rho_0(u))\right\}\rightarrow_d \mathcal{N}\left(0,V(u)\right).
$
	\item For some $\varepsilon>0$, $\sup_{u\in\mathcal{U}}\sup_{\rho\in\mathcal{E}}\left\|\frac{\partial \Psi_{T}(\rho;u)}{\partial\rho^{\prime}}-\frac{\partial\Psi_0(\rho;u)}{\partial\rho^{\prime}}\right\|=o_{p}(1)$.
	\item $\Psi_0(\rho;u)$ and $\partial\Psi_0(\rho;u)/\partial\rho^{\prime}$ are Lipschitz continuous in both $u$ and $\rho$.
	\item $\partial\Psi_0(\rho_0(u);u)/\partial\rho^{\prime}$ is invertible for all $u\in\mathcal{U}$. 
	\item $\sup_{u\in\mathcal{U}}\|\partial \rho_0(u;\theta_0(u))/\partial u\|<\infty$,  $\sup_{u\in\mathcal{U}}\|\partial^2 \rho_0(u;\theta_0(u))/\partial u^2\|<\infty$, and $\partial \rho_0(u;\theta_0(u))/\partial \theta'$ is full column rank for all $u\in\mathcal{U}$. 
\end{enumerate}

\end{Assumption}
\begin{remark}\normalfont
	Assumption \ref{ass:anormal} amounts to a local version of the uniform convergence and asymptotic normality conditions required to demonstrate asymptotic normality of parametric II estimators. We note that it is feasible to consider more primitive assumptions that can guarantee the conditions in Assumption \ref{ass:anormal} (see, e.g., \cite{kristensen2019local} and \cite{dahlhaus2019towards} for discussion).  However, such an approach would require considerable technical effort and is not necessarily germane to the main message of this paper. Therefore, we leave the study of more primitive approaches to obtaining the required regularity in Assumption \ref{ass:anormal} for future research. \hfill\(\Box\)
\end{remark}
\begin{theorem}\label{thm2}
If Assumptions \ref{ass:1}-\ref{ass:anormal} are satisfied, and if $Th^3=o(1)$, then as $T\rightarrow\infty$ $$\sqrt{Th}\left[\hat{\theta}(u)-\theta_{0}(u)\right]\rightarrow_{d}\mathcal{N}\left(0,\kappa_2[Q^{-1}W V W' Q^{-1}](u)\right),
$$
where 
$Q(u)=\left\{\frac{\partial\rho_{0}(u;\theta)^{\prime}}{\partial\theta}\Omega\frac{\partial\rho_{0}(u;\theta)}{\partial\theta^{\prime}}\right\}\bigg{|}_{\theta=\theta_{0}(u)}$ and $W(u)=\frac{\partial \rho_0(u;\theta)^{\prime}}{\partial\theta}\Omega \left(\frac{\partial \Psi_{0}(\rho;u)}{\partial\rho^{\prime}}\right)^{-1}\bigg{|}_{\theta=\theta_{0}(u),\;\rho=\rho_{0}(u)}$.
\end{theorem}
\begin{remark}\normalfont
	In the L-II context, the bandwidth, $h$, affects the structural estimates through the estimated auxiliary parameter $\hat{\rho}(\cdot)$. Therefore, the bandwidth must be chosen with respect to the estimated auxiliary parameters. Theorem \ref{thm2} demonstrates that so long as the bandwidth satisfies $Th^3=o(1)$, the L-II estimator $\hat{\theta}(\cdot)$ will be asymptotically normal, and will not exhibit any asymptotic bias. Indeed, \cite{dahlhaus2019towards} argue that, for many different classes of locally stationary models estimated by local maximum likelihood, such a choice of bandwidth is optimal in terms of mean squared error. However, if one considers a slower rate for $h$, the resulting L-II estimator will be contaminated by an asymptotic bias. In such cases, the results given in \cite{kristensen2019local}, in particular their Corollary 1, can be used to deduce the general form of the bias. \hfill\(\Box\)
\end{remark}
\begin{remark}\normalfont
	For $H\geq1$ denoting the number of model simulations, the reader may notice that the $(1+1/H)$ term that generally appears in the asymptotic distribution of II estimators is absent in Theorem \ref{thm2}. The absence of this term is a consequence of matching a parametric estimator, $\hat\rho(u;\theta)$, against a nonparametric estimator, $\hat\rho(u)$. Recall that the estimator $\hat\rho(u;\theta)$ is based on $H$ simulated paths of length $T$, i.e., $TH$ total observations. Consequently, under regularity conditions, $\|\hat\rho(u;\theta_0(u))-\rho_0(u;\theta_0(u))\|=O_p\{(H^2{T})^{-1/2}\}$. In contrast, the local nature of the nonparametric estimator ensures that $\hat\rho(u)$ is based on an effective sample of $Th$ observations, which ensures that $\|\hat\rho(u)-\rho_0(u)\|=O_p(1/\sqrt{Th})$. Therefore, since $\rho_0(u)=\rho_0(u;\theta_0(u))$,  
	\begin{flalign*}
	\left\{\hat\rho(u;\theta_0(u))-\rho_0(u;\theta_0(u))\right\}+\left\{\rho_0(u;\theta_0(u))-\hat\rho(u)\right\}&=\left\{\hat\rho(u;\theta_0(u))-\rho_0(u)\right\}+\left\{\rho_0(u)-\hat\rho(u)\right\}\\&=O_p\left\{(H^2{T})^{-1/2}+(Th)^{-1/2}\right\},
	\end{flalign*}
	and the dominant order is $O_p(1/\sqrt{Th})$. Indeed, scaling the above by $\sqrt{Th}$ the first term is $O_p(\sqrt{h}/H)=o_p(1)$ for any $H\geq1$, since $h\rightarrow0$\bk.{ Therefore, the term $(1+1/H)$ will} not appear in the asymptotic distribution of the L-II estimator. 
	
	Intuitively, since L-II uses $TH$ simulated data points and $Th$ ``observed data points'', we are in a regime where we have more simulated data than observed data. In particular, since $h\rightarrow0$ as $T\rightarrow\infty$, the number of simulated observations, $TH$, diverges faster than the number of ``observed data points'', $Th$. In parametric II estimation it is well-known that if the number of simulated observations diverges faster than the number of observed data points, the $(1+1/H)$ factor does not appear in the asymptotic variance. \hfill\(\Box\)

\end{remark}
\begin{remark}\normalfont
We note that our L-II approach and its asymptotic behavior differ from the ``kernel-based'' II approach of \cite{billio2003kernel}. In the confines of a \textit{fully parametric} structural model, \cite{billio2003kernel} generate a simulated conditional auxiliary criterion function, via kernel smoothing, which is used to construct estimators of the auxiliary parameters. Matching auxiliary estimators based on the observed and simulated data, then ``knocks out'' the bias of the nonparametric estimator in the asymptotic distribution of the kernel II estimator. As such, in the authors \textit{parametric} context, the bandwidth used in estimation will have little impact on the behavior of the structural parameter estimates, and, hence, the researcher has liberty to choose this tuning parameter as they see fit. However, unlike \cite{billio2003kernel}, our structural model is \textit{nonparametric} and we must rely on local simulation (and estimation) of the structural model, in the neighborhood of the time point $u$. This local approach is required since in our context there is no reason to believe that a global approach will guarantee identification. As a result, we must pay the price for nonparametric estimation, which results in a slower rate of convergence.\hfill\(\Box\)
\end{remark}

\begin{remark}\normalfont
As is generally true of nonparametric estimators, the asymptotic distributions given in Theorem \ref{thm2} will only accurately reflect the sampling properties of the estimator in relatively large samples. As such, in cases with moderate sample sizes, we suggest the use of bootstrap techniques to conduct inference. While the bootstrap theory for locally stationary processes is still evolving, the bootstrap procedures of \cite{paparoditis2002local}, \cite{dowla2013local}, and \cite{kreiss2015bootstrapping} have been shown to be consistent in a wide variety of LS models.\hfill\(\Box\)
\end{remark}

\section{Simple example} 
\label{BBBsim}
In this section, we consider a simple generalization of the time-varying moving average model that allows the roots of the moving average lag polynomial to be time-varying. After presenting the model, we demonstrate how our L-II approach can be applied to estimate the model and present simulation results on the effectiveness of this strategy. 
\subsection{MA(1) time-varying parameters}
We consider the semiparametric locally stationary MA(1)-process
\begin{equation}
Y_{t,T}=\epsilon_{t}+\epsilon_{t-1}\theta_{0}(t/T).
\label{eq:fma1}
\end{equation}We further assume that $\epsilon_t$ is a white noise process with mean zero and unit variance, and $E\vert \epsilon_t\vert ^{4+\eta}<\infty$ for any arbitrarily small positive number $\eta$, and {we have that $\sup_{u\in \mathcal{U}}\vert\theta_0(u)\vert < 1$.} 

Our goal is to estimate the unknown function $\theta_0(\cdot)$ {via our L-II approach.} 
In doing so, we approximate (\ref{eq:fma1}) by a family of stationary MA(1) processes indexed by $u\in \mathcal{U}$  with some small trimming positive $\delta=o(1)$,
\begin{equation}
y_{u,t}=\epsilon_{t}+\epsilon_{t-1}\theta_{0}(u),\;\;\theta_0(u)\in[-1+\delta,1-\delta] \text{ } \forall u\in \mathcal{U}
\label{eq:fma1_ls}
\end{equation}
We consider an auxiliary model with the a locally stationary AR(1) structure:
\begin{equation}
y_{u,t}=\rho(u)y_{u,t-1}+\nu_{t},\text{ where }\rho(u)\in [-1+\delta,1-\delta].\label{eq: tv_reg}
\end{equation}For fixed $u$, the auxiliary model is a simple AR(1) model.

Recall that, in parametric MA models, when the roots lie near the region of non-inveribility, the resulting estimators can display a loss in accuracy. Therefore, since for any fixed $u$, the structural model is well-approximated by a parametric MA(1) model, it is likely that the same issue will be present if $\sup_u|\theta_0(u)|$ is close to unity.

We use the above auxiliary model to present a L-II estimator of $\theta_{0}(u)$. Algorithm \ref{II} describes the L-II estimation procedure for \eqref{eq:fma1}.
\begin{algorithm}
	\caption{L-II algorithm for locally stationary MA(1) processes}\label{II}
	\begin{algorithmic}[1]
		\State Based on observed data and auxiliary model \eqref{eq: tv_reg}, the estimator $\hat{\rho}(u)$ is defined as $$\hat{\rho}(u)={\sum_{t=1}^{T}Y_{t-1,T}Y_{t,T}K\left(\frac{u-t/T}{h}\right)}\bigg{/}{\sum_{t=1}^{T}Y_{t-1,T}^{2}K\left(\frac{u-t/T}{h}\right)} , $$where $K(\cdot)$ is a kernel function and $h$ is a bandwidth parameter.
		\State Based on the structural model \eqref{eq:fma1_ls}, given $j=1,\dots,H$ independent simulated realizations $\{\tilde{\epsilon}^{[j]}_{t}\}_{t=1}^{T}$, we generate $\{\tilde{y}^{[j]}_{u,t}; j=1,...,H\}$.
		\State Based on the simulated data $\{\tilde{y}^{[j]}_{u,t}\}_{j=1,...,H}$, obtain a set of estimators $\{\hat{\rho}^{[j]}(u;\theta)\}_{j=1,...,H}$ defined as
		$$\hat{\rho}^{[j]}(u;\theta)={\sum_{t=1}^{T}\tilde{y}^{[j]}_{u,t-1}(\theta)\tilde{y}^{[j]}_{u,t}(\theta)}\bigg{/}\sum_{t=1}^{T}\left(\tilde{y}^{[j]}_{u_i,t-1}(\theta)\right)^2$$and define $\hat{\rho}(u_{};\theta)=\frac{1}{H}\sum_{j=1}^{H}\hat{\rho}^{[j]}(u_{};\theta).$
		\State Define the estimator $\hat{\theta}(u)$ as the solution of  $\argmax_{\theta\in \Theta}-\|\hat{\rho}(u)-\hat{\rho}(u;\theta)\|$.
		\State Repeat the above procedure for different time points, say $\{u_i\}_{i=1,...,m}$ to estimate  $\theta_0(\cdot)$. 
	\end{algorithmic}
\end{algorithm}

In comparison with the general structure, the time-varying AR(1) auxiliary model in \eqref{eq: tv_reg} corresponds to taking $z_{u,t}=y_{u,t-1}$ and considering that $g(y_{u,t};\rho)=(y_{u,t}-\rho(u)y_{u,t-1})^2$. Note that it would also be possible to consider additional lags of $y_{u,t}$ in $z_{u,t}$ to accommodate LS-MA models of higher order. It is also useful to note that under weak conditions on the error term, the process $y_{t,T}$ defined in the auxiliary model \eqref{eq: tv_reg} is strong-mixing; see \citet{orbe2005nonparametric}.
	
In this specific model, using the result of Corollary \ref{Lem:uniform_theta}, we can deduce the consistency result in Theorem \ref{thm1} to obtain the following uniform convergence of $\hat{\theta}(u) $ in the LS-MA(1) model to $\theta_0(u)$.\footnote{The proof of Corollary \ref{corol1} follows from Corollary \ref{Lem:uniform_theta}, however, for clarity we give a more primitive proof in the Supplementary appendix.} 

\begin{corollary}\label{corol1}
	Under {Assumptions} \ref{ass:1}-\ref{ass:AAker}, if $\sup_{u\in\mathcal{U}}|\rho_0(u)|<1$ with uniformly bounded second-derivatives, the estimator $\hat{\theta}(u):=\arg\max_{{\theta}\in\Theta}-\|\hat{\rho}(u)-\hat{\rho}(u;\theta)\|_{}$ satisfies $\sup_{u\in\mathcal{U}}\|\hat{\theta}(u)-\theta_{0}(u)\|=o_{p}(1).$
\end{corollary}

\subsection{Monte Carlo experiments}
We demonstrate the usefulness of the L-II approach using a series of Monte Carlo experiments. We consider a sample size of $T=1000$ generated according to the LS-MA(1) model $$Y_{t,T}=\epsilon_{t}+\epsilon_{t-1}\theta_{0}(t/T),\;\epsilon_{t}\sim \mathcal{N}(0,1).$$Data is generated according to one of three functional specifications for $\theta_{0}(u)$:
\begin{itemize}
	\item[(a)] $\theta_{0}(t/T)=0.5\cdot(t/T)^2$;
	\item[(b)] $\theta_{0}(t/T)=0.25+(t/T)-(t/T)^2$;
	\item[(c)] $\theta_{0}(t/T)=0.5$.
\end{itemize} 
For inference on $\theta_{0}(\cdot)$, we use Algorithm \ref{II} with a Gaussian kernel and the rule of thumb bandwidth $h=1.06 T^{-1/5}$. We take $H=2$ for all simulation experiments.\footnote{As demonstrated in Theorem \ref{thm2}, the choice of $H$ does not have an asymptotic impact on the estimates. However, in finite samples this choice may affect the estimated values of $\theta_0(\cdot)$, since a larger value of $H$ generally yields a smoother criterion function, and potentially a more accurate optimizer.} We estimate $\theta_{0}(\cdot)$ across the grid of points $u\in\{.05,.10,.20,\dots,.90,.95\}$.

We consider 5,000 replications of the above design across the three different specifications for $\theta_{0}(\cdot)$. The following three figures illustrate the sampling distribution, across the Monte Carlo replications for each of the three specifications.

\begin{figure}[h]
		\centering
		\begin{subfigure}[b]{0.7\textwidth}
			\includegraphics[height=3cm, width=1\linewidth]{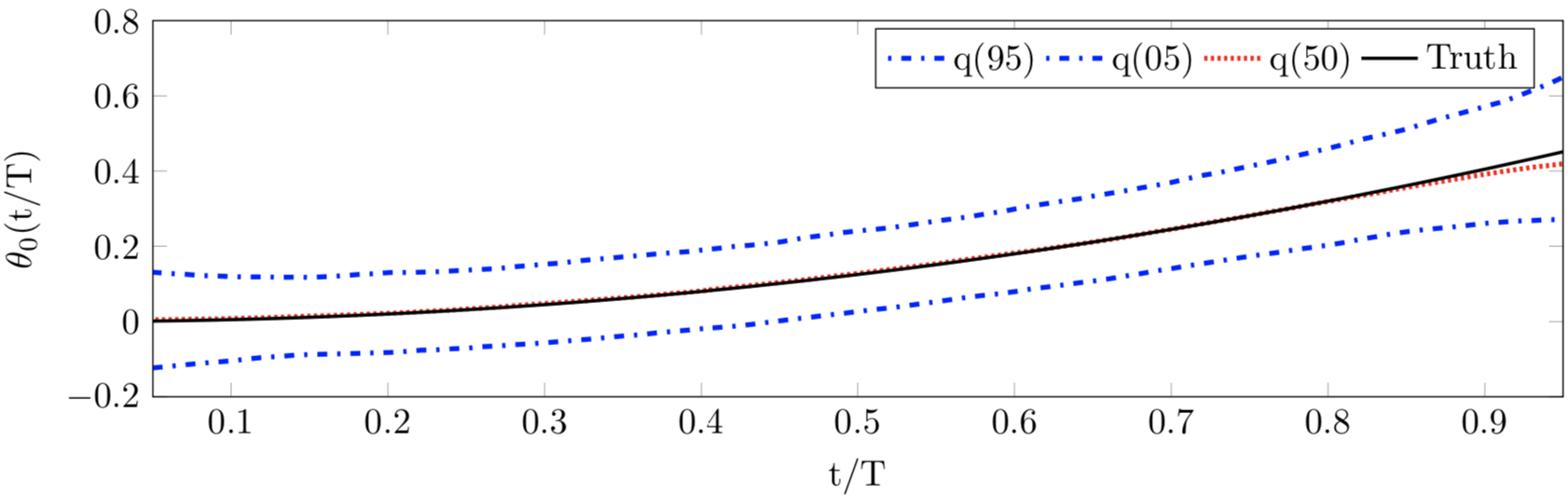}
			\caption{$\theta_{0}(t/T)=0.5\cdot(t/T)^2$}
			\label{fig:1} 
		\end{subfigure}		
		\begin{subfigure}[b]{0.7\textwidth}
			\includegraphics[height=3cm, width=1\linewidth]{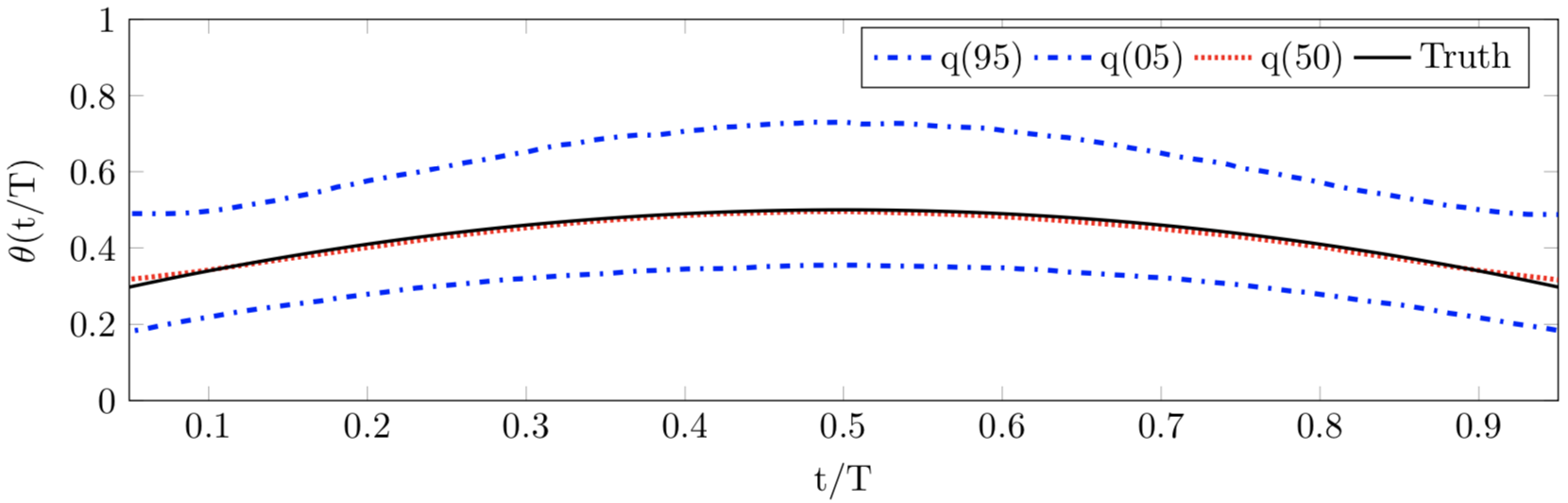}
			\caption{$\theta_{0}(t/T)=0.25+(t/T)-(t/T)^2$}
			\label{fig:2}
		\end{subfigure}	
		\begin{subfigure}[b]{0.7\textwidth}
		\includegraphics[height=3cm, width=1\linewidth]{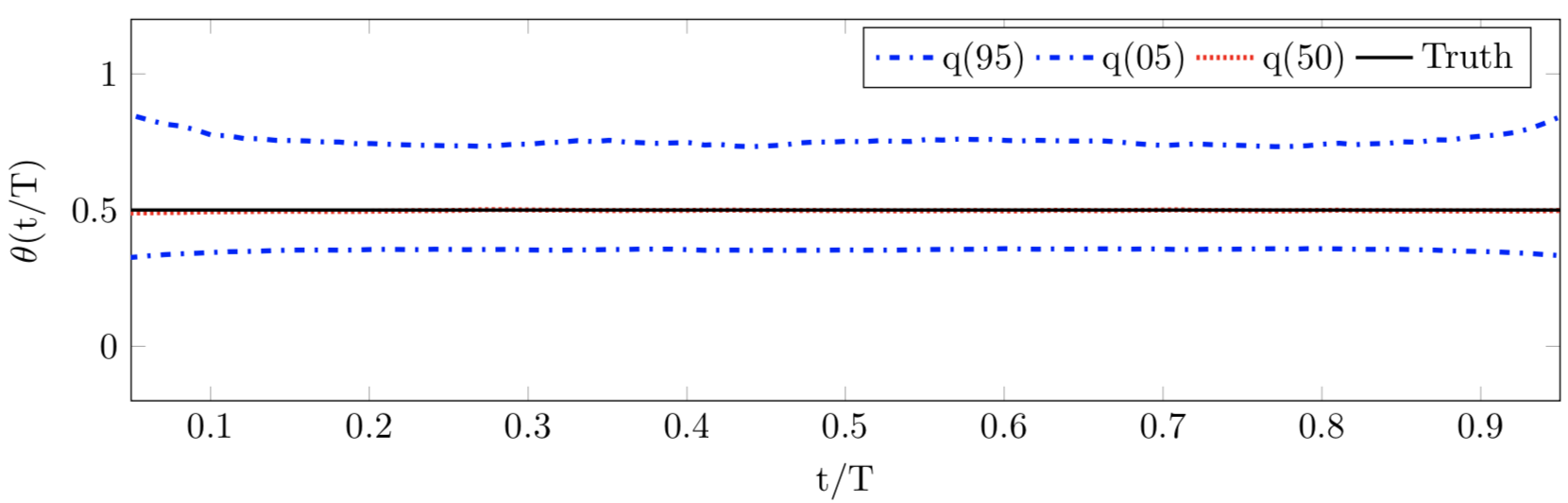}
		\caption{$\theta_{0}(t/T)=0.5$}
		\label{fig:3}
		\end{subfigure}	
	\caption{Sampling distribution of $\hat{\theta}(u)$ in the LS-MA(1) model based on 5,000 Monte Carlo replications. The dashed dotted lines represent the 0.05 and 0.95 quantiles across the Monte Carlo replications, while the dotted line represents the 0.50 quantile. The solid line represents the respective true unknown function.}
	\label{fig:simul}
\end{figure}

Figure \ref{fig:simul} demonstrates the ability of the L-II approach to obtain consistent estimators of the unknown function $\theta_{0}(\cdot)$ over $u\in\{.05,.10,.20,\dots,.90,.95\}$ across the three Monte Carlo designs. The bounds are truncated due to the well-known boundary bias associated with local constant nonparametric estimation. We note that, outside of these bounds, given the relatively short nature of the time series, these estimators are likely to be poorly behaved. This issue can be addressed through the use of local-linear smoothing approaches.

\section{Time-varying multiplicative stochastic volatility model}\label{sec:lssv}
The use of stochastic volatility to capture the conditional heteroskedastic movements of asset returns is now commonplace in economics and finance. Recently, however, several authors have suggested that volatility should be decomposed into short and long-run components (see, e.g, \citealp{ER08} and \citealp{engle2013stock}). Such a decomposition has given rise to the class of multiplicative time-varying GARCH models, e.g.  \cite{koo2015let}. Such models decompose volatility into a short-run  component, which is conveniently captured via a GARCH model, and a long-run component that slowly varies with larger macroeconomic factors that are captured nonparametrically. 


The class of multiplicative GARCH models can capture both short and long-run features, however, it is generally accepted that stochastic volatility models are superior to GARCH models in terms of modeling flexibility and their overall ability to capture fluctuations in short-run  volatility. Given this feature, one would suspect that a multiplicative extension of the standard SV model should perform well in many cases. While such a model would be similar to multiplicative GARCH models, the introduction of latent stochastic volatility ensures that direct estimation approaches become infeasible. However, this issue is immaterial for our L-II estimation approach since we can simulate the latent volatility 

To this end, in this section we propose a new model where volatility evolves as the product of a short and long-run component: the long-run component is captured by a slowly time-varying function, and the short-run  component is captured via an autoregressive SV model. In the context of simulation experiments, we demonstrate that our L-II approach can accurately estimate this new model. We then apply this model to analyze the volatility of monthly returns on twenty-five Fama-French portfolios, with the results indicating that long-run volatility changes dramatically over the sample period under analysis. 

Given the general nature of this paper, we leave a thorough discussion on the theoretical properties of this new SV model for future study.

\subsection{Model}
We now consider a multiplicative extension of the traditional stochastic volatility model. The observed demeaned data is generated according to
\begin{align}
Y_{t,T} &= \sqrt{\xi(t/T)}\exp{(h_t/2)}\nu_{1,t} ,\nonumber \\
h_{t+1} &= \mu + \phi h_t + \nu_{2,t}, \label{m1}
\end{align}
and where
\begin{equation*}
\begin{bmatrix}
\nu_{1,t}\\
\nu_{2,t}
\end{bmatrix}
\overset{iid}{\sim} \mathcal{N}
\begin{pmatrix}
\begin{bmatrix}
0\\0  
\end{bmatrix},
\begin{bmatrix}
1 &  \gamma_{\nu} \sigma\\
\gamma_\nu\sigma &  \sigma^2
\end{bmatrix}
\end{pmatrix},
\end{equation*}
with $\gamma_{\nu}$ the correlation coefficient between $\nu_{1,t}$ and $\nu_{2,t}$. 
In this model, the long-run trend is captured by the deterministic function $\sqrt{\xi({t/T})}$ whereas the short-run dynamics, $h_t$, are represented by the stochastic volatility model. We implicitly assume that $\{Y_{t,T}\}$ changes smoothly over time and if it were not for $\xi(\cdot)$, the slowly time-varying long-run trend, then $\{Y_{t,T}\}$ would be stationary. That is, we implicitly maintain that $\xi(\cdot)$ is uniformly positive and twice continuously differentiable, and $h_t$ is stationary, so that the process $\{Y_{t,T}/\sqrt{\xi(t/T)}\}_{}$ would be stationary. In the supplementary material, we give precise conditions on the function $\xi(\cdot)$ and the remaining parameters that ensure the resulting model is locally stationary.  

Directly estimating the structural model \eqref{m1}, and conducting statistical inference on the resulting estimates, is generally infeasible with existing methods. Instead, we propose to conduct inference on the structural model through L-II and by using as our auxiliary model the following locally stationary multiplicative GJR-GARCH model:
\begin{align}
y_{u,t} &= \sqrt{\tau(u)}\sigma_t z_t ,\nonumber \\
\sigma^2_{t+1} &= \omega + \alpha\sigma^2_{t} + \beta\left(\frac{y_{u,t}}{\sqrt{\tau(u)}}\right)^2 + \gamma \left(\frac{y_{u,t}}{\sqrt{\tau(u)}}\right)^2 I_{t}, \label{m2}
\end{align}
where $z_t \overset{iid}{\sim}\mathcal{N}(0,1)$ and $I_t = 0$ if ${y_{u,t}}/{\sqrt{\tau(u)}} \geq 0$, and $I_t = 1$ if ${y_{u,t}}/{\sqrt{\tau(u)}}< 0$.
In this setting, we will use the parameters in the auxiliary model, $\rho(\cdot) = (\tau(\cdot),\omega,\alpha,\beta,\gamma)'$, to conduct inference on the parameters of interest in the structural model, $\theta(\cdot) = (\xi(\cdot),\mu,\phi,\gamma_{\nu},\sigma)'$.

\citet{koo2015let} demonstrate that locally stationary multiplicative GARCH models can be estimated relatively easily. Note, however, that the symmetry of a GARCH(1,1) model would ensure that it is an unsuitable auxiliary model, as there is no parameter that can be readily matched to the correlation coefficient $\gamma_{\nu}$. Therefore, we employ the GJR-GARCH(1,1) model so that the leverage effect $\gamma_{\nu}$ is captured by the asymmetry parameter $\gamma$ in the auxiliary model. 

\subsubsection{Estimation procedure}

Before we discuss estimation of the LS-SV model, we note that, due to the multiplicative nature of the model for $Y_{t,T}$ in \eqref{m1}, an additional identification restriction is required in order to identify the unknown parameters. The restriction can be imposed on either the long-run or the short-run  part. For instance, while \citet{koo2015let} impose a restriction on the long-run component, \citet{engle2013stock} impose a restriction on the short-run  component. For our L-II, we impose a restriction on the short-run component for the LS-SV model because the L-II is applied over a finite number of fixed time points and therefore, a restriction on the long-run component in the structural model is difficult to implement. 

In particular, we impose the restriction that $\mu = 0$ for the structural model. Equivalently, for the auxiliary multiplicative GJR-GARCH model, we restrict $\omega = 1-\alpha-\beta -\frac{\gamma}{2}$ such that the GJR-GARCH process has unit unconditional variance ($\frac{\omega}{1-\alpha-\beta-\frac{\gamma}{2}}$ = 1). Under this setup, we conduct our L-II as follows. 

\paragraph{Estimation of the auxiliary model:}
Using the observations $\{Y_{t,T}\}_{t\le T}$, we estimate the auxiliary multiplicative GJR-GARCH model \`{a} la \citet{ER08} and \citet{koo2015let}. Specifically, from \eqref{m2}, for $\mathcal{I}_t$ denoting the information set at $t$, 
\begin{flalign*}
{E}(\log y^2_{u,t}|\mathcal{I}_{t-1})&=\log \tau(u) + {E}(\log \sigma^2_t|\mathcal{I}_{t-1}) + {E}(\log z^2_t|\mathcal{I}_{t-1})\\
&=\log \tau(u) \exp(C) = \log \tau^{\ast}(u)
\end{flalign*}
under the stationarity of $\sigma^2_t$ and $z_t$ and $\tau^{\ast}(u)=\tau(u)\exp(C)$ with $C=E(\log \sigma^2_t z_t^2|\mathcal{I}_{t-1})$.
 
We obtain an initial estimate $\log\hat{\tau}^{\ast}(u)$ as
\[
\log \hat{\tau}^{\ast}(u) = \operatorname{arg min}_{\tau^{\ast} \in \mathbb{R}_{+}}\sum_{t=1}^{T}(\log y^2_{u,t}-\log \tau^{\ast}(u))^2 K_h(u-t/T),
\]
where $K_h(\cdot)=K(\cdot/h)/h$ with a bandwidth $h$. Once we obtain $\hat{\tau}^{\ast}(u)$,  we calculate the intermediate estimator $\check{\tau}(u)$:
\[
	\check{\tau}(u)=\frac{\hat{\tau}^{\ast}(u) }{\int_{0}^{1}\hat{\tau}^{\ast}(u) du}
\]
because 
\[
\frac{{\tau}^{\ast}(u) }{\int_{0}^{1}{\tau}^{\ast}(u) du} = \frac{{\tau}(u)\exp(C)}{\int_{0}^{1}{\tau}(u)\exp(C) du}=\tau(u),
\]
when we impose a restriction that $\int_0^1\tau(u)du=1$.

{Note that the restriction, $\int_0^1\tau(u)du=1$ is not a model restriction but rather an estimation restriction that can be re-normalized or reconstructed arbitrarily. 
Once $\check{\tau}(u)$ is obtained, we estimate the GJR-GARCH parameters via maximum likelihood estimation based on the following transformed data $\check{y}_{u,t}=y_{u,t}\big{/}\sqrt{\check{\tau}(u)}$ and obtain the estimators $(\check{\omega},\check{\alpha},\check{\beta},\check{\gamma})'$. However, note that $\check{\rho}=(\check{\tau}(\cdot),\check{\omega},\check{\alpha},\check{\beta},\check{\gamma})'$ does not satisfy the restriction $\omega = 1-\alpha-\beta -\frac{\gamma}{2}$. To obtain a vector of parameter estimates that satisfy this restriction, we calculate $\hat{\tau}(u)=\check{\tau}(u)\left({\check{\omega}}/{1-\check{\alpha}-\check{\beta}-\frac{\check{\gamma}}{2}}\right)$ and use $\hat{\tau}(u)$ to construct $\hat{y}_{u,t} = y_{u,t}\big{/}\sqrt{\hat{\tau}(u)}$. Estimating the parameters in the GJR-GARCH model using the transformed dataset $\{\hat{y}_{u,t}\}_{t\le T}$ then yields $(\hat{\omega},\hat{\alpha},\hat{\beta},\hat{\gamma})'$. The vector of estimates $\hat{\rho}=(\hat{\tau}(u),\hat{\omega},\hat{\alpha},\hat{\beta},\hat{\gamma})'$ is then used in L-II as the auxiliary parameter estimates.\footnote{Imposing a restriction in maximum likelihood estimation is usually difficult but we avoid complicated constrained optimization in this way. This restriction or constraint is important for the L-II of this particular model. Another type of constraint is required for another type of structural and auxiliary models for L-II. We believe that imposing a general type of constraint in the context of L-II will open up another important research topic. We leave the analysis of constrained L-II for future research.}  

\paragraph{Simulation of the structural model:}
Based on \eqref{m1}, for a given $u\in\mathcal{U}$, we simulate $H$ independent structural processes under the restriction $\mu=0$, for some value of $\theta\in\Theta$ according to:
\begin{align}
\tilde{y}^{[j]}_{u,t} &= \sqrt{{\xi}(u)}\exp{(\tilde{h}^{[j]}_{t}/2)}\tilde{\nu}^{[j]}_{1,t} \nonumber \\
\tilde{h}^{[j]}_{t+1} &= {\phi} \tilde{h}^{[j]}_t + \tilde{\nu}^{[j]}_{2,t}, \label{simul}
\end{align}
with
\begin{equation*}
\begin{bmatrix}
\tilde{\nu}_{1,t}\\
\tilde{\nu}_{2,t}
\end{bmatrix}
\overset{iid}{\sim} \mathcal{N}
\begin{pmatrix}
\begin{bmatrix}
0\\0  
\end{bmatrix},
\begin{bmatrix}
1 &  {\gamma_{\nu}}{\sigma}_{}\\
{\gamma_{\nu}}{\sigma}_{} &  {\sigma}^2_{} 
\end{bmatrix}
\end{pmatrix},
\end{equation*}
In the simulation step, we restrict $\mu = 0$ to impose unit unconditional variance for the multiplicative SV model, which is compatible with the restriction on the auxiliary model,  $\omega = 1-\alpha-\beta -\frac{\gamma}{2}$.

\paragraph{Estimation of the simulated structural model via the auxiliary model and L-II:}
{For a given time point $u\in\mathcal{U}$, based on the simulated data $\{\tilde{y}^{[j]}_{u,t};j=1,...,H\}$, we first obtain a set of estimators $\{\hat{\rho}^{[j]}(u;\theta)\}_{j=1}^{H}$. Note that when $\{\hat{\rho}^{[j]}(u;\theta)\}_{j=1}^{H}$ is estimated for each fixed time point, $u$,  the parameter $\tau(u)$ in the auxiliary model is an unknown constant, not a function. This implies that we just estimate the GJR-GARCH model based on the simulated data $\{\tilde{y}^{[j]}_{u,t};j=1,...,H\}$, to obtain $\{\check{\omega}^{[j]},\check{\alpha}^{[j]},\check{\beta}^{[j]},\check{\gamma}^{[j]}\}_{j=1}^{H}$ and then obtain $\{\hat{\tau}^{[j]}(u)\}_{j=1}^{H}$, such that $\hat{\tau}^{[j]}(u)=\frac{\check{\omega}}{1-\check{\alpha}-\check{\beta}-\check{\gamma}/2}$ thanks to the restriction $\omega = 1-\alpha-\beta -\frac{\gamma}{2}$. Then we create transformed or normalized data $\hat{y}_{u,t}=\tilde{y}^{[j]}_{u,t}\big{/}\sqrt{\hat{\tau}^{[j]}(u)}$ and obtain $\{\hat{\omega}^{[j]},\hat{\alpha}^{[j]},\hat{\beta}^{[j]},\tilde{\gamma}^{[j]}\}_{j=1}^{H}$. From $\{\hat{\rho}^{[j]}(u;\theta)\}_{j=1}^{H}$ we can then construct $\hat{\rho}(u;\theta)=\sum_{j=1}^{H}\hat{\rho}^{[j]}(u;\theta)/{H}$.}

{Based on $\hat{\rho}(u)$ and $\hat{\rho}(u;\theta)$, we search for the best candidate for the given time point $u$ and define the estimator $\hat{\theta}(u)$ as the solution to:  $\argmax_{\theta\in \Theta}-\|\hat{\rho}(u)-\hat{\rho}(u;\theta)\|^2_{\Omega}$
where $\Theta$ is the parameter space for $\theta_0(u)$. The above procedure can then be repeated across a grid of points, say $\{u_i\}_{i=1,...,m}$ to estimate the whole functional form of $\theta_0(\cdot)$. } 

Summing up, Algorithm \ref{alg:msv} is employed for the L-II estimation of the locally stationary multiplicative stochastic volatility model.

\begin{algorithm}
	\caption{L-II algorithm for multiplicative stochastic volatility processes}\label{alg:msv}
	\begin{algorithmic}[1]					
		\State Using the auxiliary model (\ref{m2}), $\hat{\rho}(u)$ is estimated based on the observed data and under the identification restriction $\omega = 1-\alpha-\beta -\frac{\gamma}{2}$.
		\State For fixed $u$, simulate $H$ independent structural processes $\{\tilde{y}^{[j]}_{u,t}; j=1,...,H\}$ according to the structural model in equation \eqref{m1}, and under the restrictions $\mu = 0$ and $\theta \in \Theta$.
		\State Based on the simulated data $\{\tilde{y}^{[j]}_{u,t}\}_{j=1,...,H}$, using the auxiliary model (\ref{m2}), obtain $\hat{\rho}(u;\theta)=\sum_{j=1}^{H}\hat{\rho}^{[j]}(u;\theta)/{H}$, again, under the restriction of $\omega = 1-\alpha-\beta -\frac{\gamma}{2}$. 
		\State 	Based on $\hat{\rho}(u)$ and $\hat{\rho}(u;\theta)$, we search for the best candidate for the given time point $u$ and define the estimator $\hat{\theta}(u)$ as the solution to:  $\argmax_{\theta\in \Theta}-\|\hat{\rho}(u)-\hat{\rho}(u;\theta)\|^2_{\Omega}$.
		\State Repeat the above procedure across a grid of points, say $\{u_i\}_{i=1,...,m}$, to estimate $\theta_0(\cdot)$. 
	\end{algorithmic}
\end{algorithm}

\subsubsection{Monte Carlo experiment}

We now conduct a Monte Carlo experiment to illustrate L-II estimation of the locally stationary multiplicative stochastic volatility (LS-SV) model . We fix the sample size to be $T=200$, and we generate 5000 Monte Carlo replications from the LS-SV model in equation \eqref{m1} with parameters values given by $$\mu=0,\phi=0.2,\gamma_{\nu}=-0.5,\sigma=1,$$ and where the long-run volatility component is given by 
$$\xi(t/T)=0.2\sin(0.5\pi t/T)+0.8\cos(0.5\pi t/T).$$
We take as our auxiliary model for this Monte Carlo experiment the LS-GJR-GARCH(1,1) auxiliary model in equation \eqref{m2}.

Similar to the Monte Carlo experiments for the LS-MA(1) model, we estimate the auxiliary parameter via local constant estimation with a Gaussian kernel and rule of thumb bandwidth. We again set the number of simulations to be $H=2$. For full details of the estimation procedure, please refer to Algorithm \ref{alg:msv}. Across each Monte Carlo replication we apply the LS-II approach, and record the estimated function $\hat{\xi}(\cdot)$.\footnote{Results for the parametric components of the model are similar to those obtained for other II estimators, and are not presented for the sake of brevity.}  The estimation results for the unknown function are presented graphically in Figure \ref{sv_fig}. Similar to the results for the LS-MA(1) model, the LS-II procedure yields good estimates of the unknown function.\footnote{Similar to the previous Monte Carlo, we truncate the function estimate due to boundary bias problems associated with the local-constant smoothing approach considered in this implementation.}

\begin{figure}[h]
	\centering 
		\centering
	\resizebox{9cm}{6cm}{
	\input{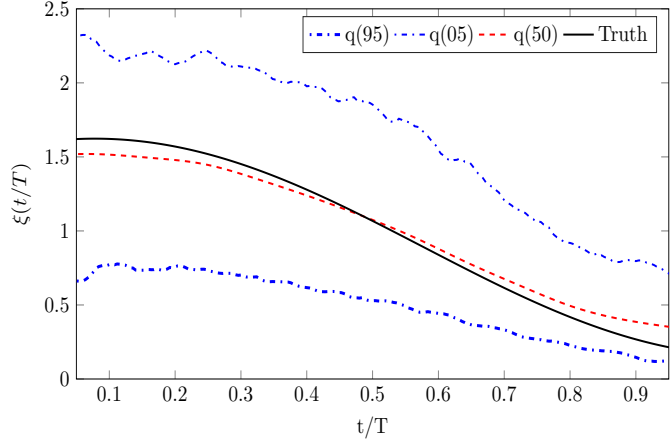} 			
	}	
	\caption{Sampling distribution of the estimated function $\hat{\xi}(u)$ across the 5,000 Monte Carlo replications. The dashed dotted paths represent the 0.05 and 0.95 quantiles across the Monte Carlo replications, while the dashed path represents the 0.50 quantile. The solid line represents the true unknown function and is given by $\xi(t/T)=0.2\sin(0.5\pi t/T)+0.8\cos(0.5\pi t/T)$.}
	\label{sv_fig}
\end{figure}

\subsection{Empirical application: LS-SV model}\label{BBBappl_lsma}
Herein, we analyse the behavior of monthly returns from January 1952 until December 2018 on 25 Fama-French portfolios formed from the intersection of five portfolios on size and five portfolios on book-to-market, and where the breakpoints for the portfolios are taken from the NYSE quintiles and are ordered from smallest to largest.\footnote{The data is freely available from Kenneth French's website.} The monthly return series on the Fama-French portfolios covers a long period of observation, and it is unlikely that these series display constant conditional covariance features over the entire sample period. In particular, while it is fairly widely accepted that these portfolios seem to display constant mean dynamics, the large fluctuations in the volatility of these series do not engender confidence that the conditional variance is constant throughout the sample period.\footnote{Considering an ARCH test of the demeaned returns for each of the 25 portfolios, where each test uses five lags, returns overwhelming support for the alternative hypothesis across all portfolios. The specific values can be found in the supplementary appendix.}  

Moreover, given the long time-span over which the data is measured, we argue that it is not realistic to assume that the volatility dynamics that were present in the 1950s have persisted unchanged until 2018. In particular, it is likely that underlying macroeconomic factors would cause these portfolios to exhibit patterns of volatility that display both short-term and long-run fluctuations, which can not be adequately captured by a stationary volatility model. To capture the long-run volatility patterns in the data, we consider a LS-SV version of the Fama-French three factor model. For $r_{t,j},\;j\in\{1,...,25\}$, denoting excess returns on the $j$-th portfolio,  we assume that $r_{t,j}$ evolves according to
\begin{flalign*}
r_{t,j}&=\alpha+\beta_1 r_{t,m}+\beta_{2}\text{SMB}_{t}+\beta_{3}\text{HML}_{t}+\epsilon_{t,j},\\\epsilon_{t,j}&=\sqrt{\xi_j(t/T)}\exp(h_{t,j}/2)\nu_{t,1j},
\end{flalign*}where $r_{t,m}$ denotes excess returns on the market factor, $\text{SMB}_t$ is the size factor, and $\text{HML}_t$ is the value factor. We model the short-term volatility component $h_{t,j}$ as
\begin{flalign*}
h_{t,j}&=\phi_j h_{t-1,j}+\sigma_{v,j}\nu_{t,2j},\;\;\;\text{corr}(\nu_{t,1j},\nu_{t,2j})=\gamma_{\nu,j},
\end{flalign*}
where  we require that the mean of the short-term SV component be zero to ensure the scale of $\xi(\cdot)$ can be properly identified. The above LS-SV model considers that volatility is the composition of two components: a long-run volatility trend that moves slowly and is captured by $\xi_j(t/T)$, and a term, measured by $h_{t,j}$, that captures short-term fluctuations around $\xi_j(t/T)$.

Estimation in the above LS-SV model can be carried out in two steps: first, we estimate the regression parameters to obtain $\hat{\alpha},\hat{\beta}_1,\hat{\beta}_2,\hat{\beta}_3$; in the second step, the residuals $$y_{t,j}=\left(r_{t,j}-\hat{\alpha}-\hat{\beta}_{1}r_{t,m} -\hat{\beta}_{2}\text{SMB}_{t}-\hat{\beta}_{3}\text{HML}_{t}\right)$$ are used within the L-II algorithm for the LS-SV model, along with a LS-GJR-GARCH auxiliary model (we refer the reader to Algorithm \ref{alg:msv} for specific implementation details). Before moving on, we note that the two-step nature of the L-II approach in this example means that it is straightforward to treat more complicated regression functions, such as, for instance, models with time-varying $\alpha$ and $\beta$. We refer the interested reader to the supplementary appendix where we consider an alternative specification for the conditional mean function that allows $\alpha,\;\beta$ to be time-varying.\footnote{These results largely mirror those given in the main text, and so we relegate these details to the supplementary material. In particular, we find that time varying versions of $\alpha$ and $\beta$ do not meaningfully deviate from constants for the sample period under analysis.}

L-II is used to estimate the short-term and long-run volatility components for all 25 portfolios. However, given the nature of the above estimation approach, uncertainty quantification is carried out using the local block-bootstrap (LBB) of \cite{paparoditis2002local}. The LBB is operationally similar to the block bootstrap but accounts for the changing stochastic structure of the observation process. Given observed data $y_{1},\dots,y_{T}$ the LBB generates a bootstrapped series of data, $y_{1}^*,\dots,y_{T}^*$, via the following steps. 
\begin{itemize}
	\item Select an integer block size $b$, and a fraction $B\in(0,1]$ such that $T\cdot B$ is an integer. 
	\item For $\lceil x \rceil$  the smallest integer that is greater than or equal to $x$, define $q:=\left(\lceil TB \rceil-1\right)$. For $i=0,1\dots,q$, let $k_0,\dots,k_q$ be i.i.d. integers generated from the uniform distribution that assigns probability $w(k)=1/(2TB+1)$ to the value $k$ when $-TB\leq k\leq TB$ and zero else. 
	\item Construct the bootstrap series $y_{1}^*,\dots,y_T^*$ by setting $y^*_{j+ib}=y_{j+ib+k_i}$ for $j=1,\dots,b$, and where $k_i$ is as given above and for $i=0,\dots,q$. 
\end{itemize}

In the following examples, across each of the 25 portfolios, we implement the LBB using $R=999$ bootstrap replications. Furthermore, we set the LBB block size, $b$, to be $b=10$, and take the local bootstrap parameter, $B$, to be $B\approx0.11$. 

The estimation results for $\alpha$ and $\beta$ are given in Table \ref{tab:beta_parms}, and the results for the parametric SV components are given in Table \ref{tab:sv_parms}. Focusing on the values of $\alpha,\;\beta$, we see that these estimated parameters are generally statistically significant and have the anticipated signs. Analysing Table \ref{tab:sv_parms}, we see that the short-term volatility parameters generally have statistically significant autocorrelation coefficients between 0.5 and 0.7, which indicates a moderate amount of short-term volatility persistence. The majority of the estimated values for $\sigma_v$ are between 1.5 and 2.0, indicating a relatively large level of noise in the short-term volatility process. Interestingly, none of the estimated leverage effects are statistically significant for the short-term volatility process. To ensure that this insignificance is not an artifact of the chosen auxiliary model, in Table \ref{tab:tarch} we report 99\% confidence intervals for the corresponding LS-GJR-GARCH auxiliary parameter $\gamma$, which captures the impact of asymmetric news on volatility, and where the confidence intervals are calculated using QMLE sandwich form standard errors. For 24 out of the 25 portfolios, the resulting LS-GJR-GARCH asymmetry parameter is statistically insignificant at the one percent significance level.
 
Leverage effects account for asymmetric reactions to volatility, possible due to larger macroeconomic forces. By their very nature, these macroeconomic forces are generally slowly varying, and their impact on volatility can then be adequately captured using the time-varying volatility approach considered herein. The insignificance of the estimated leverage effects can then be interpreted as follows: by decomposing volatility into a short-term  and long-run component, and by modeling the impact of such macroeconomic forces nonparametrically, the leverage effect is soaked-up by the long-run volatility component; its inclusion in the short-term volatility component is then redundant and, hence, statistically insignificant.

We present the estimates of $\xi(\cdot)$ graphically in Figures \ref{FF_fig1}-\ref{FF_fig5} in the appendix. The reported confidence bounds are the corresponding pointwise, for each value of $u=t/T$, confidence bounds obtained using the LBB.

The long-run volatility component captures gradual changes in volatility, possibly due to slowly-varying macroeconomic factors that affect returns (see, e.g, \citealp{ER08} and \citealp{engle2013stock} for a detailed discussion). Given this aim, the results in Figures \ref{FF_fig1}-\ref{FF_fig5} are compelling as they closely align with the larger macroeconomic risk profile of returns over the sample period under analysis. In particular, during the 1950s to the early 1960s most series display relatively low volatility that is either flat or slightly increasing till the early-to-mid 1960s, with the overall trend of most series decreasing after about 1965. This overall trend is then maintained all the way through the great moderation of the 1980s. However, after the end of the great moderation, virtually every series exhibits a significant upswing in long-run volatility. This pattern then continues and culminates around the time of the global financial crisis in the late 2000s, after which there is another sustained decrease in long-run volatility. 

Given how well our results correspond to the overarching long-run volatility patters, we note that more than half of these return series now exhibit an additional steeping of long-run volatility. This may indicate that since 2016 we have entered into a new period of long-run macroeconomic volatility.

\section{Discussion}
\label{BBBconcl}
We propose a novel indirect inference estimator for locally stationary processes and thereby extend, for the first time, the use of indirect inference estimation to general classes of semiparametric models with slowly time-varying parameters. As part of this study, we also propose a novel local stationary multiplicative stochastic volatility (LS-SV) model. We leave two important topics for future research: the efficiency of the L-II estimator, and the ensuing semiparametric efficiency bound for the class of locally stationary models considered in this paper; and the incorporation of shape restriction for nonparametric estimation within L-II, which may improve efficiency, e.g. \citet{horowitz2017nonparametric}, at the cost of a more complicated estimation approach. 

\newpage 
\small{
	\bibliographystyle{apalike}
	\bibliography{LII}}
\newpage

\appendix
\section{Proofs of main results}
\label{app: proofs}
	\begin{proof}[Proof of Lemma \ref{Lem:struct_ls}]
	From the triangle inequality, for all $u_0\in \mathcal{U}$, we have
	\begin{equation*}
	\left\vert Y_{t,T}-y_{u_0,t}\right\vert \leq \left\vert Y_{t,T}-y
	_{t/T,t}\right\vert +\left\vert y_{t/T,t}-y_{u_0,t}\right\vert \le O_{p}(T^{-1})+\left\vert y_{t/T,t}-y_{u_0,t}\right\vert,
	\end{equation*}where the $O_p(T^{-1})$ term follows from Definition \ref{def: ls}. Now, consider $\left\vert y_{t/T,t}-y_{u_0,t}\right\vert$ and expand $y_{t/T,t}$, via \eqref{struct2}, in a neighborhood of $u_{0}$:
	\begin{flalign*}
	y_{t/T,t}&=y_{u_{0},t}+\left[ {t}/{T}-u_{0}\right] \left.
	\frac{\partial y_{u,t}}{\partial u}\right|_{u=u_{0}}+\frac{1}{%
		2}\left[ {t}/{T}-u_0\right] ^{2}\left. \frac{\partial ^{2}y_{u,t}}{%
		\partial u^{2}}\right|_{u=u_{0}}+O_{p}\left( \left[ {t}/{T}-u_0%
	\right] ^{2}\right)\\
	&=y_{u_{0},t}+\left[ {t}/{T}-u_{0}\right] \left.
	\left[\frac{\partial r}{\partial\theta}+\frac{\partial r}{\partial\varphi}\frac{\partial\varphi}{\partial\theta}\right]\frac{\partial\theta_0}{\partial u} \right| _{u=u_{0}}+o_p\left( \left\vert {t}/{T}-u_0%
	\right\vert\right).
	\end{flalign*}%
	From Assumption \ref{ass:1}, in particular the (uniform) bounded second-derivatives of $r(\cdot),\varphi(\cdot),\theta_0(\cdot)$, it follows that
	$
	\left\vert y_{t/T,t}-y_{u_0,t}\right\vert = O_{p}\left(
	\left\vert {t}/{T}-u_0\right\vert \right)$. We then have that
	\begin{equation*}
	\left\vert Y_{t,T}-y_{u_0,t}\right\vert \leq O_p(T^{-1}) + O_{p}\left(
	\left\vert {t}/{T}-u_0\right\vert \right)=O_p\left(\left\vert {t}/{T}-u_0\right\vert+T^{-1}\right).
	\end{equation*} 				
\end{proof}	

\subsection{Proof of Theorem \ref{thm0}}
\begin{proof}
	Theorem \ref{thm0} consists of two uniform consistency results: 1. uniform consistency of the auxiliary estimator, $\hat{\rho}(u)$ based on the observed sample $\{Y_{t,T}\}$ to the pseudo-true value $\rho_0(u)$; 2. uniform consistency of the auxiliary estimator, $\hat{\rho}(u;\theta)$ based on the simulated sample $\{\tilde{y}_{u,t}\}$ to the pseudo-true value $\rho_0(u;\theta)$.
	
	Our proof strategy is twofold. In Part 1, firstly we show that for the true $\theta_0(u)$, 
	\begin{equation}
	\sup_{u\in{\mathcal{U}}}\|\hat{\rho}(u)-\rho_0(u)\|=\sup_{u\in{\mathcal{U}}}\|\hat{\rho}(u;\theta_0(u))-\rho_0(u;\theta_0(u))\|=o_p(1),
	\end{equation} 
	which proves the first part as in \eqref{eq:uniformity_m_est} of Theorem \ref{thm0}.
	
	In Part 2, combined with Part 1, we show the uniform consistency over $u\in\mathcal{U}$ and $\theta\in\Theta$, i.e. the second part as in \eqref{eq:uniformity_m_est2}  of Theorem \ref{thm0}  by using the simulated data and local stationarity.  
	
	\bigskip 
	
	\noindent\textbf{Part 1:}	
	In what follows, we suppress the dependence of $Y_{t,T}$ and $\rho(u)$ on $\theta_0$.
	
	Define $$\Psi_{T}(u,\rho(u))=\sum_{t=1}^{T}w_{t}(u)q(Y_{t,T};f(Z_{t,T},\rho))$$
	where $w_{t}(u)=(Th)^{-1}K_{ut}$, $q(Y_{t,T};f(Z_{t,T},\rho))=(\partial/\partial\rho)g(Y_{t,T};f(Z_{t,T},\rho))$ and $K_{ut}=K\left((u-t/T)/h\right)$. By construction, 
	\begin{equation}
	\Psi_{T}(u,\hat{\rho}(u))=0.
	\label{eq:H_T_zero}
	\end{equation}	
	For an arbitrarily small number $\varepsilon>0$, let $\|\hat{\rho}(u)-\rho_{0}(u)\|\le\varepsilon$.
	Firstly, we focus on the existence of unique minimizer of $M_T(\rho)$ or solution to \eqref{eq:H_T_zero}. 
	We consider w.l.o.g. $D$ as a compact $d_{\rho}$-dimensional set
	in the vicinity of the origin. We divide $D$ into $N$ disjoint coverings
	of the form such that $B_{j}=\{\delta:\|\delta-\delta_{j}\|\le\epsilon_{T}\};j=1,...,N$ for some $\epsilon_T>0$ and $\epsilon_T=o(1)$.
	Since $D$ is compact, it can be covered by a finite number of $B_{j}$s
	for $j=1,...,N$ and $N\le c/\epsilon_{T}$. 
	\begin{flalign*}
	& \sup_{u\in\mathcal{U}}\sup_{\delta\in D}|\Psi_{T}(u,\rho_{0}(u)+\delta)-E\Psi_{T}(u,\rho_{0}(u)+\delta)|\\
	\le & \sup_{u\in\mathcal{U}}\max_{1\le j\le N}\sup_{D\cap B_{j}}\vert \Psi_{T}(u,\rho_{0}(u)+\delta)-\Psi_{T}(u,\rho_{0}(u)+\delta_{j})\vert\\
	& +\sup_{u\in\mathcal{U}}\max_{1\le j\le N}\vert \Psi_{T}(u,\rho_{0}(u)+\delta_{j})-E\Psi_{T}(u,\rho_{0}(u)+\delta_{j})\vert\\
	& +\sup_{u\in\mathcal{U}}\max_{1\le j\le N}\sup_{D\cap B_{j}}\vert E\Psi_{T}(u,\rho_{0}(u)+\delta_{j})-E\Psi_{T}(u,\rho_{0}(u)+\delta)\vert=\mathcal{S}_{1}+\mathcal{S}_{2}+\mathcal{S}_{3}.
	\end{flalign*}
	Due to Assumption \ref{ass:auxiliary}.(iii) and Assumption \ref{ass:AAker}.(i), 
	\[
	\mathcal{S}_{1}\le c_{q}C\|\delta-\delta_{j}\|=O_p(r_{T}).
	\]
	where $r_T=((m_T\log T)\big/Th)^{1/2}$.
	For $\mathcal{S}_{3}$, in a similar way, for some $\epsilon$, 
	\[
	P\left(\max_{1\le j\le N}\sup_{D\cap B_{j}}\vert E\Psi_{T}(u,\rho_{0}(u)+\delta_{j})-E\Psi_{T}(u,\rho_{0}(u)+\delta)\vert>\epsilon\right)=O(r_{T}).
	\]
	For $\mathcal{S}_{2}$, 
	\[
	\mathcal{R}_{T}=\sup_{u\in\mathcal{U}}\max_{1\le j\le N}|\Psi_{T}(u,\rho_{0}(u)+\delta_{j})-E\Psi_{T}(u,\rho_{0}(u)+\delta_{j})|
	\]
	\begin{equation}
	P(\mathcal{R}_{T}>\epsilon)\le\sum_{j=1}^{N}P(\sup_{u\in\mathcal{U}}|\Psi_{T}(u,\rho_{0}(u)+\delta_{j})-E\Psi_{T}(u,\rho_{0}(u)+\delta_{j})|>\epsilon)\label{eq:lemma2}
	\end{equation}
	Due to Lemma \ref{Lem:uniform_u}, for some finite numbers, $\epsilon$, $N$, $C_1$ and $C_2$,
	\[
	P(\mathcal{R}_{T}>\epsilon)\le N C_1 e^{-C_2 Th/m_T}.
	\]
	Note that $e^{-C_2Th/m}<T^{-C_2\tau}$ with $\tau\rightarrow\infty$ as
	$T\rightarrow\infty$. 
	This implies that 
	\[
	\sum_{T=1}^{\infty}P(r_T\mathcal{R}_{T}>\epsilon)<\infty.
	\]
	Combining all the above results with the Borel-Cantelli Lemma yields 
	\begin{equation}
	\sup_{u\in\mathcal{U}}\sup_{\rho\in\Gamma}P\left\{ |\Psi_{T}(u,\rho)-E\Psi_{T}(u,\rho)|\ge\epsilon\right\} \rightarrow0\text{ w.p.1}.
	\label{eq:stochastic_sup}
	\end{equation}
	By Assumptions \ref{ass:auxiliary}.(iii) and \ref{ass:auxiliary}.(iv), and Assumption \ref{ass:AAker}, for any $\delta \in \mathbb{R}^{d_{\rho}}$ which satisfies that $\Psi_T(u,\rho_0(u)+\delta)\neq 0$, 
	$\Psi_{0}(u,\rho_0+\delta)\neq 0$ so that \eqref{eq:stochastic_sup} implies it with probability approaching to zero for all $u\in\mathcal{U}$ as $T$ tends to infinity.
	For the uniform consistency, due to Assumption \ref{ass:auxiliary}.(iii), the strict monotonicity of $q(\cdot)$ at the pseudo-true value, $\rho_0$
	implies for $u\in\mathcal{U}$, and for $\iota$ a $d_\rho$ dimensional vector of ones,
	\[
	[\Psi_{0}(u,\rho(u)+\varepsilon\cdot \iota)]_j<0<[\Psi_{0}(u,\rho(u)-\varepsilon\cdot \iota)]_j,\text{ for }j=1,\dots,d_\rho.
	\]
	where $\Psi_{0}(u,\rho)$ is defined as in \eqref{eq:m_lim_obj} and where, for $X\in\mathbb{R}^{d_\rho}$, $[X]_j$ denotes the $j$-th element of the vector. 
	This implies that for all $u\in\mathcal{U}$, as $T\rightarrow\infty$,
	\begin{equation}
	[\Psi_{T}(u,\rho(u)+\varepsilon\cdot\iota)]_j<0<[\Psi_{T}(u,\rho(u)-\varepsilon\cdot\iota)]_j,\text{ for }j=1,\dots,d_\rho\label{eq:H_T}
	\end{equation}
	By construction, (\ref{eq:H_T}) means that for all $u\in\mathcal{U}$, w.p.1.,
	\[
	\rho_0(u)-\varepsilon\cdot\iota<\hat{\rho}(u)<\rho_0(u)+\varepsilon\cdot\iota
	\]
	due to \eqref{eq:H_T_zero} and $K(\cdot)>0$ in Assumption \ref{ass:AAker}. In combination with Assumption \ref{ass:auxiliary}.(v), for $\theta_0$, the first part of Theorem \ref{thm0} as in \eqref{eq:uniformity_m_est}  holds:
	\[
	\sup_{u \in \mathcal{U}}\left\Vert\hat{\rho}(u)-\rho_{0}(u)\right\Vert=\sup_{u \in \mathcal{U}}\left\Vert\hat{\rho}(u;\theta_0(u))-\rho_{0}(u;\theta_0(u))\right\Vert\rightarrow 0 \ \ \ \ \text{w.p.1}.
	\]
	\medskip
	
	\noindent\textbf{Part 2:} 
	We first show that, for a given $u$, the resulting auxiliary criterion function, based on the observed data, is uniformly well-behaved and close to its limit counterpart. By virtue of the stationary nature of the simulated data, and, in particular, Assumptions \ref{ass:uniform}.(ii) and \ref{ass:auxiliary}.(v), we show that the same conclusion remains for the simulated criterion function. Lastly, continuity of the simulated objective function, in $\rho$, and compactness of the parameter spaces, $\Theta$ and $\Gamma$, can be used to show that $\hat\rho(u;\theta)$ is uniformly close to $\rho_0(u;\theta)$ in $\theta$, for all $u\in\mathcal{U}$, which yields the result. When no confusion will result, we again suppress the dependence of observed quantities on $\theta_0$ and simulated quantities on $\theta\in\Theta$, respectively.
	
	
	Simplify notation by denoting $g(\rho):=g(Y_{t,T};f(Z_{t,T},\rho))$ and $g(\rho_{0}):=g(Y_{t,T};f(Z_{t,T},\rho_{0}))$ and define $p_{t}(\rho)=\left[g(\rho)-g(\rho_{0})\right]$. Consider 
	\begin{flalign*}
	& M_{T}(\rho)-M_{T}(\rho_{0})\\
	= & \frac{1}{Th}\sum_{t=1}^{T}g(Y_{t,T};f(Z_{t,T},\rho))K_{ut}-\frac{1}{Th}\sum_{t=1}^{T}g(Y_{t,T};f(Z_{t,T},\rho_{0}))K_{ut}\\
	= & E\left[g(\rho)-g(\rho_{0})\right]-E\left[g(\rho)-g(\rho_{0})\right]\\
	& +\frac{1}{Th}\sum_{t=1}^{T}\left(\left[g(\rho)-g(\rho_{0})\right]-E\left[g(\rho)-g(\rho_{0})\right]\right)K_{ut}\\
	& +\frac{1}{Th}\sum_{t=1}^{T}E\left[g(\rho)-g(\rho_{0})\right]K_{ut}\\
	= & \underbrace{Ep_{t}(\rho)}_{\mathcal{M}_{1}(\rho)}+\underbrace{Ep_{t}(\rho)\left[\frac{1}{Th}\sum_{t=1}^{T}K_{ut}-1\right]+\frac{1}{Th}\sum_{t=1}^{T}\left[p_{t}(\rho)-Ep_{t}(\rho)\right]K_{ut}}_{\mathcal{M}_{2}(\rho)}
	\end{flalign*}
	
	\noindent Firstly regarding $\mathcal{M}_{1}(\rho)$, due to Assumptions
	\ref{ass:auxiliary}.(i), (ii) and (vi), 
	with dominated convergence theorem,
	$\mathcal{M}_{1}(\rho)=E[g(\rho)-g(\rho_{0})]$ is continuous at
	$\rho_{0}(u)$, $\mathcal{M}_{1}(\rho)$ is nonstochastic and constant
	with respect to $\mathcal{E}$. For identifiability, due to Assumption
	\ref{ass:auxiliary}.(iv), $|\mathcal{M}_{1}(\rho)|>0$ for all $\rho\in\Gamma$ except
	for $\rho_{0}$, i.e. $|\mathcal{M}_{1}(\rho)|>0$ whenever $\rho\neq\rho_{0}(u)$.
	This and continuity of $\mathcal{M}_{1}(\rho)$ imply that $\mathcal{M}_{1}(\rho)$
	is bounded away from 0 whenever $\rho\in\mathcal{E}^{c}$, i.e. $\rho$
	is outside of a neighborhood of $\rho_{0}(u)$. Furthermore, by compactness
	of $\Gamma$ and continuity, $\sup_{\rho\in\Gamma}|\mathcal{M}_{1}(\rho)|<\infty$.
	
	\noindent Meanwhile, with respect to $\mathcal{M}_{2}(\rho)$, we
	have two components. Firstly, for the first term of $\mathcal{M}_{2}$,
	\[
	\sup_{\rho\in\Gamma}\left|Ep_{t}(\rho)\left[\frac{1}{Th}\sum_{t=1}^{T}K_{ut}-1\right]\right|\stackrel{p}{\rightarrow}0
	\]
	since, as $T\rightarrow\infty$, $\frac{1}{Th}\sum_{t=1}^{T}K_{ut}\rightarrow1$
	and $\sup_{\rho\in\Gamma}\left|\mathcal{M}_{1}(\rho)\right|<\infty$
	as mentioned previously. For the second term of $\mathcal{M}_{2}(\rho)$,
	we need to show
	\[
	\sup_{\rho\in\Gamma}\left|\frac{1}{Th}\sum_{t=1}^{T}\left[p_{t}(\rho)-Ep_{t}(\rho)\right]\right|\stackrel{p}{\rightarrow}0.
	\]
	We discuss two cases: 1) middle part 2) tail part. For some constant
	$C<\infty$, let us define $p_{t}^{\ast}(\rho)=p_{t}(\rho)1(|p_{t}(\rho)|\le C)$
	where $1(\cdot)$ is the indicator function and $p_{t}^{\ast\ast}(\rho)=p_{t}(\rho)1(|p_{t}(\rho)|>C)$
	or $p_{t}^{\ast\ast}(\rho)=p_{t}(\rho)-p_{t}^{\ast}(\rho)$. 
	\begin{flalign*}
	& E\left|\frac{1}{Th}\sum_{t=1}^{T}K_{ut}\left\{ p_{t}(\rho)-Ep_{t}(\rho)\right\} \right|\\
	\le & E\left|\frac{1}{Th}\sum_{t=1}^{T}K_{ut}\left[p_{t}^{\ast}(\rho)-Ep_{t}^{\ast}(\rho)\right]\right|+E\left|\frac{1}{Th}\sum_{t=1}^{T}K_{ut}\left[p_{t}^{\ast\ast}(\rho)-Ep_{t}^{\ast\ast}(\rho)\right]\right|
	\end{flalign*}
	For any fixed $\rho$, 
	\[
	E\left|\frac{1}{Th}\sum_{t=1}^{T}K_{ut}\left[p_{t}^{\ast\ast}(\rho)-Ep_{t}^{\ast\ast}(\rho)\right]\right|\le\frac{2}{Th}\sum_{t=1}^{T}|K_{ut}|E|p_{t}^{\ast\ast}(\rho)|
	\]
	which can be arbitrarily small for $C$ and $T$ large enough irrespective
	of $\rho$.
	
	For some constant $0<J<C$ such that data is selected via Kernel $(u-t/T)/h\le J$,
	\begin{flalign*}
	& E\left|\frac{1}{Th}\sum_{t=1}^{T}K_{ut}\left[p_{t}^{\ast}(\rho)-Ep_{t}^{\ast}(\rho)\right]\right|\\
	= & E\left|\frac{1}{Th}\sum_{|t-uT|\le JTh}^{T}K_{ut}\left[p_{t}^{\ast}(\rho)-Ep_{t}^{\ast}(\rho)\right]+\frac{1}{Th}\sum_{|t-uT|>JTh}^{T}K_{ut}\left[p_{t}^{\ast}(\rho)-Ep_{t}^{\ast}(\rho)\right]\right|\\
	\le & E\left|\frac{1}{Th}\sum_{|t-uT|\le JTh}^{T}K_{ut}\left[p_{t}^{\ast}(\rho)-Ep_{t}^{\ast}(\rho)\right]\right|+\frac{2C}{Th}\sum_{|t-uT|>JTh}^{T}\left|K_{ut}\right|
	\end{flalign*}
	The second term tends to zero as $T\rightarrow\infty$. For the first
	term, 
	\begin{flalign*}
	& E\left[\frac{1}{Th}\sum_{|t-s|\le JTh}^{T}K_{ut}\left[p_{t}^{\ast}(\rho)-Ep_{t}^{\ast}(\rho)\right]\right]^{2}\\
	\le & \frac{C^{2}}{T^{2}h^{2}}\left\{ \sum_{|t-uT|\le JTh}K^2_{ut}+\sum_{|t-uT|\le JTh}\sum_{|s-uT|\le JTh;s\neq t}\left|K_{ut}K_{st}\right|\phi(|t-s|)\right\} \\
	= & O\left(\frac{C^{2}}{Th}\sum_{j\le Th}\phi(j)\right)
	\end{flalign*}
	where $\phi(\cdot)$ is the $\phi$-mixing coefficient defined as in Assumption \ref{ass:uniform}.(i). Due to Assumption \ref{ass:uniform}.(i),
	the term tends to zero in probability for each fixed $\rho\in\Gamma$ and consequently, $\sup_{\rho\in\Gamma}\left|\mathcal{M}_{2}(\rho)\right|\stackrel{p}{\rightarrow}0$. 
	
	From Assumption \ref{ass:uniform}.(ii), it can be directly verified that the above result follows if we replace $\{Y_{t,T}\}$, $\{Z_{t,T}\}$ and $K_{ut}$ in the above with the simulated counterparts $\{\tilde{y}_{u,t}\}$, $\{\tilde{z}_{u,t}\}$ and $1$, respectively (and for any $\theta\in\Theta$). From this we conclude, with obvious notations for this simulated counterpart, for any fixed $\theta\in\Theta$ and any fixed $u\in\mathcal{U}$,
		$$
		\tilde{M}_T(\rho,\theta)-\tilde{M}_T(\rho_0;\theta)={E}\left[\tilde{g}(\rho,\theta)-\tilde{g}(\rho_0,\theta)\right]+o_p(1).
		$$where $\tilde{g}(\rho,\theta)=g[\tilde{y}_{u,t};f(\tilde{z}_{u,t},\rho)]$. Moreover, due to Assumption \ref{ass:auxiliary}.(iv), the right hand side of the above satisfies, uniformly in $\theta$, $\sup_{\rho\in\Gamma}|{E}\left[\tilde{g}(\rho,\theta)-\tilde{g}(\rho_0,\theta)\right]|>0$, so that, by Assumption \ref{ass:auxiliary}.(iv) applied to the simulated data, we can conclude that the right hand side is uniquely minimized at $\rho_0(\cdot;\theta)$. The above pointwise convergence, the continuity of $\tilde{M}_T(\rho,\theta)$ in $\rho$, and the compactness of $\Theta$ and $\Gamma$, allows us to conclude, via the usual equicontinuity arguments (Assumption \ref{ass:auxiliary}.(v)), that $$\sup_{\theta\in\Theta}|\tilde{M}_T(\rho,\theta)-{E}\left[\tilde{g}(\rho,\theta)\right]|=o_p(1).$$
		
		Now, using	continuity of $\theta\mapsto\rho_0(\cdot;\theta)$, Assumption \ref{ass:auxiliary}.(vi), conclude that, for any $\delta>0$ there exists some $\varepsilon>0$ such that
		$$\sup_{\theta\in\Theta}\|\rho-\rho_0(u;\theta)\|\geq\delta\implies \sup_{\theta\in\Theta}|{E}\left[\tilde{g}(\rho,\theta)-\tilde{g}(\rho_0,\theta)\right]|>\varepsilon.$$ The remainder of the result follows the same lines as Theorem 5.7 in \cite{van2000asymptotic}, and hence is omitted.
\end{proof}

In what follows, we provide Lemma \ref{Lem:uniform_u} and its proof. For the proof of Lemma  \ref{Lem:uniform_u}, we need Lemma \ref{lem:uniformity}.

\begin{lemma}\label{lem:uniformity}
	Let $\{W_{t,T}\}$ be a triangular array such that 
	\[
	EW_{t,T}=0
	\]
	with $\left|W_{t,T}\right|\le d$ and $E\left|W_{t,T}\right|\le\delta$
	and $EW_{t,T}^{2}\le D$. 
	$\{W_{t,T}\}$ are also $\phi$-mixing and we denote $\phi(k)$
	as the $\phi$-mixing coefficient such that $\tilde{\phi}(m)=\sum_{j=1}^{m}\phi(j)$.
	Let there exist an increasing sequence $m_T:T\in \mathbb{N}$ of positive integers such that
	\begin{equation}
	\exists C<\infty: T\phi(m_T)/m_T\le C, 1\le m_T\le T, \forall T\in \mathbb{N}.\label{eq:m_T}
	\end{equation}
	Then, for any positive number $\epsilon$ and $c$, we have
	\[
	P\text{\ensuremath{\left(\left|\sum_{t=1}^{T}W_{t,T}\right|\ge\epsilon\right)}\ensuremath{\ensuremath{\le c_{1}\exp\left(-c\epsilon+c^2c_{2}T\right)}}}
	\]
	where $\phi(m_T)\rightarrow0$ as $m_T\rightarrow0$, $c_1=2e^{\frac{3T}{m_T}e^{1/2}\phi(m)}$ and $c_2=6c^2[D+4\delta d\tilde{\phi}(m_T)]$.
\end{lemma}

\begin{proof}[Proof of Lemma \ref{lem:uniformity}]
	Define $S=\sum_{t=1}^{T}W_{t,T}$. Consider a number $n_{0}$ such
	that $2m(n_{0}-1)\le T\le2mn_{0}$ with $m=m_T$. For all $j=1,2$ and $k=1,...,n_{0}$,
	we consider $A_{j,k}=\sum_{t=t_{1}}^{t_{2}}W_{t,T}$ where $t_{1}=\inf[(2k+j-3)m+1,T]$
	and $t_{2}=\inf[t_{1}+m-1,T]$. Note that the size of block for $A_{j,k}$ is $m$. Then, 
	\begin{equation}
	S=B_{1,n_{0}}+B_{2,n_{0}}\label{eq:sum of A}
	\end{equation} where  
	$B_{j,k}=\sum_{t=1}^{k}A_{j,t}$ for $j=1,2$
	with $B_{j,0}=0$. By construction,
	for some constant $c$,
	\begin{align}
	E\exp\{cS\} & \le(E\exp\{2cB_{1,n_{0}}\}+E\exp\{2cB_{2,n_{0}}\})\big/2.\label{eq:decompse}
	\end{align}	
	From (\ref{eq:sum of A}), applying (20.28) in \citet[pp 171]{billingsley1968}, we have
	\begin{flalign}
	E\exp\{2cB_{j,k}\}= & E\exp\{2cB_{j,k-1}\}\exp\{2cA_{j,k}\}.\nonumber \\
	\le & E\exp\{2cB_{j,k-1}\}E\exp\{2cA_{j,k}\}+2E\exp\{2cB_{j,k-1}\}\|\exp\{2cA_{j,k}\}\|_{\infty}\phi(m)\label{eq:billingsley}
	\end{flalign}
	Setting $cmd=1/4$ yields
	\begin{equation}
	\left|2cA_{j,k}\right|\le2cmd=\frac{1}{2}\label{eq:abs A_inequal}
	\end{equation}
	This implies that since $e^{x}\le1+x+x^{2}$ for $|x|\le1/2$, 
	\[
	\exp\{2cA_{j,k}\}\le1+2cA_{j,k}+4c^{2}A_{j,k}^{2}.
	\]
	Moreover, from $1+x\le e^{x}$, $1+4c^{2}EA_{j,k}^{2}\le e^{4c^{2}EA_{j,k}^{2}}$.
	Combining the above two inequalities,
	\begin{equation}
	Ee^{2cA_{j,k}}\le e^{4c^{2}EA_{j,k}^{2}}.\label{eq:inequality_A1}
	\end{equation}
	From the definition of $A_{j,k}$,
	\begin{align*}
	EA_{j,k}^{2}= & \sum_{t=t_{1}}^{t_{2}}EW_{t,T}^{2}+\sum_{t=t_{1}}^{t_{2}}\sum_{s=t_{1},s\neq t}^{t_{2}}EW_{t,T}W_{s,T}\\
	\le & m[D+4\delta d\tilde{\phi}(m)].
	\end{align*}
	where the inequality comes from $\left|EW_{t,T}W_{s,T}\right|\le2\delta d\phi(|t-s|)$.
	With this and (\ref{eq:inequality_A1}), 
	\[
	Ee^{2cA_{j,k}}\le e^{4c^{2}EA_{j,k}^{2}}\le e^{4c^{2}mC}
	\]
	where $C=[D+4\delta d\tilde{\phi}(m)]$. In combination with (\ref{eq:billingsley})
	and (\ref{eq:abs A_inequal}), the inequality leads to 
	\begin{flalign*}
	Ee^{2cB_{j,k}} & \le[e^{4c^{2}mC}+2e^{1/2}\phi(m)]Ee^{2cB_{j,k-1}}\\
	& =e^{4c^{2}mC}[1+2e^{1/2-4c^{2}mC}\phi(m)]Ee^{2cB_{j,k-1}}\\
	& \le e^{4c^{2}mC}[1+2e^{1/2}\phi(m)]Ee^{2cB_{j,k-1}}.
	\end{flalign*}
	Iterating the same procedure yields 
	\[
	Ee^{2cB_{j,n_{0}}}\le e^{4c^{2}n_{0}mC}(1+2e^{1/2}\phi(m))^{n_{0}}
	\]
	Recalling that $n_{0}$ is chosen such that $2m(n_{0}-1)\le T\le2mn_{0}$,
	we set $n_{0}\le\frac{3T}{2m}$. From (\ref{eq:decompse}),
	\[
	E\exp\{cS\}\le c_{1}\exp\{c_2T\}
	\]
	where $c_{1}=[1+2e^{1/2}\phi(m)]^{\frac{3T}{2m}}=\exp\{\frac{3T}{2m}\log[1+2e^{1/2}\phi(m)]\}\le\exp\{\frac{3T}{m}e^{1/2}\phi(m)\}$ and $c_2=6c^2[D+4\delta d\tilde{\phi}(m)]$.
	This is due to the fact that $\forall x\ge0$, $\log(1+x)\le x$. Finally, due to Markov inequality,
	\[
	P(|S|>\epsilon)\le P(S>\epsilon)\le e^{-c\epsilon}Ee^{c|S|} \le 2e^{-c\epsilon}Ee^{cS}.
	\]
	This completes the proof. 
\end{proof}

\begin{lemma}
	\label{Lem:uniform_u}
	Under the Assumptions of Theorem \ref{thm0}, for some positive constants, $\epsilon$, $C_1$ and $C_2$, 
	\[
	P(\mathcal{R}_{T}>\epsilon)\le NC_1 e^{-C_2 \epsilon Th/m_T}.
	\]	
	where 
	\[
	\mathcal{R}_{T}=\max_{1\le j\le N}\sup_{u\in\mathcal{U}}|\Psi_{T}(u,\rho_{0}(u)+\delta_{j})-E\Psi_{T}(u,\rho_{0}(u)+\delta_{j})|
	\]
	with $\Psi_{T}(u,\rho(u))=\sum_{t=1}^{T}w_{t}(u)g(Y_{t,T};f(Z_{t,T},\rho))$.
	
\end{lemma}

\begin{proof}[Proof of Lemma \ref{Lem:uniform_u}]
	Let $S(\delta_{j}):=\sum_{t=1}^{T}W_{t,T}(\delta_{j})=\Psi_{T}(u,\rho(u)+\delta_{j})-E\Psi_{T}(u,\rho(u)+\delta_{j})$ where $|W_{t,T}(\delta_{j})|\le d_j$.
	Under the Assumptions of Theorem \ref{thm0}, $g(\cdot)$ is bounded and the Kernel function satisfies boundedness and Lipschitz continuity. Due to Assumption \ref{ass:uniform}, there exists $m_T$ satisfying \eqref{eq:m_T}. Setting a constant $c$ proportional to $Th/m_T$, applying Lemma \ref{lem:uniformity} yields that, for some finite positive constants $C_1$ and $C_2$,
	\begin{equation}
	\sup_{\delta_{j},j=1,...,N}P\left(\left|S(\delta_{j})\right|>\epsilon\right)\le \sup_{\delta_{j},j=1,...,N}C_{1}e^{-C_{2}(\epsilon-\kappa(d_j,\tilde{\phi}(m_T)/m_T))Th/m_{T}}\le C_{1}e^{-C_{2}\epsilon Th/m_T}\label{eq:sup_p_S}
	\end{equation}
	where $\kappa(d_j,\tilde{\phi}(m_T)/m_T)$ is proportional to $c_2$ in Lemma \ref{lem:uniformity}.
	
	Note that
	\[
	P(\max_{1\le j\le N}\left|S(\delta_{j})\right|>\epsilon)\le\sum_{1\le j\le N}P\left(\left|S(\delta_{j})\right|>\epsilon\right)\le N\sup_{\delta_{j},j=1,...,N}P\left(\left|S(\delta_{j})\right|>\epsilon\right).
	\]
	From (\ref{eq:sup_p_S}), 
	\[
	P(\max_{1\le j\le N}\left|S(\delta_{j})\right|>\epsilon)\le NC_{1}e^{-C_{2}\epsilon Th/m_T},
	\]
	which completes the proof.
\end{proof}

\subsection{Proof of Corollary \ref{Lem:uniform_theta}}
\begin{proof}
By construction, $\sup_{u\in\mathcal{U};|u-t/T|\le T^{-1}}|f(Z_{t,T};\rho_0(t/T))-f(Z_{t,T};\rho_0(u))|=O(T^{-1})$ and therefore 
$Y_{t,T}=f(Z_{t,T};\rho(t/T))+\eta_{t}=f(Z_{t,T};\rho(u))+\eta_{t}+O(T^{-1})$. In what follows, $O(T^{-1})$ is suppressed. 

Define $p_{t}(\rho)=f(Z_{t,T};\rho)-f(Z_{t,T};\rho_{0})$
for a given $u\in\mathcal{U}$ where $\rho:=\rho(u) \in \Gamma $ and $\rho_{0}:=\rho_{0}(u)$.
Then, we have
\[
M_{T}(\rho)-M_{T}(\rho_{0})=\mathcal{M}_{1}(\rho)+\mathcal{M}_{2}(\rho)
\]
where, for a given $u\in \mathcal{U}$ such that $|u-t/T|\le T^{-1}$,
\begin{flalign*}
\mathcal{M}_{1}(\rho)= & E\left\{ f(Z_{t,T};\rho(u))-f(Z_{t,T};\rho_0(u)\right\} ^{2}=E[p_{t}(\rho)]^{2}\\
\mathcal{\mathcal{M}}_{2}(\rho)= & \mathcal{M}_{1}(\rho)\left(\frac{1}{Th}\sum_{t=1}^{T}K_{ut}-1\right)+\frac{1}{Th}\sum_{t=1}^{T}K_{ut}\left\{ p_{t}^{2}(\rho)-\mathcal{M}_{1}(\rho)\right\} \\
& -\frac{2}{Th}\sum_{t=1}^{T}K_{ut}\eta_{t}p_{t}(\rho).
\end{flalign*}		

Once noting that the absolute summability implies the square summability, 
everything else is analogous to the proof of Theorem \ref{thm0}. This completes the proof.
\end{proof}

\subsection{Proof of Theorem \ref{thm1}}
\begin{proof} The proof is similar to others found in the literature on semiparametric estimation, see. e.g., \cite{chen2003estimation} (pg 1604), and in particular is similar to Lemma 1 in \cite{frazier2019simple} (pg, 136-137). 

From the definitions of $\rho_0(u)$ and $\rho_0(u,;\theta)$, and the injectivity and continuity of $\rho_0(\cdot;\theta)$, for all $\delta>0,$ there exists some $\epsilon>0$ such that, if $\sup_{u}\|\theta-\theta_0(u)\|\geq\delta$, then $$\sup_{u}\{Q_0[u,\theta_0(u)]-Q_0[u,\theta]\}\geq\epsilon.$$ Applying this fact we see that 
\begin{flalign}\label{new1}
\text{P}\left(\sup_{u}\|\hat{\theta}(u)-\theta_0(u)\|\geq\delta\right)\leq\text{P}\left(\sup_{u}\{Q_0[u,\theta_0(u)]-Q_0[u,\hat{\theta}(u)]\}\geq\epsilon\right)
\end{flalign}and the results follows if the right hand side of the above is $o_{p}(1)$. 

To this end, first note that, by the definitions of $Q_{T}(u,\theta)$ and $Q_{0}(u,\theta)$,
\begin{flalign*}
\sup_{u\in\mathcal{U}}\sup_{\theta\in\Theta}\big{|}Q_{T}(u,\theta)-Q_{0}(u;\theta)\big{|}&=\sup_{u\in\mathcal{U}}\sup_{\theta\in\Theta}\big{|}- \|\hat{\rho}(u)-\hat{\rho}(u;\theta)\|+\|{\rho}_{0}(u)-{\rho}_{0}(u;\theta)\|\big{|}\\&\leq \sup_{u\in\mathcal{U}}\sup_{\theta\in\Theta}\| \hat{\rho}(u)-\hat{\rho}(u;\theta) -{\rho}_{0}(u)+{\rho}_{0}(u;\theta)\|\\&\leq{\sup_{u\in\mathcal{U}}\| \hat{\rho}(u) -{\rho}_{0}(u)\|}+\sup_{u\in\mathcal{U}}\sup_{\theta\in\Theta}\|{\rho}_{0}(u;\theta)-\hat{\rho}(u;\theta)\|,
\end{flalign*}where the second inequality follows from the reverse triangle inequality and the third from the regular triangle inequality. The uniform convergence now follows from the results in Theorem \ref{thm0}. 

Now, we show that for any $\tau>0$
	$$\lim_{T\rightarrow\infty}\text{P}\left(\sup_{u\in\mathcal{U}}\left\{Q_{0}[u,{\theta}_{0}(u)]-Q_{0}[u,\hat{\theta}(u)]\right\}<\tau\right)=1.$$ 
	From the definition of $\hat{\theta}(u)$, for every $u\in\mathcal{U}$, 
	\begin{equation*}
	Q_{T}[u,\hat{\theta}(u)]\geq Q_{T}[u,\theta_{0}(u)],
	\end{equation*}and 
	\begin{equation}\label{one}
	\sup_{u\in\mathcal{U}}\left\{ Q_{T}[u,\theta_{0}(u)]-Q_{T}[u,\hat{\theta}(u)]\right\}\leq0. 
	\end{equation}Moreover, by uniform convergence of $Q_{T}[u,\theta]$ to $Q_{0}[u,\theta]$ we have, \begin{flalign}
	\lim_{T\rightarrow\infty}\text{P}\bigg{(}\sup_{u\in\mathcal{U}}&\left\{ Q_{T}[u,\hat{\theta}(u)]-Q_{0}[u,\hat{\theta}(u)]\right\}<\tau/2\bigg{)}=1\label{two}\\\lim_{T\rightarrow\infty}\text{P}\bigg{(}\sup_{u\in\mathcal{U}}&\left\{ Q_{0}[u,{\theta}_{0}(u)]-Q_{T}[u,{\theta}_{0}(u)]\right\}<\tau/2\bigg{)}=1\label{three}
	\end{flalign}
	Now, consider
	\begin{flalign*}
	\sup_{u}&\left\{Q_{0}[u,\theta_{0}(u)]-Q_{0}[u,\hat{\theta}(u)]\right\}=\sup_{u}\left\{Q_{0}[u,\theta_{0}(u)]-Q_{0}[u,\hat{\theta}(u)]+Q_{T}[u,{\theta}_{0}(u)]-Q_{T}[u,{\theta}_{0}(u)]\right\}\\\leq&\sup_{u}\left\{Q_{0}[u,\theta_{0}(u)]-Q_{T}[u,{\theta}_{0}(u)]\right\}+\sup_{u}\left\{Q_{T}[u,\theta_{0}(u)]-Q_{0}[u,\hat{\theta}(u)]+Q_{T}[u,\hat{\theta}(u)]-Q_{T}[u,\hat{\theta}(u)]\right\}\\\leq& \sup_{u}\left\{Q_{0}[u,\theta_{0}(u)]-Q_{T}[u,{\theta}_{0}(u)]\right\}+\sup_{u}\left\{Q_{T}[u,\theta_{0}(u)]-Q_{T}[u,\hat{\theta}(u)]\right\}+\sup_{u}\left\{Q_{T}[u,\hat{\theta}(u)]-Q_{0}[u,\hat{\theta}(u)]\right\}\\\leq& \sup_{u}\left\{Q_{0}[u,\theta_{0}(u)]-Q_{T}[u,{\theta}_{0}(u)]\right\}+\sup_{u}\left\{Q_{T}[u,\hat{\theta}(u)]-Q_{0}[u,\hat{\theta}(u)]\right\}
	\end{flalign*}  where the last inequality comes from equation \eqref{one}. Therefore, from the uniform convergence in \eqref{two} and \eqref{three},
	\begin{equation}\label{goal1}
	\lim_{T\rightarrow\infty}\text{P}\bigg{(}\sup_{u\in\mathcal{U}}\left\{Q_{0}[u,\theta_{0}(u)]-Q_{0}[u,\hat{\theta}(u)]\right\}<\tau\bigg{)}=1.
	\end{equation} 
	
	The result then follows by taking $\tau=\epsilon$ in \eqref{new1}.

\end{proof}

\subsection{Proof of Theorem \ref{thm2}}
We break the proof down into two parts: first, we derive the asymptotic expansion of the estimating equations based on the observed estimator and derives the order of these expansions; we then use this result to deduce the stated result.

\medskip 

\noindent\textbf{Part 1:} To simplify notation, in what follows we take $q(Y_{t,T},\rho(u))=q[Y_{t,T};f(Z_{t,T},\rho(u))]$. By the definition of $\hat{\rho}(u)$,
\begin{flalign*}
0=&\frac{1}{Th}\sum_{t=1}^{T}q(Y_{t,T},\hat{\rho}(u))K\left(\frac{u-t/T}{h} \right)\\=&\frac{1}{Th}\sum_{t=1}^{T}q(Y_{t,T}, {\rho}_{0}(u))K\left(\frac{u-t/T}{h} \right)+\left\{\frac{1}{Th}\sum_{t=1}^{T}\frac{ \partial q(Y_{t,T}, \rho_{0}(u))}{\partial\rho'}K\left(\frac{u-t/T}{h} \right)\right\}(\hat{\rho}(u)-\rho_{0}(u))\\&+O_{p}(\|\hat{\rho}(u)-\rho_{0}(u)\|^2),
\end{flalign*}
where $\partial q(x_0)\big{/}\partial x:=\partial q(x)\big{/}\partial x\big{|}_{x=x_0}$. It can be rewritten as 
\begin{flalign}
0= & \underbrace{\frac{1}{Th}\sum_{t=1}^{T}q(Y_{t,T},\rho_{0}(u))K\left(\frac{u-t/T}{h}\right)}_{\eqref{tse1}.1}+\underbrace{\frac{\partial \Psi_{0}(\rho_{0}(u);u)}{\partial\rho'}(\hat{\rho}(u)-\rho_{0}(u))}_{\eqref{tse1}.2}+O_{p}(\|\hat{\rho}(u)-\rho_{0}(u)\|^{2})\nonumber\\
& +\underbrace{\left\{ \frac{1}{Th}\sum_{t=1}^{T}\frac{\partial q(Y_{t,T},\rho_{0}(u))}{\partial\rho'}K\left(\frac{u-t/T}{h}\right)-\frac{\partial \Psi_{0}(\rho_{0}(u);u)}{\partial\rho'}\right\} }_{\eqref{tse1}.3}(\hat{\rho}(u)-\rho_{0}(u))\label{tse1}
\end{flalign}
where $\partial \Psi_{0}(\rho_{0}(u);u)\big/\partial\rho'=\lim_{T\rightarrow\infty}(Th)^{-1}\sum_{t=1}^{T}E\left[\partial q(Y_{t,T},\rho_{0}(u))/\partial\rho'\right]K((u-t/T)/h)$.

Firstly, the term, \eqref{tse1}.3 is $o_p(1)$ given that  for each $u\in\mathcal{U}$ and $t\in\mathbb{N}\le T$ such that $|u-t/T|<T^{-1}$,
\begin{equation}
\frac{1}{Th}\sum_{t=1}^{T}\frac{\partial q(Y_{t,T},\rho_{0}(u))}{\partial\rho'}K\left(\frac{u-t/T}{h}\right)=\frac{\partial\Psi_{0}(\rho_{0}(u);u)}{\partial\rho'}+o_{p}(1),
\label{eq:localmatch}
\end{equation}
due to local stationarity of ${Y_{t,T}}$ and Assumptions \ref{ass:uniform}-\ref{ass:anormal}. 
To see this, for each $u_0=t_0/T$,
	\begin{flalign}
	\frac{1}{Th}\sum_{t=1}^{T}K\left(\frac{u_0-t/T}{h}\right)Z_{t,T} &= \frac{1}{Th} \sum_{k=-M}^{M}K\left(\frac{k}{Th}\right)Z_{k-t_0,T}\nonumber\\
	&=\underbrace{\frac{1}{Th}\sum_{k=-M}^{0}K\left(\frac{k}{Th}\right)Z_{k-t_0,T}}_{\eqref{eq:local_mom}.1}+\underbrace{\frac{1}{Th}\sum_{k=1}^{M}K\left(\frac{k}{Th}\right)Z_{k-t_0,T}}_{\eqref{eq:local_mom}.2}
	\label{eq:local_mom}
	\end{flalign}
	where $Z_{t,T}=\left[\frac{\partial q(Y_{t,T},\rho_{0}(u))}{\partial\rho'}-E\left[\frac{\partial q(Y_{t,T},\rho_{0}(u)))}{\partial\rho'}\right]\right]$, and $M=ThL$ with $L$ being the bound of support of a Kernel function as in Assumption \ref{ass:AAker}. Regarding \eqref{eq:local_mom}.2,
	\begin{flalign*}
	\left|\frac{1}{Th}\sum_{k=1}^{M}K\left(\frac{k}{Th}\right)Z_{k-t_0,T}\right|=&\left|\frac{1}{Th}\sum_{k=1}^{M-1}\left[K\left(\frac{k}{Th}\right)-K\left(\frac{k+1}{Th}\right)\right]\mathcal{S}_k
	+\frac{1}{Th}K\left(\frac{M}{Th}\right)\mathcal{S}_M\right|\\
	\le & \frac{C}{Th}\sup_{k\le M} |\mathcal{S}_k| \stackrel{p}{\rightarrow} 0
	\end{flalign*}
	where $\mathcal{S}_k=\sum_{i=1}^{k} Z_{i-t_0,T}$ and a generic constant $C$. The first equality comes from summation by parts and the inequality is due to Assumption \ref{ass:AAker}, i.e. $K(\cdot)$ is of bounded variation. The convergence to zero in probability is ensured by the ergodic theorem. Applying the same argument to the term in equation \eqref{eq:local_mom}.1., the result in equation \eqref{eq:localmatch} follows.
	It is worth noting that similar arguments to \eqref{eq:localmatch} are used under various modeling set-ups. For instance, see Lemma A.5 in \citet{DSR06}, {Lemma A.1} in \citet{fryzlewicz2008normalized} and the Proof of Theorem 2 in \citet{KL12}.

From Assumption \ref{ass:anormal}.1, the term \eqref{tse1}.1 in \eqref{tse1} satisfies,
\begin{flalign*}
\frac{1}{Th}\sum_{t=1}^{T}q(Y_{t,T}, {\rho}_{0}(u))K\left(\frac{u-t/T}{h} \right)&=O_p(1/\sqrt{Th})
\end{flalign*}
and we can rearrange terms in equation \eqref{tse1} to obtain the result: for $|u-t/T|<T^{-1}$, apply Assumption \ref{ass:anormal}.5 to obtain
\begin{flalign}
(\hat{\rho}(u)-\rho_{0}(u))=-\left\{\frac{\partial \Psi_{0}(\rho_{0}(u);u)}{\partial\rho'}\right\}^{-1}\Psi_T(\rho_0;u)+O_p(T^{-1})\label{eq:rho_exp},
\end{flalign}
where $\Psi_T(\rho_0;u)=\frac{1}{Th}\sum_{t=1}^{T}q(Y_{t,T}, {\rho}_{0}(u))K\left(\frac{u-t/T}{h} \right)$. 

\medskip

\noindent\textbf{Part 2:} We now use the above expansion to deduce the asymptotic distribution of the L-II estimator. 

From the definition of $\hat{\theta}:=\hat{\theta}(u)$,
\begin{flalign*}
0=\frac{\partial \hat{\rho}(u;\hat{\theta})^{\prime}}{\partial\theta}\Omega (\hat{\rho}(u)-\hat{\rho}(u;\hat{\theta}))
\end{flalign*}
Note that 
\begin{flalign*}
(\hat{\rho}(u)-\hat{\rho}(u;\hat{\theta}))&=[\hat{\rho}(u)-\rho_{0}(u)]-[\hat{\rho}(u;\hat{\theta})-\rho_{0}(u)]
\end{flalign*}and
\begin{flalign}
[\hat{\rho}(u;\hat{\theta})-\rho_{0}(u)]&=[\hat{\rho}(u;\hat{\theta})-\hat{\rho}(u;\theta_0(u))]+[\hat{\rho}(u;\theta_0(u))-\rho_{0}(u)]]\nonumber\\
&=\frac{\partial \rho(u,\theta_{0}(u))}{\partial\theta'}(\hat{\theta}(u)-\theta_{0}(u))+[\rho_{0}(u)-\hat{\rho}(u,\theta_{0}(u))]+O_p(T^{-1/2})\label{th_tse1}
\end{flalign}Using equation \eqref{th_tse1} within the FOCs, and the consistency of $\hat{\rho}(u)$ and $\hat{\theta}(u)$ obtained in Theorem \ref{thm0} and Theorem \ref{thm1} respectively, we obtain
\begin{flalign*}
0=\frac{\partial \rho(u,\theta_{0}(u))^{\prime}}{\partial\theta}\Omega\bigg{\{ }
(\hat{\rho}(u)-\rho_{0}(u))-\frac{\partial \rho(u,\theta_{0}(u))}{\partial\theta'}(\hat{\theta}(u)-\theta_{0}(u))+[\rho_{0}(u)-\rho(u,\theta_{0}(u))]+O_p(T^{-1/2})
\bigg{ \} }
\end{flalign*}which implies
\begin{flalign}\label{tse1_theta}
\hat{\theta}(u)-\theta_{0}(u)=\left\{\frac{\partial \rho(u,\theta_{0}(u))^{\prime}}{\partial\theta}\Omega\frac{\partial \rho(u,\theta_{0}(u))}{\partial\theta'}\right\}^{-1}\frac{\partial \rho(u,\theta_{0}(u))^{\prime}}{\partial\theta}\Omega\left\{(\hat{\rho}(u)-\rho_{0}(u))+O_p(T^{-1/2})\right\},
\end{flalign}where we have used the injectivity of $\rho(u,\theta)$ in $\theta$. The result now follows by substituting in the expansion for $\{\hat\rho(u)-\rho_0(u)\}$ given in \eqref{eq:rho_exp} and multiplying by $\sqrt{Th}$.
\clearpage

\section{Figures and tables}
\label{app: tf}

\subsection{Figures}
\begin{figure}[H]
	\centering 
	\includegraphics[height=0.4\textheight, width=0.8\linewidth]{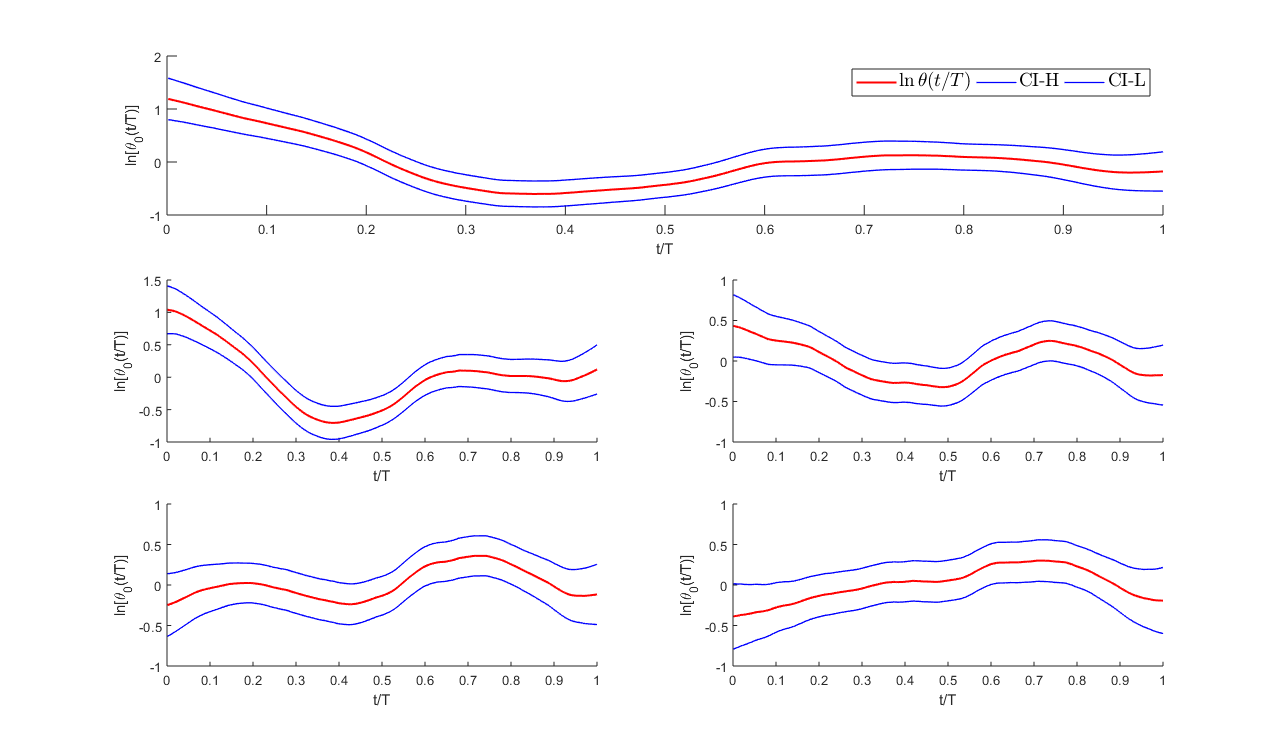}
	\caption{First five Fama-French portfolios: First quintile of size, intersected with the five quintiles of book-to-market. The top most figure in the panel represents the results for the portfolio formed from the intersection of the first quintile of size, and the first quintile of book-to-market. Running from right to left, the results then correspond to the first quintile of size and the second through fifth quintiles of book-to-market. CI-H and CI-L, represent pointwise confidence bands and are calculated using the local block bootstrap (LBB) at each time point $u=t/T$. For presentation of the results, the time span represented on the $x$-axis has been placed on the unit interval. The following correspondence can be used to aid interpretation of the results. The first time point in the sample corresponds to January 1952, the time point $u=0.5$ corresponds to June 1985, and the time point $u=.99$ corresponds to December 2018.}
	\label{FF_fig1}
\end{figure}

\clearpage

\begin{figure}[H]
	\centering 
		\includegraphics[height=0.4\textheight, width=0.8\linewidth]{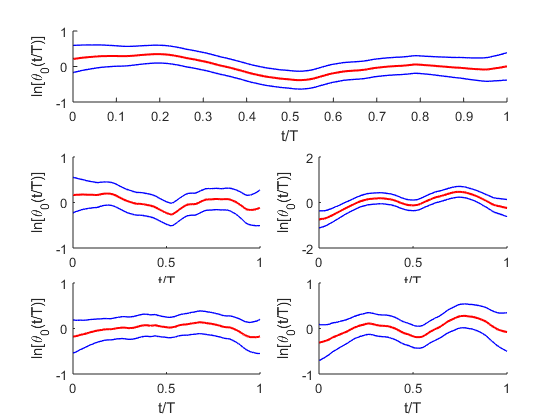}
	\caption{Second five Fama-French portfolios: Second quintile of size, intersected with the five quintiles of book-to-market. The figure has the same interpretation as Figure \ref{FF_fig1} but the results correspond to the second quintile of size.}
	\label{FF_fig2}
\end{figure}
\begin{figure}[H]
	\centering 
		\includegraphics[height=0.4\textheight, width=0.8\linewidth]{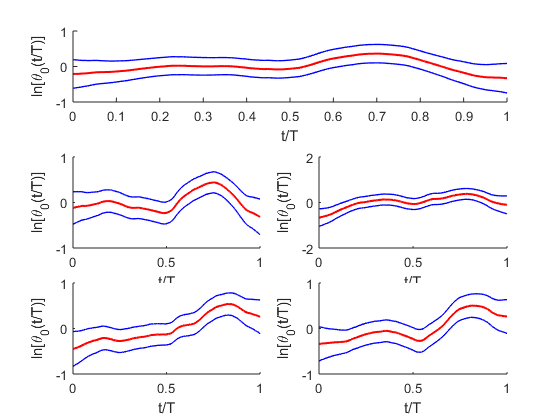}
	\caption{Third five Fama-French portfolios: Third quintile of size, intersected with the five quintiles of book-to-market. The figure has the same interpretation as Figure \ref{FF_fig1} but the results correspond to the third quintile of size.}
	\label{FF_fig3}
\end{figure}
\begin{figure}[H]
	\centering 
		\includegraphics[height=0.4\textheight, width=0.8\linewidth]{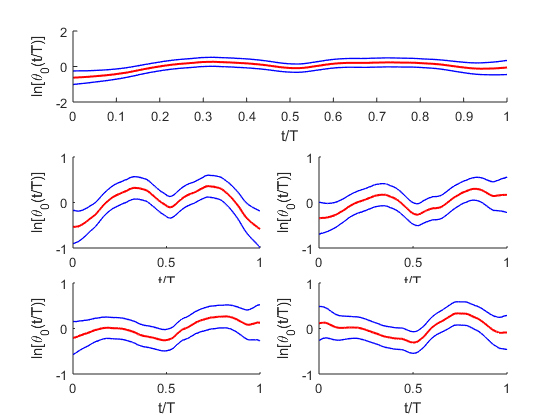}
	\caption{Fourth five Fama-French portfolios: Fourth quintile of size, intersected with the five quintiles of book-to-market. The figure has the same interpretation as Figure \ref{FF_fig1} but the results correspond to the fourth quintile of size.}
	\label{FF_fig4}
\end{figure}
\begin{figure}[H]
	\centering 
		\includegraphics[height=0.4\textheight, width=0.8\linewidth]{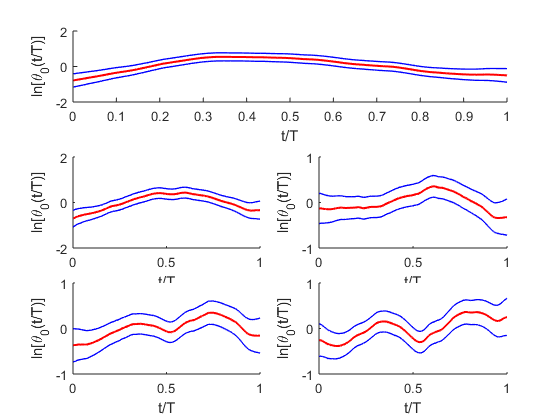}
	\caption{Last five Fama-French portfolios: Fifth quintile of size, intersected with the five quintiles of book-to-market. The figure has the same interpretation as Figure \ref{FF_fig1} but the results correspond to the fifth quintile of size.}
	\label{FF_fig5}
\end{figure}

\newpage

\subsection{Tables}

\begin{landscape}
	\begin{table}[htbp]
		\centering
		\caption{Estimated $\alpha$ and $\beta$ parameters from the modified three-factor Fama-French model. For each column, Est denotes the point estimator and SE the standard error. Standard errors are calculated using the LBB.}
		\begin{tabular}{rrrrrrrrrrrrrrrrrr}
			&       &       &       &       &       &       &       &       &       &       &       &       &       &       &       &       &  \\\hline
			& \multicolumn{1}{l}{} &       & \multicolumn{1}{l}{Size-1} &       &       & \multicolumn{1}{l}{Size-2} &       &       & \multicolumn{1}{l}{Size-3} &       &       & \multicolumn{1}{l}{Size-4} &       &       & \multicolumn{1}{l}{Size-5} &       &  \\\hline		
			&       &                            &  Est     &  SE   &       &   Est    &  SE   &       & Est      &   SE &       &   Est    &  SE   &       &   Est    & SE  & \\\hline
			
			&       & \multicolumn{1}{l}{$\alpha$} &        -0.544 & 0.181 &       & -0.056 & 0.108 &       & -0.054 & 0.075 &       & 0.229 & 0.079 &       & 0.285 & 0.080 &  \\			
			&       & \multicolumn{1}{l}{$\beta_1$} &    1.103 & 0.030 &       & 1.003 & 0.022 &       & 0.929 & 0.014 &       & 0.888 & 0.014 &       & 0.958 & 0.016 &  \\
			& \multicolumn{1}{l}{BM-1} & \multicolumn{1}{l}{$\beta_2$} &    1.388 & 0.044 &       & 1.298 & 0.059 &       & 1.083 & 0.024 &       & 1.054 & 0.025 &       & 1.086 & 0.033 &  \\
			&       & \multicolumn{1}{l}{$\beta_3$} &    -0.220 & 0.052 &       & 0.079 & 0.044 &       & 0.310 & 0.027 &       & 0.460 & 0.025 &       & 0.691 & 0.028 &  \\
			&       & \multicolumn{1}{l}{$\beta_4$} &    1.118 & 0.391 &       & 1.044 & 0.203 &       & 1.035 & 0.174 &       & 0.836 & 0.172 &       & 0.558 & 0.163 &  \\
						&       &       &       &       &       &       &       &       &       &       &       &       &       &       &       &       &  \\
			&       & \multicolumn{1}{l}{$\alpha$} &    -0.256 & 0.100 &       & -0.026 & 0.085 &       & 0.127 & 0.070 &       & -0.043 & 0.074 &       & -0.032 & 0.075 &  \\
			&       & \multicolumn{1}{l}{$\beta_1$} &    1.126 & 0.020 &       & 1.014 & 0.016 &       & 0.956 & 0.017 &       & 0.964 & 0.013 &       & 1.072 & 0.014 &  \\
			& \multicolumn{1}{l}{BM-2} & \multicolumn{1}{l}{$\beta_2$} &   1.012 & 0.035 &       & 0.886 & 0.042 &       & 0.769 & 0.056 &       & 0.722 & 0.036 &       & 0.884 & 0.023 &  \\
			&       & \multicolumn{1}{l}{$\beta_3$} &    -0.361 & 0.038 &       & 0.112 & 0.051 &       & 0.368 & 0.055 &       & 0.562 & 0.034 &       & 0.773 & 0.030 &  \\
			&       & \multicolumn{1}{l}{$\beta_4$} &    1.130 & 0.215 &       & 1.099 & 0.167 &       & 0.850 & 0.156 &       & 1.276 & 0.176 &       & 1.071 & 0.154 &  \\
			&       &       &       &       &       &       &       &       &       &       &       &       &       &       &       &       &  \\
			&       & \multicolumn{1}{l}{$\alpha$} &    -0.109 & 0.074 &       & 0.064 & 0.081 &       & 0.050 & 0.081 &       & 0.163 & 0.078 &       & -0.075 & 0.120 &  \\
			&       & \multicolumn{1}{l}{$\beta_1$} &    1.097 & 0.018 &       & 1.016 & 0.016 &       & 0.983 & 0.021 &       & 0.980 & 0.017 &       & 1.067 & 0.023 &  \\
			& \multicolumn{1}{l}{BM-3} & \multicolumn{1}{l}{$\beta_2$} &   0.752 & 0.024 &       & 0.553 & 0.058 &       & 0.442 & 0.065 &       & 0.433 & 0.054 &       & 0.573 & 0.071 &  \\
			&       & \multicolumn{1}{l}{$\beta_3$} &   -0.426 & 0.024 &       & 0.156 & 0.060 &       & 0.407 & 0.060 &       & 0.597 & 0.064 &       & 0.795 & 0.054 &  \\
			&       & \multicolumn{1}{l}{$\beta_4$} &     1.142 & 0.156 &       & 1.035 & 0.181 &       & 0.832 & 0.159 &       & 0.746 & 0.159 &       & 1.223 & 0.227 &  \\
			&       &       &       &       &       &       &       &       &       &       &       &       &       &       &       &       &  \\
			&       & \multicolumn{1}{l}{$\alpha$} &    -0.012 & 0.068 &       & 0.060 & 0.088 &       & 0.049 & 0.102 &       & 0.086 & 0.074 &       & -0.235 & 0.117 &  \\
			&       & \multicolumn{1}{l}{$\beta_1$} &   1.078 & 0.015 &       & 1.042 & 0.022 &       & 1.030 & 0.020 &       & 1.005 & 0.017 &       & 1.152 & 0.029 &  \\
			& \multicolumn{1}{l}{BM-4} & \multicolumn{1}{l}{$\beta_2$} &   0.403 & 0.032 &       & 0.212 & 0.061 &       & 0.188 & 0.062 &       & 0.223 & 0.031 &       & 0.293 & 0.061 &  \\
			&       & \multicolumn{1}{l}{$\beta_3$} &     -0.395 & 0.027 &       & 0.171 & 0.063 &       & 0.405 & 0.073 &       & 0.558 & 0.051 &       & 0.795 & 0.046 &  \\
			&       & \multicolumn{1}{l}{$\beta_4$} &    1.253 & 0.181 &       & 0.717 & 0.204 &       & 0.874 & 0.230 &       & 0.865 & 0.169 &       & 1.269 & 0.251 &  \\
			&       &       &       &       &       &       &       &       &       &       &       &       &       &       &       &       &  \\
			&       & \multicolumn{1}{l}{$\alpha$} &     0.124 & 0.065 &       & -0.023 & 0.067 &       & 0.180 & 0.094 &       & -0.325 & 0.128 &       & -0.197 & 0.172 &  \\
			&       & \multicolumn{1}{l}{$\beta_1$} &   0.985 & 0.013 &       & 0.983 & 0.016 &       & 0.934 & 0.021 &       & 1.033 & 0.025 &       & 1.121 & 0.031 &  \\
			& \multicolumn{1}{l}{BM-5 } & \multicolumn{1}{l}{$\beta_2$} &    -0.240 & 0.019 &       & -0.204 & 0.035 &       & -0.246 & 0.042 &       & -0.211 & 0.031 &       & -0.094 & 0.043 &  \\
			&       & \multicolumn{1}{l}{$\beta_3$} &   -0.365 & 0.025 &       & 0.074 & 0.046 &       & 0.298 & 0.045 &       & 0.649 & 0.047 &       & 0.839 & 0.040 &  \\
			&       & \multicolumn{1}{l}{$\beta_4$} &   1.050 & 0.176 &       & 1.098 & 0.158 &       & 0.615 & 0.271 &       & 1.265 & 0.255 &       & 0.997 & 0.316 &  \\\hline\hline
		\end{tabular}%
		\label{tab:beta_parms}%
	\end{table}%
\end{landscape}
\clearpage

\begin{table}[h]
	\centering
	\caption{Short-term stochastic volatility parameter estimates across the 25 Fama-French portfolios. For each column, Est denotes the point estimator and SE the standard error. Standard errors are calculated using the LBB.}
	\begin{tabular}{rlrrrrrrrrrrrrrr}\hline
		\multicolumn{1}{l}{} &       & {Size-1} &       &       & {Size-2} &       &       & {Size-3} &       &       & {Size-4} &       &       & {Size-5} &  \\\hline&     &Est. & SE &       & Est & SE &       &Est  &SE  &       & Est & SE &       & Est  &SE  \\\hline
		& $\phi$   &     0.693 & 0.075 &       & 0.627 & 0.242 &       & 0.347 & 0.293 &       & 0.679 & 0.341 &       & 0.523 & 0.176 \\
		\multicolumn{1}{l}{BM-1 } & $\sigma_v$  &     1.902 & 0.327 &       & 1.690 & 0.622 &       & 1.709 & 0.317 &       & 1.317 & 0.655 &       & 1.573 & 0.363 \\
		& $\gamma_\nu$   &  0.259 & 0.154 &       & -0.162 & 0.372 &       & -0.513 & 0.285 &       & 0.153 & 0.330 &       & 0.235 & 0.348 \\
		&       &       &       &       &       &       &       &       &       &       &       &       &       &       &  \\
		& $\phi$    &      0.654 & 0.314 &       & 0.543 & 0.150 &       & 0.584 & 0.202 &       & 0.499 & 0.292 &       & 0.565 & 0.244 \\
		\multicolumn{1}{l}{BM-2 } & $\sigma_v$  &    1.602 & 0.484 &       & 1.742 & 0.310 &       & 1.694 & 0.480 &       & 1.560 & 0.408 &       & 1.656 & 0.507 \\
		& $\gamma_\nu$   &   0.123 & 0.315 &       & -0.262 & 0.242 &       & 0.007 & 0.503 &       & -0.011 & 0.179 &       & 0.320 & 0.507 \\
		&       &       &       &       &       &       &       &       &       &       &       &       &       &       &  \\
		& $\phi$    &   0.603 & 0.104 &       & 0.529 & 0.299 &       & 0.600 & 0.212 &       & 0.619 & 0.226 &       & 0.762 & 0.172 \\
		\multicolumn{1}{l}{BM-3 } & $\sigma_v$  &   1.879 & 0.373 &       & 1.832 & 0.465 &       & 1.703 & 0.620 &       & 1.598 & 0.453 &       & 1.416 & 0.364 \\
		& $\gamma_\nu$   &    0.270 & 0.217 &       & -0.018 & 0.257 &       & -0.127 & 0.626 &       & 0.165 & 0.360 &       & -0.101 & 0.214 \\
		&       &       &       &       &       &       &       &       &       &       &       &       &       &       &  \\
		& $\phi$    &   0.473 & 0.239 &       & 0.686 & 0.109 &       & 0.523 & 0.264 &       & 0.596 & 0.093 &       & 0.592 & 0.123 \\
		\multicolumn{1}{l}{BM-4 } & $\sigma_v$  &     1.668 & 0.547 &       & 1.609 & 0.432 &       & 1.692 & 0.467 &       & 1.580 & 0.215 &       & 1.783 & 0.139 \\
		& $\gamma_\nu$   &      -0.069 & 0.319 &       & -0.245 & 0.218 &       & 0.242 & 0.467 &       & -0.096 & 0.242 &       & -0.149 & 0.374 \\
		&       &       &       &       &       &       &       &       &       &       &       &       &       &       &  \\
		& $\phi$    &    0.572 & 0.274 &       & 0.427 & 0.237 &       & 0.592 & 0.303 &       & 0.428 & 0.302 &       & 0.530 & 0.249 \\
		\multicolumn{1}{l}{BM-5} & $\sigma_v$  &     1.507 & 0.618 &       & 1.879 & 0.396 &       & 1.668 & 0.584 &       & 1.623 & 0.619 &       & 1.747 & 0.414 \\
		& $\gamma_\nu$   &    0.192 & 0.416 &       & -0.095 & 0.503 &       & 0.428 & 0.429 &       & -0.452 & 0.271 &       & -0.157 & 0.483 \\\hline
	\end{tabular}%
	\label{tab:sv_parms}%
\end{table}%

\begin{table}[h]
	\centering
	\caption{99\% Confidence intervals for asymmetry parameter $\gamma$. For the entires in the table, (x,y) refers to the lower and upper level of the confidence interval, respectively, as calculated using QMLE robust standard errors. Across all 25 portfolios, only a single asymmetry parameter is statistically significant, which we mark in bold text. }
	\begin{tabular}{rrrrrrrr}
		&       &       &       & &       &  &  \\\hline
		&       & Size-1     & Size-2     & Size-3     & Size-4     & Size-5     &  \\\hline
		& BM-1     & \multicolumn{1}{l}{-0.0549,   0.1610} & \multicolumn{1}{l}{-0.0622,  0.1516} & \multicolumn{1}{l}{-0.0919,  0.1859} & \multicolumn{1}{l}{-0.1203,   0.0712} & \multicolumn{1}{l}{-0.0023,   0.1900} &  \\
		& BM-2     & \multicolumn{1}{l}{-0.1915,    0.0556} & \multicolumn{1}{l}{-0.2716,    0.0399} & \multicolumn{1}{l}{-0.0687,  0.0874} & \multicolumn{1}{l}{-0.0324,    0.1940} & \multicolumn{1}{l}{-0.1255,  0.0738} &  \\
		\multicolumn{1}{l}{} & BM-3     & \multicolumn{1}{l}{-0.2185,   0.0268} & \multicolumn{1}{l}{-0.0597,   0.1614} & \multicolumn{1}{l}{-0.0188,   0.2030} & \multicolumn{1}{l}{-0.0383,   0.2178} & \multicolumn{1}{l}{-0.1372,   0.0626} &  \\
		& BM-4     & \multicolumn{1}{l}{-0.1331,   0.0865} & \multicolumn{1}{l}{-0.0403,  0.0981} & \multicolumn{1}{l}{-0.0742,   0.1751} & \multicolumn{1}{l}{-0.1258,  0.0755} & \multicolumn{1}{l}{-0.0718,  0.1619} &  \\
		& BM-5     & \multicolumn{1}{l}{\textbf{0.0175,  0.1154}} & \multicolumn{1}{l}{-0.0762,  0.1446} & \multicolumn{1}{l}{-0.0558,  0.1325} & \multicolumn{1}{l}{-0.0072,  0.1698} & \multicolumn{1}{l}{-0.0931,   0.1709} &  \\\hline
		&       &       &       &       &       &       &  \\
	\end{tabular}%
	\label{tab:tarch}%
\end{table}%

\clearpage

\appendix
\begin{center}
	\textbf{\Large{{Supplementary Appendix}}}	
\end{center}

\section{Proof of Corollary \ref{corol1}}
Note that due to $\sup_{u \in \mathcal{U}}\vert \theta_0(u) \vert <1$ and local stationarity, there exists an invertible moving average process corresponding to the structural model \eqref{eq: ex1_str} in the vicinity of any given time point $u\in \mathcal{U}$. Therefore, there exists a Autoregressive process such that
$$y_{u,t}=\varepsilon_{u,t}+\sum_{s=1}^{\infty}(-\theta(u))^s y_{u,t-s},$$
in the neighborhood of a given time point $u\in\mathcal{U}$.
The auxiliary model AR(1) process is a misspecified version of the above model with a wrong order of lags.
Define $M_T(\cdot)$ and $M_0(\cdot)$ as  
$$M_{T}[\{Y_{t,T}\}_{t=0}^{T};\rho(u)]:=\frac{1}{Th}\sum_{t=1}^{T}(Y_{t,T}-\rho(t/T)Y_{t-1,T})^2K\left(\frac{u-t/T}{h}\right)$$
and $M_0(\cdot)=\lim_{T\rightarrow \infty}{E}M_T(\cdot)$.
$M_T(\cdot)$ and $M_0(\cdot)$ in this setting are well-defined and well behaved.	Furthermore, for any given time point $u\in \mathcal{U}$, the pseudo-true value, $\rho_{0}(u)$ can be represented by
$\rho_{0}(u)={\theta(u)}/(1+{\theta(u)}^2)$.
The minimizer $\rho_0(u)$ for any $u \in \mathcal{U}$ is continuous and strictly monotonic in $\theta(u)$. All of these ensure that the map $\theta\mapsto\rho_{0}(u;\theta)$ is continuous and injective in $\theta$ and hence Assumptions \ref{ass:auxiliary}.(vi) is met. Also, the assumption that $\sup_{u\in\mathcal{U}}|\rho_0(u)|<1$ implies the compactness of $\Theta$ and $\Gamma$ in conjunction with injectivity. With Lemma \ref{lemma1} that ensures Theorem \ref{thm0} and Corollary\ref{Lem:uniform_theta}, the proof of corollary \ref{corol1} follows directly from verification of the requisite regularity conditions stated in Theorem \ref{thm1}.
\begin{lemma}\label{lemma1}
	Under Assumptions \ref{ass:uniform} and \ref{ass:AAker}, the following are satisfied, 
	\begin{flalign}
	\sup_{u\in\mathcal{U}}&| \hat{\rho}(u)-\rho_{0}(u)|=o_p(1),\label{eq: uniform_nonpara}\\\sup_{\hat{\theta}\in[-1+\delta,1-\delta]}&| \hat{\rho}(u;\theta(u))-\rho_{0}(u;\theta(u))|= o_p\left(1\right) \ \ \ \text{a.s.}\label{eq: auxiliary_uniform}
	\end{flalign}
\end{lemma}
\begin{proof}	
	For (\ref{eq: uniform_nonpara}), under Assumptions \ref{ass:uniform}, and \ref{ass:AAker}, note that the following result follows straightforwardly from Theorem 2 in \citet{kristensen2009uniform}.
	\begin{equation*}
	\sup_{u\in\mathcal{U}}| \hat{\rho}(u)-\rho_{0}(u)|=O_p(h^2)+O_p\left(\sqrt{\frac{\ln T}{Th}}\right)=o_p(1)
	\end{equation*}
	
	For (\ref{eq: auxiliary_uniform}), recall that \begin{equation*}
	\hat{\rho}(u;\theta(u))={\hat{\psi}_1}(u;\theta(u))\small{/}{\hat{\psi}_2}(u;\theta(u)).
	\label{eq:limit_rho}
	\end{equation*}
	where $\hat{\psi}_1(u;\theta(u))=T^{-1}\sum_{t=1}^{T}\tilde{y}_{u,t-1}(u;\theta(u))\tilde{y}_{u,t}(u;\theta(u))$ and $\hat{\psi}_2(u;\theta(u))=T^{-1}\sum_{t=1}^{T}\tilde{y}_{u,t-1}^{2}(u;\theta(u))$.	
	For the sake of notation simplicity, we drop $(u)$ since it is clear that our argument is based on the fixed time point $u$.
	
	Due to the mean value theorem, 
	\begin{flalign*}
	\sup_{\theta\in \theta} \left|\hat{\rho}(\theta(u))-\rho_{0}(\theta(u))\right|
	&\le
	\sup_{\theta\in \theta} \left|\frac{\hat{\psi}_1(\theta)}{\hat{\psi}_2(\theta)}-\frac{{\psi}_1(\theta)}{{\psi}_2(\theta)}\right|\\
	&=
	\sup_{\theta\in \theta} \left|\frac{\hat{\psi}_1(\theta)-{\psi}_1(\theta)}{\bar{\psi}_2(\theta)}-\frac{\bar{\psi}_1(\theta)}{\bar{\psi}_2(\theta)}(\hat{\psi}_2(\theta)-\psi_2(\theta))\right|\\
	&\le
	\sup_{\theta\in \theta} \left|\frac{\hat{\psi}_1(\theta)-{\psi}_1(\theta)}{\bar{\psi}_2(\theta)}\right|+\sup_{\theta\in \theta}\left|\frac{\bar{\psi}_1(\theta)}{\bar{\psi}_2(\theta)}(\hat{\psi}_2(\theta)-\psi_2(\theta))\right|
	\end{flalign*}
	where $\theta=[-1+\delta,1-\delta]$ and $\bar{\psi}_k(\theta)\in[\hat{\psi}_k(\theta),{\psi}_k(\theta)]<\infty$ for $k=1,2$. 
	This implies that the uniform convergence rate for the left hand side is determined by $\vert \hat{\psi}_1(\theta)-{\psi}_1(\theta) \vert$ and $\vert \hat{\psi}_2(\theta)-\psi_2(\theta) \vert$ only.
	\begin{eqnarray}
	\sup_{\theta\in\theta}&| \hat{\psi}_1(\theta)-{\psi}_1(\theta)|\le
	\underbrace{\sup_{\theta\in\theta}| \hat{\psi}_1(\theta)-{E}\hat{\psi}_1(\theta)|}_{A.1}+
	\underbrace{\sup_{\theta\in\theta}| {E}\hat{\psi}_1(\theta)-{\psi}_1(\theta)|}_{A.2}
	\end{eqnarray}
	For A.2, $o_p(1)$ by construction.
	For A.1, we have to show that
	\begin{equation}
	\sup_{\theta\in\theta}| \hat{\psi}_1(\theta)-{E}\hat{\psi}_1(\theta)|\rightarrow0 \ \ \ a.s.\label{eq:A.1}
	\end{equation}
	
	The proof for (\ref{eq:A.1}) is organized as follows. Define $Z_t(\theta)=\tilde{y}_{u,t-1}(\theta)\tilde{y}_{u,t}(\theta)$. We replace $Z_t(\theta)$ with the truncated process $Z_t(\theta)\mathbb{I}(\vert Z_t(\theta) \vert \le \gamma_T)$ where $\mathbb{I}$ is the indicator function and $\gamma_T=\tau_T^{-1/(k-1)}$ such that $\tau_T=\sqrt{\ln T/T}$ for some $k>2$. Note that $\tau_T=o(1)$. Then, we replace the supremum in (\ref{eq:A.1}) with a maximization over a finite $N$ grids. Finally, we use the exponential inequality in Theorem 2.1. in \citet{liebscher1996strong} to bound the remainder. 
	
	First, consider truncation of $Z_t(\theta)$.
	\begin{eqnarray*}
		R_T(\theta)&=&\hat{\psi}_1(\theta)-\frac{1}{T}\sum_{t=1}^{T}Z_t(\theta)\mathbb{I}(\vert Y_t(\theta) \vert \le \gamma_T)\\
		&=&\frac{1}{T}\sum_{t=1}^{T}Z_t(\theta)\mathbb{I}(\vert Y_t(\theta) \vert > \gamma_T)
	\end{eqnarray*}
	where  $Z_t(\theta)=\tilde{y}_{u,t-1}(\theta)\tilde{y}_{u,t}(\theta)$.
	Then, 
	\begin{eqnarray}
	\vert {E}R_T(\theta)\vert&\le&{E}\left[\vert Y_t(\theta)\vert\mathbb{I}(\vert Z_t(\theta) \vert > \gamma_T)\right]\\
	&\le&{E}\left[\vert Z_t(\theta)\vert \vert Z^{k-1}_t(\theta) \gamma_T^{-(k-1)} \vert\mathbb{I}(\vert Z_t(\theta) \vert > \gamma_T)\right]\\
	&\le&\gamma_T^{-(k-1)}{E}\left[ \vert Z^{k}_t(\theta) \vert\right]
	\end{eqnarray}
	Due to Markov's inequality,
	$$
	\vert R_T(\theta)-{E}R_T(\theta) \vert =O_p(\gamma_T^{-(k-1)})=O_p(\tau_T).
	$$
	Therefore, we can focus on $Z_t(\theta)\mathbb{I}(\vert Z_t(\theta)\vert \le \gamma_T)$ since replacing $Z_t(\theta)$ with $Z_t(\theta)\mathbb{I}(\vert Z_t(\theta)\vert \le \gamma_T)$ incurs only an approximation error of order $O_p(\tau_T)$, which can be made arbitrarily small. In what follows, 
	$\vert Y_t(\theta)\vert \le \gamma_T$.
	
	Next, consider a set of grids or coverings of the form such that $B_j=\{\theta:\|\theta-\theta_j\|\le \tau_T \};j=1,...,N$. Since $\theta$ is compact, it can be covered by a finite number of $B_j$s for $j=1,...,N$ and $N\le c/\tau_T$. 
	Note that
	\begin{eqnarray*}
		\sup_{\theta\in \theta}| \hat{\psi}_1(\theta)-{E}\hat{\psi}_1(\theta)|&=&\max_{1\le j \le N}\sup_{\theta \cap B_j}\vert 
		\hat{\psi}_1(\theta)-{E}\hat{\psi}_1(\theta)\vert\\
		&\le&\max_{1\le j \le N}\sup_{\theta \cap B_j}\vert 
		\hat{\psi}_1(\theta)-\hat{\psi}_1(\theta_j)\vert\\
		&&+\max_{1\le j \le N}\vert 
		\hat{\psi}_1(\theta_j)-{E}\hat{\psi}_1(\theta_j)\vert\\
		&&+\max_{1\le j \le N}\sup_{\theta \cap B_j}\vert 
		{E}\hat{\psi}_1(\theta_j)-{E}\hat{\psi}_1(\theta)\vert\\
		&=&\mathcal{S}_1+\mathcal{S}_2+\mathcal{S}_3.
	\end{eqnarray*}
	For $\mathcal{S}_1$, due to the assumption of Lipschitz condition and boundedness of the first derivative, $\dot{Z}_t(\cdot)$ , 
	\begin{eqnarray}
	\max_{1\le j \le N}\sup_{\theta \cap B_j}\vert 
	\hat{\psi}_1(\theta)-\hat{\psi}_1(\theta_j)\vert&\le& C\dot{Z}_t(\bar{\theta})\|\theta-\theta_j\|=O_p(\tau_T) 
	\label{eq: s1}
	\end{eqnarray}
	For $\mathcal{S}_3$, the similar argument applies and hence
	\begin{equation}
	\text{P}\left(\max_{1\le j \le N}\sup_{\theta \cap B_j}\vert 
	{E}\hat{\psi}_1(\theta_j)-{E}\hat{\psi}_1(\theta)\vert \right)=O_p(\tau_T)
	\label{eq: s3}
	\end{equation}
	For $\mathcal{S}_2$, let $T^{-1}\sum_{t=1}^{T}D_t(\theta_j)= 
	\hat{\psi}_1(\theta_j)-{E}\hat{\psi}_1(\theta_j)$, i.e. $D_t(\theta_j)=Z_t(\theta)-{E}Z_t(\theta)$
	\begin{eqnarray*}
		\text{P}\left(\mathcal{S}_2>\tau_T\right)&=&\text{P}\left(\max_{1\le j \le N}\left| 
		\sum_{t=1}^{T}D_t(\theta_j)\right|>T\tau_T\right)\\
		&\le& \sum_{j=1}^{N}\text{P}\left(\left|\sum_{t=1}^{T}D_t(\theta_j)\right|>T\tau_T\right)\\
		&\le & N\sup_{\theta\in\theta}\text{P}\left(\left|\sum_{t=1}^{T}D_t(\theta)\right|>T\tau_T\right)\\
		&\le&c\tau_T^{-1}\sup_{\theta\in\theta}\text{P}\left(\left|\sum_{t=1}^{T}D_t(\theta)\right|>T\tau_T\right)
	\end{eqnarray*}
	Here, we apply the result of Theorem 2.1. in \citet{liebscher1996strong} (pg 71) on the strong convergence of sums of dependent strong mixing processes defined as follows with
	its mixing coefficients $\alpha (k)$ such that for $k>0$,
	\begin{equation*}
		\alpha (k)=\sup_{-T\leq t\leq T}\sup_{A\in \mathcal{F}_{-\infty }^{T,t},B\in
			\mathcal{F}_{T,t+k}^{\infty }}\left\vert P_{T}\left( A\cap B\right)
		-P_{T}\left( A\right) P_{T}\left( B\right) \right\vert
	\end{equation*}%
	where $\alpha (k)$ converges exponentially fast to zero as $k\rightarrow
	\infty$.\footnote{Note that $\phi$-mixing in Assumption \ref{ass:uniform} implies strong mixing and hence it is consistent with Assumption \ref{ass:uniform}.}
	For a stationary zero mean real valued process $M_t$ such that $\vert M_t \vert\le b_T$ with strong mixing coefficients $\alpha_m$, 
	\begin{equation}
	P\left(\left|\sum_{t=1}^{T} D_t\right| > \varepsilon   \right)
	\le 
	4\exp\left[-\frac{\varepsilon^2}{64\sigma^2_m \frac{T}{m}+\frac{8}{3}\varepsilon m b_T}\right]+4\frac{T}{m}\alpha_m
	\label{eq: exp_inequal}
	\end{equation}
	where $\sigma^2_m={E}\left(\sum_{i=1}^{m}D_i\right)^2$.
	We will use this exponential inequality to prove (\ref{eq:A.1}). Set $m=\gamma_T^{-1}\tau_T^{-1}$ and note that $m<T$ and $m<\varepsilon b/4$ where $\varepsilon=T\tau_T$ and $b=\tau_T$ for any $\theta$ and sufficiently large $T$.  
	Also, note that
	\[
	{E}\left(\sum_{t=1}^{m}D_t(\theta)\right)^2 \le Cm.
	\]
	From (\ref{eq: exp_inequal}),
	\begin{eqnarray*}
		\text{P}\left(\left|\sum_{t=1}^{T}D_t(\theta)\right|>T\tau_T\right)&\le&4\exp\left[-\frac{T^2\tau^2_T}{64CT +\frac{8}{3}T\tau_T  \gamma^{-1}_T \tau^{-1}_T \tau_T}\right]+4\frac{T}{m}\alpha_m\\
		&\le& 4\exp\left[-\frac{\ln T}{C}\right]+4\frac{T}{m}\alpha_m\\
		&\le& 4T^{-1/C}+o(1)
	\end{eqnarray*}
	where $C$ is a constant and the second term tends to zero due to the assumption on the strong mixing coefficient, which is assumed in the statement of the result. 
	Note that the last bound is independent of $\theta$, it is the uniform bound. Then,
	$$
	\text{P}\left(\mathcal{S}_2>\tau_T\right) \le O_p(\tau^{-1}_T T^{-1/C})
	$$
	Moreover, with sufficiently large strong mixing coefficient decay rate, $\beta$, 
	\begin{equation}
	\sum_{t=1}^{\infty}\text{P}\left(\mathcal{S}_2>\tau_T\right)\le \infty.
	\label{eq: s2}
	\end{equation}
	Then, the desired result follows from the Borel-Cantelli Lemma.
	Combining all the results, (\ref{eq: s1}), (\ref{eq: s3}) and (\ref{eq: s2}) proves (\ref{eq:A.1}). 
	The proof in relation to $\hat{\psi}_2(\theta)$ is similar and hence is omitted. This completes the proof.
\end{proof}

\section{Additional details for the locally stationary multiplicative SV model example}
\subsection{Local stationarity}
Recall the locally stationary SV (LS-SV) model: for all $u\in[0,1]$, $0<\xi(u)<\infty$, 
\begin{align*}
Y_{t,T} &= \sqrt{\xi(t/T)}\exp{(h_t/2)}\nu_{1,t} ,\nonumber \\
h_{t+1} &= \mu + \phi h_t + \nu_{2,t},
\end{align*}
where $\nu_{1,t}\sim_{iid} N(0,1)$. Let $\eta_{t}\sim_{iid} N(0,1)$, with $\nu_{1,t},\eta_{t}$ independent and define $$\nu_{2,t}= \gamma_\nu \nu_{1,t}+\sqrt{(1-\gamma_\nu^2)}\sigma \eta_{t}.$$ Partition the unknown parameter $\theta$ as $$\theta(u)=(\theta_1(u),\theta_2')'=(\xi(u),\mu,\phi,\sigma,\rho)'.$$ 
Using these definitions, the LS-SV model can be placed in the general form of the structural model in equation \eqref{struct1}:
\begin{flalign*}
Y_{t,T}&=r(\epsilon_{t,T};\theta)=\sqrt{\theta_1(t/T)}\epsilon_{t,T},\;\\\epsilon_{t,T}&=\varphi(\nu_t,\theta_2)=\exp(h_t(\theta_2,\nu_t)/2)\nu_{1,t}
\end{flalign*}and where we have $h_t(\theta,\nu_t)=\mu+\phi h_t(\theta,\nu_{t-1})+\nu_{2,t}$.

Under a weak assumption regarding the growth of $\xi(\cdot)$, and under compactness for the remaining components of $\theta$, this model is locally stationary.
\begin{corollary}For all $u_1,u_2\in[0,1]$, assume that $|\sqrt{\xi(u_1)}-\sqrt{\xi(u_2)}|_{}\leq K|u_1-u_2|$, with $K$ finite and independent of $u_1,u_2$. Assume that $\theta_2\in\Theta_2$, and $\Theta_2$ compact, with $|\rho|\leq 1-\epsilon$ and $|\phi|\leq 1-\epsilon$ for some $\epsilon>0$. Then, there exists a measurable random variable $C_t$ such that, for some $\eta>0$, $\sup_{t\leq T}E(|C_t|^\eta)<\infty$, and 
$$
\text{P}\left(\max_{1\leq t\leq T}|Y_{t,T}-y_{t,t/T}|\leq C_{T}/T\right)=1.
$$
where $C_T=\sup_{t\le T} C_t$.
\end{corollary}
\begin{proof}
	Note that 
	$$
	Y_{t,T}-y_{u,t}:=\left\{\zeta(t/T)-\zeta(u)\right\}\exp(h_t/2)v_{1,t}, \text{ where }\zeta(u):=\sqrt{\xi(u)}
	$$Now, for any $u\in[0,1]$, 
	$$
	Y_{t,T}-y_{u,t}=\left\{\zeta(t/T)-\zeta(u)\right\}\exp(h_t/2)v_{1,t}=\left\{\zeta(t/T)-\zeta(u)\right\}\varphi(\nu_t,\theta_2)
	$$
	By the Lipschitz assumption on $\zeta(\cdot)$, for any $t/T,u$, for some finite $K$, and any $u\in[0,1]$, 
	$$
	\max_{t\leq T}|\zeta(t/T)-\zeta(u)|\leq K \max_{t\leq T}|t/T-u|\leq K.
	$$
	As a consequence, for any $t\leq T$, $$|Y_{t,T}-y_{t,t/T}|\leq K\left\{\exp(h_t/2)|v_{1,t}|\right\}.$$ 
	
	Define $A_t=K\left\{\exp(h_t/2)\nu_{1,t}\right\}$. We now demonstrate that $E[A_t^2]<\infty$ for all $t\ge1$. By construction 
	$$
	\nu_t\sim\mathcal{N}\left(\begin{bmatrix}
	0\\0
	\end{bmatrix},\begin{bmatrix}
	1&\gamma_\nu\sigma\\\gamma_\nu\sigma&\sigma^2
	\end{bmatrix}\right).
	$$By assumption $\nu_t,\nu_{t-1}$, for any $t\ge1$, and it then follows that
	$$
	\begin{bmatrix}\nu_{1,t}\\h_{t}
	\end{bmatrix}\sim\mathcal{N}\left(\begin{bmatrix}
	0\\\mu/(1-\phi)
	\end{bmatrix},\begin{bmatrix}
	1&\gamma_\nu\sigma_h\\\gamma_\nu\sigma_h&\sigma^2_h
	\end{bmatrix}\right),
	$$where $\sigma^2_h=\sigma^2/(1-\phi^2)$, and from which we conclude that 
	\begin{equation}\label{eq:cond1}
	h_t|\nu_{1,t},h_{t-1}\sim\mathcal{N}\left(\mu+\phi h_{t-1}-\gamma_{\nu}\sigma_h \nu_{1,t},(1-\gamma_{\nu}^2)\sigma^2_h\right)
	\end{equation}
	Now, $A_t^2=\nu_{1,t}^2\exp(h_t)$ and let us calculate $E[A_t^2]$,
	\begin{flalign*}
	E\left[\nu_{1,t}\exp\left(h_{t}\right)\right]&=E\left[\nu_{1,t}E\left[\exp(h_t)|\nu_{1,t} \right]\right].
	\end{flalign*}
	Let $Z_t\sim_{iid} \mathcal{N}(0,1)$, with $Z_t\perp\nu_{1,t}$, for
	$$
	\tilde{Z}_t=\mu+\phi h_{t-1}+Z_t\sqrt{\sigma^2_h}
	$$ the moments of the random variable in \eqref{eq:cond1} are equivalent to those of $\tilde{Z}_t-\gamma_{\nu}\sigma_h \nu_{1,t}$. Moreover, since $\nu_{t}\perp\nu_{t-1}$, it follows that $\tilde{Z}_t\perp\nu_{1,t}$. Using this independence and the log-normality of $\exp(\tilde{Z}_t)$, deduce that 
	\begin{flalign*}
	E\left[\exp(h_t)|\nu_{1,t} \right]&=\exp(-\gamma_{\nu}\sigma_h\nu_{1,t})E[\exp(\tilde{Z}_t)]\\&=\exp(-\gamma_{\nu}\sigma_h\nu_{1,t})\exp\left(\frac{\mu}{1-\phi}+\frac{(1-\gamma_{\nu}^2)\sigma_h^2}{2}\right),
	\end{flalign*}
	Apply this formula to calculate
	\begin{flalign*}
	E\left[\nu^2_{1,t}\exp\left(h_{t}\right)\right]&=\exp\left(\frac{\mu}{1-\phi}+\frac{(1-\gamma_{\nu}^2)\sigma_h^2}{2}\right)E\left[\nu_{1,t}\exp(-\gamma_{\nu}\sigma_h\nu_{1,t})\right]\\&=\frac{\exp\left(\frac{\mu}{1-\phi}+\frac{(1-\gamma_{\nu}^2)\sigma_h^2}{2}\right)}{\sqrt{2\pi}}\int_{\mathbb{R}}\nu^2_1\exp\left(-\nu_1^2/2-\gamma_{\nu}\sigma_h\nu_{1}\right)\text{d}\nu_1\\&=M(1+\gamma_{\nu}^2\sigma^2_h){\exp\left(\frac{\mu}{1-\phi}+\frac{(1-\gamma_{\nu}^2)\sigma_h^2+2\gamma_{\nu}^2\sigma_h^2}{2}\right)},
	\end{flalign*}where $M>0$ is finite and does not depend on any parameters. 
	
	Taking $C_t=A_t^2$, for all $t\ge1$, the definition is then satisfied with $\eta=1$. 
	
\end{proof}

\subsection{Additional empirical results}
\subsubsection{Alternative conditional mean specification}
Recall the locally stationary SV (LS-SV) model considered in Section 5.2: for $r_{t,j}$ denoting excess returns on the $j$-th portfolio under analysis, $r_{t,j}$ evolves according to
\begin{flalign*}
r_{t,j}&=m_{t,j}+\epsilon_{t,j},\\\epsilon_{t,j}&=\sqrt{\xi_j(t/T)}\exp(h_{t,j}/2)\nu_{t,1j},
\end{flalign*}where the regression function $m_{t,j}$ is given by 
$$
m_{t,j}=\alpha+\beta_1 r_{t,m}+\beta_{2}\text{SMB}_{t}+\beta_{3}\text{HML}_{t},
$$
and where $r_{t,m}$ denotes excess returns on the market factor, $\text{SMB}_t$ is the size factor, and $\text{HML}_t$ is the value factor. The short-run volatility component $h_{t,j}$ is modeled as a AR(1) stochastic volatility process with leverage effects:
\begin{flalign*}
h_{t,j}&=\phi_j h_{t-1,j}+\sigma_{j}\nu_{t,2j},\;\;\;\text{corr}(\nu_{t,1j},\nu_{t,2j})=\gamma_{\nu,j},
\end{flalign*}
where we fix the mean of the short-run SV component to zero to ensure the scale of $\xi(\cdot)$ is identified. 

In this section, we revisit the empirical analysis in Section 5.2 under an alternative specification for the conditional mean. In this section, we consider an LS-SV model where the conditional mean function is a time-varying regression of the following form:$$m_{t,j}=\alpha_j(t/T)+\beta_j(t/T)r_{t,m},$$where $r_{t,m}$ denotes the market factor. Such a model is akin to a time-varying $\alpha,\;\beta$ model with a LS volatility components.

Estimation and inference for this modified LS-SV model is carried out using a similar two-step procedure to that considers in the linear regression version of the model. First, we estimate the time-varying regression parameters to obtain $\hat{\alpha}(t/T),\;\hat{\beta}(t/T)$; in the second step, the residuals $$y_{t,j}=\left(r_{t,j}-\hat{\alpha}(t/T)-\hat{\beta}(t/T)r_{t,m} \right)$$ are used within the L-II algorithm for the LS-SV model, along with a LS-GJR-GARCH auxiliary model (see Algorithm \ref{alg:msv} for specific implementation details). 

We estimate this modified version of the LS-SV model across all 25 portfolios used in Section 5.2. Pointwise confidence intervals are formed for each of the estimated functions using the local block-bootstrap approach discussed in Section 5.2, and where all elements of the LBB discussed in Section 5.2 are replicated for this analysis. 

The estimated results for $\beta_j(\cdot)$ are given graphically in Figures \ref{beta1_D2_fig1}-\ref{beta1_D2_fig5} for each portfolio, while the results for $\alpha_j(\cdot)$ are given in Figures \ref{alpha1_D2_fig1}-\ref{alpha1_D2_fig5}. To ensure the results can be easily interpreted, we only present results for $u\in[0.05,0.95]$. For values of $u$ not in this region, it is well-known that local kernel methods can display significant boundary bias and can be unreliable. Analyzing the estimates we see that across the majority of the sample the estimated functions for both $\alpha_j$ and $\beta_j$ are nearly constant, and across virtually all of the portfolios. We note, however, that certain of the estimates do appear to display some nonlinear behavior at the beginning and end of the sample. We believe that this is due to the boundary bias exhibited by the local kernel estimation method, and therefore is not genuine nonlinearity. All told, these results demonstrate that considering fixed values of $\alpha(\cdot),\beta(\cdot)$, as was done in Section 5.2 of the main paper, is likely a tenable empirical specification.

The resulting estimated functions for $\xi_j(\cdot)$ are given in Figures \ref{tau1_D2_fig1}-\ref{tau1_D2_fig5}. Largely speaking, the results for $\xi_j$ follow a very similar patter to those obtained under the linear regression specification for the conditional mean of returns. Note that there is no reason to expect that the two sets of results should be equivalent, since the results are based on two entirely different sets of residuals. 

\subsubsection{Estimated Values of $\beta(t/T)$}

\begin{figure}[H]
	\centering 
		\includegraphics[height=0.4\textheight, width=0.8\linewidth]{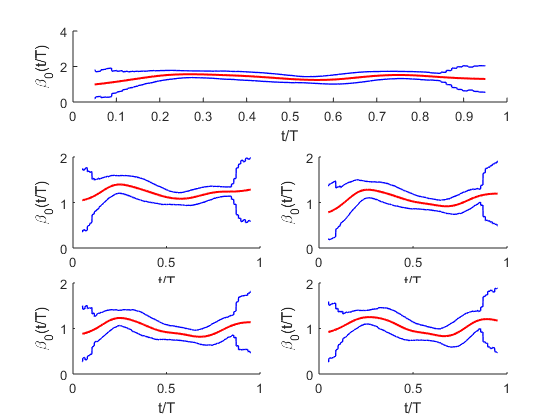}
	\caption{Estimated time-varying $\beta$ for the first five Fama-French portfolios, which is the first quintile of size, intersected with the five quintiles of book-to-market. The top most figure in the panel represents the results for the portfolio formed from the intersection of the first quintile of size, and the first quintile of book-to-market. Running from right to left, the results then correspond to the first quintile of size and the second through fifth quintiles of book-to-market. {The upper and lower confidence intervals are constructed pointwise, at each time point $u=t/T$, using the local block bootstrap.} For presentation of the results, the time span represented on the $x$-axis has been placed on the unit interval. The following correspondence can be used to aid interpretation of the results. The first time point in the sample corresponds to January 1952, the time point $u=0.5$ corresponds to June 1985, and the time point $u=.99$ corresponds to December 2018.}
	\label{beta1_D2_fig1}
\end{figure}

\begin{figure}[H]
	\centering 
			\includegraphics[height=0.4\textheight, width=0.8\linewidth]{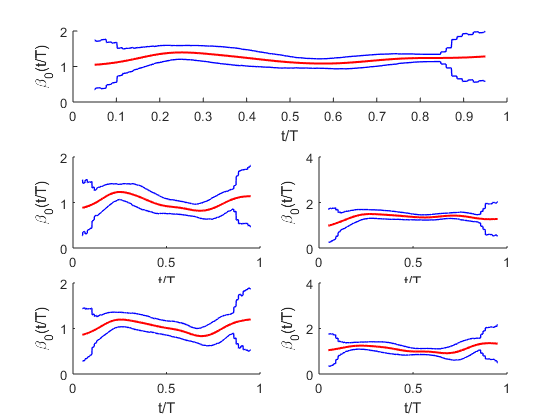}
	\caption{Second five Fama-French portfolios: Second quintile of size, intersected with the five quintiles of book-to-market. The figure has the same interpretation as Figure \ref{FF_fig1} but the results correspond to the second quintile of size.}
	\label{beta1_D2_fig2}
\end{figure}
\begin{figure}[H]
	\centering 
			\includegraphics[height=0.4\textheight, width=0.8\linewidth]{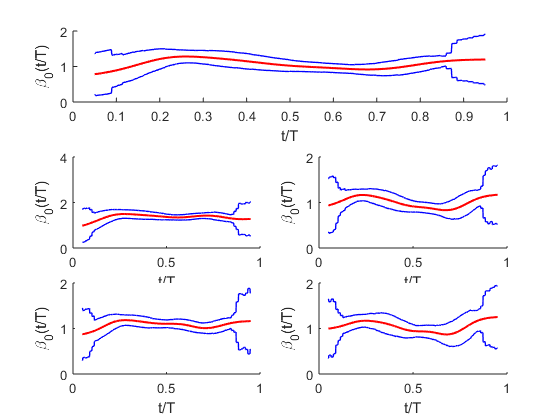}
	\caption{Third five Fama-French portfolios: Third quintile of size, intersected with the five quintiles of book-to-market. The figure has the same interpretation as Figure \ref{FF_fig1} but the results correspond to the third quintile of size.}
	\label{beta1_D2_fig3}
\end{figure}
\begin{figure}[H]
	\centering 
			\includegraphics[height=0.4\textheight, width=0.8\linewidth]{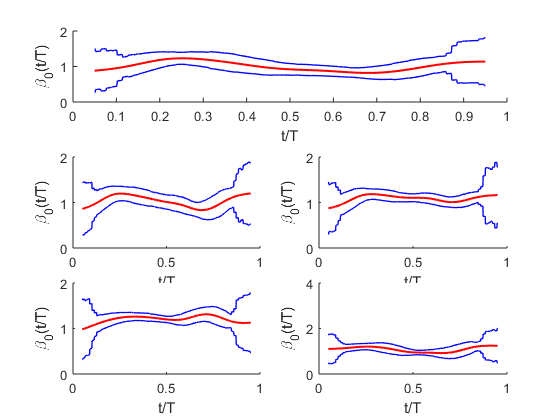}
	\caption{Fourth five Fama-French portfolios: Fourth quintile of size, intersected with the five quintiles of book-to-market. The figure has the same interpretation as Figure \ref{FF_fig1} but the results correspond to the fourth quintile of size.}
	\label{beta1_D2_fig4}
\end{figure}
\begin{figure}[H]
	\centering 
			\includegraphics[height=0.4\textheight, width=0.8\linewidth]{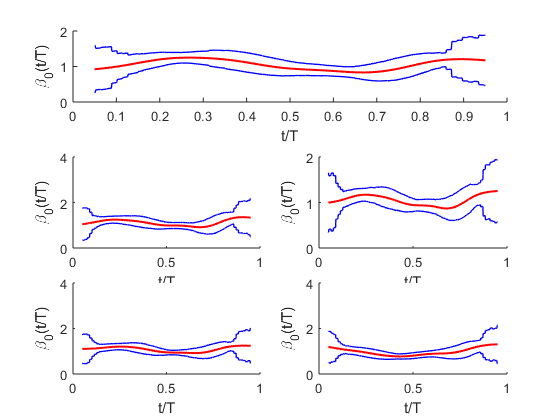}
	\caption{Last five Fama-French portfolios: Fifth quintile of size, intersected with the five quintiles of book-to-market. The figure has the same interpretation as Figure \ref{FF_fig1} but the results correspond to the fifth quintile of size.}
	\label{beta1_D2_fig5}
\end{figure}

\newpage

\subsubsection{Figures: $\alpha(t/T)$}
\begin{figure}[H]
	\centering 
		\includegraphics[height=0.4\textheight, width=0.8\linewidth]{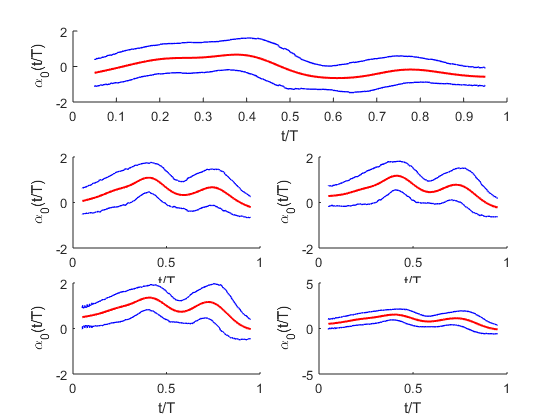}
	\caption{Estimated time-varying $\alpha$ for the first five Fama-French portfolios:  First quintile of size, intersected with the five quintiles of book-to-market. The top most figure in the panel represents the results for the portfolio formed from the intersection of the first quintile of size, and the first quintile of book-to-market. Running from right to left, the results then correspond to the first quintile of size and the second through fifth quintiles of book-to-market.  {The upper and lower confidence intervals are constructed pointwise, at each time point $u=t/T$, using the local block bootstrap.} For presentation of the results, the time span represented on the $x$-axis has been placed on the unit interval. The following correspondence can be used to aid interpretation of the results. The first time point in the sample corresponds to January 1952, the time point $u=0.5$ corresponds to June 1985, and the time point $u=.99$ corresponds to December 2018.}
	\label{alpha1_D2_fig1}
\end{figure}

\begin{figure}[H]
	\centering 
			\includegraphics[height=0.4\textheight, width=0.8\linewidth]{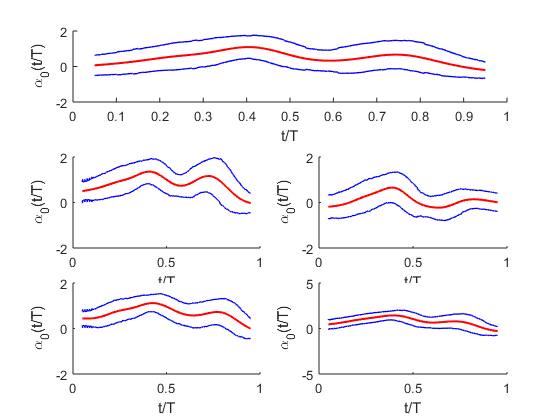}
	\caption{Second five Fama-French portfolios: Second quintile of size, intersected with the five quintiles of book-to-market. The figure has the same interpretation as Figure \ref{FF_fig1} but the results correspond to the second quintile of size.}
	\label{alpha1_D2_fig2}
\end{figure}
\begin{figure}[H]
	\centering 
			\includegraphics[height=0.4\textheight, width=0.8\linewidth]{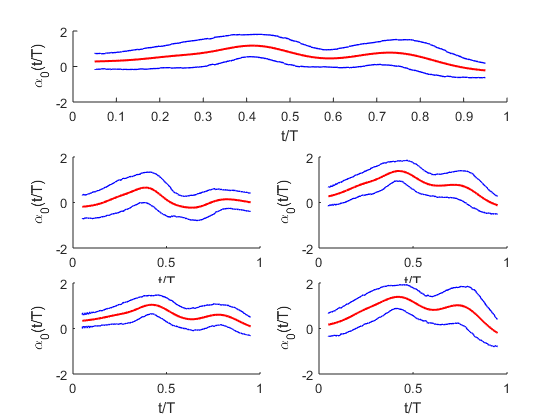}
	\caption{Third five Fama-French portfolios: Third quintile of size, intersected with the five quintiles of book-to-market. The figure has the same interpretation as Figure \ref{FF_fig1} but the results correspond to the third quintile of size.}
	\label{alpha1_D2_fig3}
\end{figure}
\begin{figure}[H]
	\centering 
			\includegraphics[height=0.4\textheight, width=0.8\linewidth]{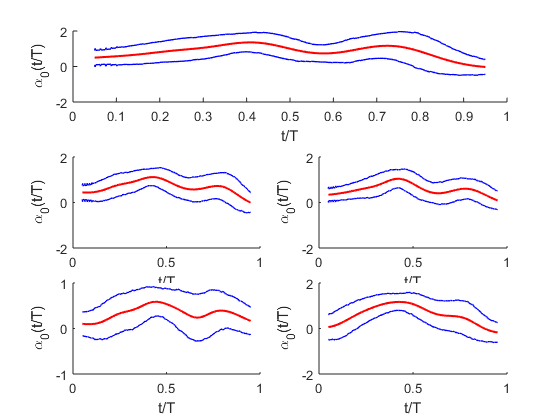}
	\caption{Fourth five Fama-French portfolios: Fourth quintile of size, intersected with the five quintiles of book-to-market. The figure has the same interpretation as Figure \ref{FF_fig1} but the results correspond to the fourth quintile of size.}
	\label{alpha1_D2_fig4}
\end{figure}
\begin{figure}[H]
	\centering 
			\includegraphics[height=0.4\textheight, width=0.8\linewidth]{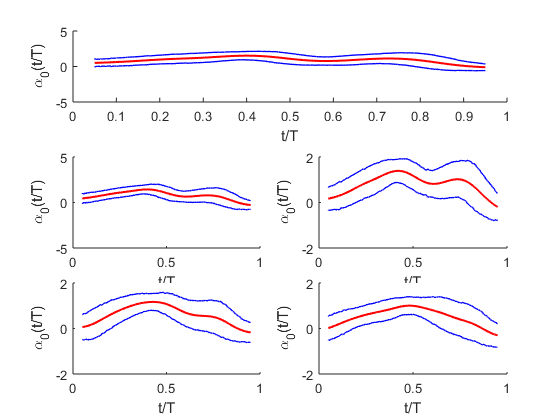}
	\caption{Last five Fama-French portfolios: Fifth quintile of size, intersected with the five quintiles of book-to-market. The figure has the same interpretation as Figure \ref{FF_fig1} but the results correspond to the fifth quintile of size.}
	\label{alpha1_D2_fig5}
\end{figure}

\newpage
\subsubsection{Figures: $\xi(t/T)$}

\begin{figure}[H]
	\centering 
		\includegraphics[height=0.4\textheight, width=0.8\linewidth]{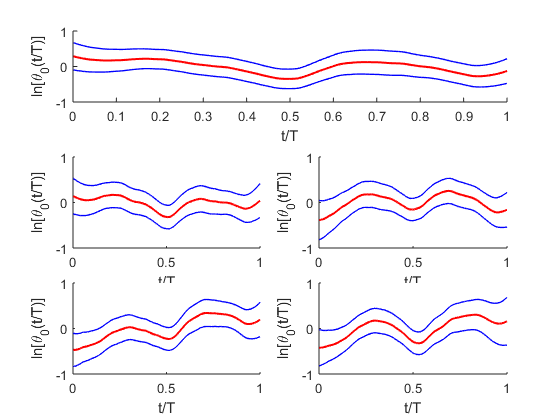}
	\caption{Estimated time-varying $\xi$ for the first five Fama-French portfolios:  First quintile of size, intersected with the five quintiles of book-to-market. The top most figure in the panel represents the results for the portfolio formed from the intersection of the first quintile of size, and the first quintile of book-to-market. Running from right to left, the results then correspond to the first quintile of size and the second through fifth quintiles of book-to-market.  {The upper and lower confidence intervals are constructed pointwise, at each time point $u=t/T$, using the local block bootstrap.} For presentation of the results, the time span represented on the $x$-axis has been placed on the unit interval. The following correspondence can be used to aid interpretation of the results. The first time point in the sample corresponds to January 1952, the time point $u=0.5$ corresponds to June 1985, and the time point $u=.99$ corresponds to December 2018.}
	\label{tau1_D2_fig1}
\end{figure}

\begin{figure}[H]
	\centering 
		\includegraphics[height=0.4\textheight, width=0.8\linewidth]{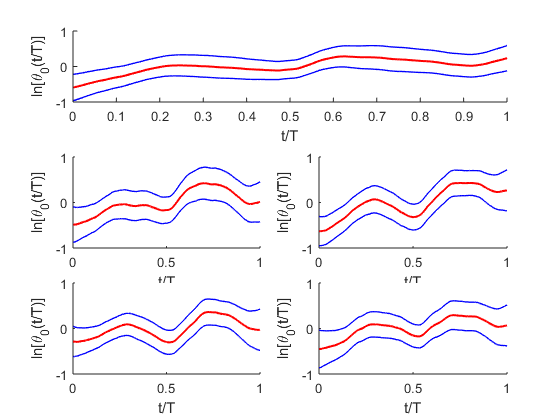}
	\caption{Second five Fama-French portfolios: Second quintile of size, intersected with the five quintiles of book-to-market. The figure has the same interpretation as Figure \ref{FF_fig1} but the results correspond to the second quintile of size.}
	\label{tau1_D2_fig2}
\end{figure}
\begin{figure}[H]
	\centering 
			\includegraphics[height=0.4\textheight, width=0.8\linewidth]{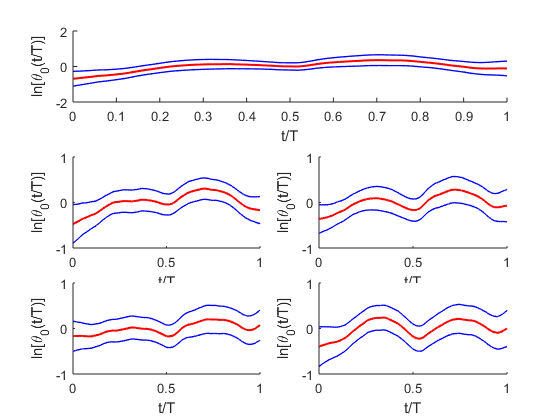}
	\caption{Third five Fama-French portfolios: Third quintile of size, intersected with the five quintiles of book-to-market. The figure has the same interpretation as Figure \ref{FF_fig1} but the results correspond to the third quintile of size.}
	\label{tau1_D2_fig3}
\end{figure}
\begin{figure}[H]
	\centering 
			\includegraphics[height=0.4\textheight, width=0.8\linewidth]{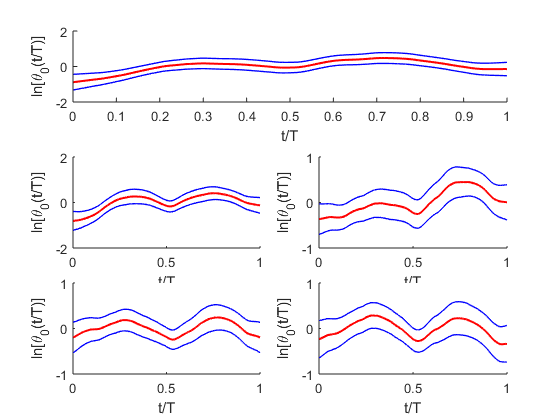}
	\caption{Fourth five Fama-French portfolios: Fourth quintile of size, intersected with the five quintiles of book-to-market. The figure has the same interpretation as Figure \ref{FF_fig1} but the results correspond to the fourth quintile of size.}
	\label{tau1_D2_fig4}
\end{figure}
\begin{figure}[H]
	\centering 
			\includegraphics[height=0.4\textheight, width=0.8\linewidth]{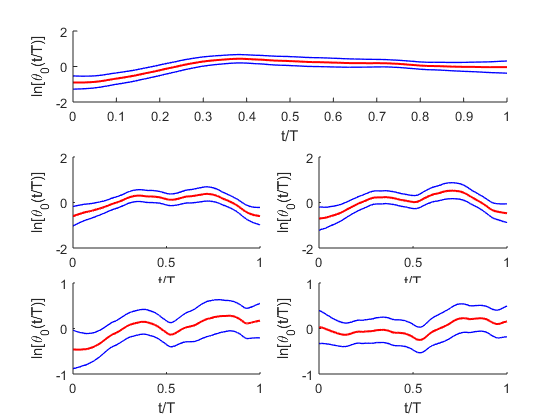}
	\caption{Last five Fama-French portfolios: Fifth quintile of size, intersected with the five quintiles of book-to-market. The figure has the same interpretation as Figure \ref{FF_fig1} but the results correspond to the fifth quintile of size.}
	\label{tau1_D2_fig5}
\end{figure}

\newpage
\subsubsection{ARCH Testing Results}
\begin{table}[H]
	\centering
	\caption{ARCH test statistics for the 25 Fama-French portfolios, calculated using centered returns and five lags for the auxiliary regression. The corresponding $\chi^2_5(.01)$ critical value is 15.08. For each portfolio, we can reject the null at the 1\% significance level. Furthermore, a similar conclusion remains at the .1\% level for all but three of the 25 Fama-French portfolios. Size-j, and BM-j, $j=1,...,5$, refer to the quintiles of size and book-to-market, respectively  }
	\begin{tabular}{rrrrrrr}
		\hline
		&       & Size-1     & Size-2     & Size-3     & Size-4     & Size-5 \\\hline
		& {BM}-1     & 25.97 & 44.31 & 37.22 & 53.03 & 48.77 \\
		& {BM}-2     & 24.49 & 40.55 & 57.38 & 60.51 & 45.35 \\
		& {BM}-3     & 41.31 & 22.39 & 39.51 & 50.72 & 35.74 \\
		& {BM}-4     & 17.54 & 36.94 & 28.51 & 35.89 & 44.28 \\
		& {BM}-5     & 18.72 & 24.92 & 41.95 & 48.33 & 16.09 \\\hline\hline
	\end{tabular}%
	\label{tab:arch}%
\end{table}%

\end{document}